\documentclass{article}[12pt]
\usepackage{amsmath,amssymb,amsthm,color,url,framed}
\usepackage[scr=boondoxo, scrscaled=1]{mathalfa}
\usepackage{graphicx,enumerate}
\usepackage[margin=2cm]{geometry}

\newtheorem{theorem}{Theorem}
\newtheorem{remark}[theorem]{Remark}

\newtheorem{definition}[theorem]{Definition}
\newtheorem{corollary}[theorem]{Corollary}

\newcommand{\pp}[2]{\frac{\partial #1}{\partial #2}}

\def\contract{\makebox[1.2em][c]{\mbox{\rule{.6em}
{.01truein}\rule{.01truein}{.6em}}}}

\numberwithin{equation}{section}

\title{Momentum maps and stochastic Clebsch action principles}
\author{Ana Bela Cruzeiro$^{1}$, Darryl D. Holm$^{2}$, Tudor S. Ratiu$^{3}$
\\ \footnotesize
(1) GFMUL and Mathematics Department,  Instituto Superior T\'ecnico, 
Lisboa, Portugal
\\ \footnotesize
(2) Mathematics Department, Imperial College, London, United Kingdom
\\ \footnotesize
(3) School of Mathematics, Jiao Tong University, Shanghai, China and
\\ \footnotesize
Section de Math\'ematiques, Universit\'e de Gen\`eve, Switzerland} 
\date{Version  16 October  2017 }

\begin{document}
\maketitle
\makeatother

\begin{abstract}
We derive stochastic differential equations whose solutions follow the flow of a stochastic nonlinear Lie algebra operation on a configuration manifold. For this purpose, we develop a stochastic Clebsch action principle, 
in which the noise couples to the phase space variables through a momentum map. This special coupling simplifies the structure of the resulting stochastic Hamilton equations for the momentum map. In particular, these stochastic Hamilton equations collectivize for Hamiltonians that depend only on the momentum map variable. The Stratonovich equations are derived from the Clebsch variational principle and then converted into It\^o form. In comparing the Stratonovich and It\^o forms of the stochastic dynamical equations governing the components of the momentum map, we find that the It\^o contraction term turns out to be a double Poisson bracket. Finally, we present the stochastic Hamiltonian formulation of the collectivized momentum map dynamics and derive the corresponding Kolmogorov forward and backward equations.  
\end{abstract}

\tableofcontents

\section{Background and motivation}\label{intro}
\subsection{Poincar\'e 1901}
In 1901 Poincar\'e noticed that when a Lie group $G$ acts transitively 
on the \textit{configuration manifold} $Q$ of a mechanical
system, then an opportunity arises, ``to cast the equations of mechanics 
into a new form which could be interesting to know'' \cite{Po1901}. The 
new form of the Euler-Lagrange equations of mechanics for a 
given Lagrangian $L(q,\dot{q})$ in Hamilton's principle $0=\delta \int_a^b 
L(q,\dot{q}) dt$ defined on the tangent bundle $TQ$ (\textit{velocity phase space}) of the manifold $Q$ (\textit{state space}) emerges when the motion is lifted to a set of dynamical equations for a curve $g(t)\in G$, parameterized by time, $t$, by writing the \textit{motion} as  
$q(t)=g^{-1}(t)q_0$, with $g(0)=e$, the identity element of the Lie group, $G$. 
Let $\mathfrak{g}$ denote the Lie algebra of $G$.
The Lie algebra action of the element $u:=g^{-1}\dot{g} \in \mathfrak{g}$, 
on the manifold $Q$ is denoted by concatenation, namely, $uq$; this is a 
vector field $Q \ni q \mapsto uq \in TQ$ on $Q$. Thus, if 
$q(t)=g^{-1}(t)q_0$, we have $\dot{q}(t)=-uq(t)$.

The action integral in Hamilton's principle transforms under $G$ 
in \cite{Po1901} as,
\begin{equation}
\int_a^b L(q,\dot{q})dt =: \int_a^b \tilde{L} (g,\dot{g};q_0)dt 
= \int_a^b \tilde{L} (e,g^{-1}\dot{g},g^{-1}q_0)dt =: 
\int_a^b \ell (u,g^{-1}q_0)  
\,,
\label{Lag-red}
\end{equation}
where $\langle\cdot , \cdot \rangle_\mathfrak{g}: \mathfrak{g}^\ast \times 
\mathfrak{g}\rightarrow \mathbb{R}$ denotes the non-degenerate pairing between the Lie algebra $\mathfrak{g}$ and its dual $\mathfrak{g}^\ast$.

Upon taking variations in Hamilton's principle, Poincar\'e cast the 
Euler-Lagrange equations for vanishing endpoint conditions  into his 
``new form''. To arrive at it, we take a
deformation $g_\varepsilon(t)$ of the curve $g_0(t): = g(t)$ for
$\varepsilon$ in a small interval centred at $0$, keeping
the endpoints fixed, i.e., $g_\varepsilon(a) = g(a)$, 
$g_\varepsilon(b) = g(b)$ for all $\varepsilon$, denote by $\delta g(t): =
\left.\frac{d}{d\varepsilon}\right|_{\varepsilon=0}g_\varepsilon(t) \in
T_{g(t)}G$, and note that $\delta g(a)=0=\delta g(b)$. Defining
$v(t): = g(t)^{-1} \delta g(t) \in \mathfrak{g}$ and $\delta u(t): =
\left.\frac{d}{d\varepsilon}\right|_{\varepsilon=0} g_\varepsilon(t)^{-1}\dot{g}_\varepsilon(t) \in \mathfrak{g}$, we
deduce the identity $\delta u(t) = \dot{v}(t) + 
\operatorname{ad}_{u(t)} v (t)$, where $\operatorname{ad}_{x}y:=
[x, y]$ for every $x, y \in\mathfrak{g}$ and we denote by 
$\operatorname{ad}_x ^\ast: \mathfrak{g}^* \rightarrow 
\mathfrak{g}^\ast$ the dual of the linear map $\operatorname{ad}_x:
\mathfrak{g} \rightarrow \mathfrak{g}$ for every $x \in \mathfrak{g}$.
A direct computation, using \eqref{Lag-red}, yields  
\begin{equation}
\int_a^b \left\langle \frac{d}{dt}\frac{\partial L}{\partial \dot{q}} 
- \frac{\partial L}{\partial q}\,,\, \delta q \right\rangle_{Q}\,dt
= 0 = \int_a^b \left\langle \frac{d}{dt}\frac{\partial \ell}{\partial u} 
- {\rm ad}^*_u \frac{\partial \ell}{\partial u} - 
\mathbf{J}\left(\frac{\partial \ell}{\partial q}\right),\, v \right\rangle_\mathfrak{g}dt \,,
\label{Lag-var}
\end{equation}
where $\langle\cdot , \cdot \rangle_{Q}: T^\ast Q\times TQ \to \mathbb{R}$ 
denotes the natural duality pairing, taken fiberwise, between the tangent 
bundle $TQ$ and its dual $T^*Q$, the \textit{phase space} of the mechanical system. The map $\mathbf{J}: T^*Q \rightarrow \mathfrak{g}^\ast$ has
the expression $\left\langle\mathbf{J}(p_q), w \right\rangle_\mathfrak{g}= 
\left\langle p_q, wq \right\rangle_Q$, for all $w \in \mathfrak{g}$ and is the \textit{momentum
map} of the cotangent lifted $G$-action on $T^*Q$. Thus, \eqref{Lag-var}
yields the classical Euler-Lagrange equations, if one uses the left
hand side of the identity, and it yields Poincar\'e's ``new form'' of the equations 
of motion, if one uses the right hand side of the identity, which is
\begin{equation}
\label{classical_EP}
\frac{d}{dt}\frac{\partial \ell}{\partial u} 
= {\rm ad}^*_u \frac{\partial \ell}{\partial u} + 
\mathbf{J}\left(\frac{\partial \ell}{\partial q}\right).
\end{equation}
For more details, many applications, and reviews of the overwhelming importance
of the momentum map in mechanics, see \cite{Ho2009,MaRa1999}.

In Poincar\'e's illustrative example, $G=SO(3)$ was the Lie group of rotations in three dimensions; the manifold $Q=\mathbb{R}^3$ was three dimensional Euclidean space; $\mathfrak{g}=\mathfrak{so}(3)\simeq \mathbb{R}^3$ and $\mathfrak{g}^\ast=\mathfrak{so}(3)^*\simeq \mathbb{R}^3$ were isomorphic to $\mathbb{R}^3$;  the pairings $\langle\cdot , \cdot \rangle_{Q}$ and $\langle\cdot , \cdot \rangle_{\mathfrak{g}}$ were both the Euclidean scalar product; the operations ${\rm ad}$, ${\rm ad}^*$ and the momentum 
map $\mathbf{J}$ were all (plus, or minus) the vector cross product 
in $\mathbb{R}^3$. Poincar\'e's new form of the equations of mechanics in that case reduced to Euler's equations for a heavy top. This 1901 result 
of Poincar\'e, together with Lie's discovery of the
natural Poisson bracket on the dual of a Lie algebra, are the two key
results, one in the Lagrangian, the other in the Hamiltonian
formulation, of what has developed into geometric mechanics. 
Poincar\'e's paper \cite{Po1901} was carefully reviewed recently from a modern perspective in \cite{Ma2013}. For textbook discussions of geometric mechanics, see, e.g., \cite{Ho2009,MaRa1999}.

The aim of the present work is to continue the theme of these earlier developments by revisiting Poincar\'e's starting point \cite{Po1901} and augmenting the Lagrangian \eqref{Lag-red} in Hamilton's principle to introduce a Lagrange multiplier $m$ into \eqref{Lag-red}, in preparation for introducing stochasticity later. That is, instead of \eqref{Lag-red}, we work with the constrained action,
\[
\int_a^b \left(\ell (u,g^{-1}q_0) + \left\langle m\,,\, g^{-1}\dot{g} - 
u \right\rangle_\mathfrak{g} \right) dt\,,
\]
in order to enforce the \textit{reconstruction relation} $g^{-1}\dot{g} = u$ for the curve $g(t)\in G$. The reconstruction relation, in turn, generates the motion $q(t)=g^{-1}(t)q_0$ along a solution curve in $Q$. This constrained form of the action facilitates the introduction of stochasticity into Poincar\'e's original framework. That is, in following Poincar\'e's lead in the deterministic case, we seek to lay the framework for \textit{stochastic} geometric mechanics.  In deterministic geometric mechanics, as we have just seen, the time-dependent dynamics is modeled by the action of a transformation group. The geometric mechanics approach lifts the dynamics on the state space to a curve in the transformation group. 

Our \textit{aim} in this paper is to generalize the time-dependent curve in the transformation group to a stochastic process, and then use Hamilton's principle to determine the stochastic dynamics of the momentum map taking values in the dual of the Lie algebra of the transformation group. 

Our \textit{approach} to achieve the transition from deterministic to stochastic geometric mechanics is to concentrate on the Lie algebra action 
$\dot{q}(t)=-uq(t)\in TQ$ of the vector field $u=g^{-1}\dot{g}\in \mathfrak{g}$, which produces the solution paths $q(t)\in Q$. The solution paths $q(t)=g^{-1}(t)q_0\in Q$ will become stochastic, if $g(t)$ is made stochastic by replacing the deterministic reconstruction equation $g^{-1}\dot{g} = u$ mentioned above by introducing the following reconstruction relation from a stochastic vector field,
\begin{equation}
g^{-1}d_t g = u\,dt + \sum_i \xi_i \circ dW^i(t)
\,,
\label{stoch-reconstruct}
\end{equation}
where subscripted $d_t$ represents stochastic time evolution, the vector fields $\xi_i$ for $i=1,2,\dots,N,$ are prescribed, and $\circ\, dW^i(t)$ denotes the \textit{Stratonovich} differential with independent Brownian motions $dW^i(t)$. 
The idea, then, is to regard the stochastic solution paths $q(t)=g(t)^{-1}q_0\in Q$ as observable data, from which we obtain the correlation eigenvectors $\xi_i$ by some form of bespoke data assimilation, and substitute them into the constrained action 
\[
\int_a^b \Big(\ell (u,g^{-1}q_0) + \Big\langle m\,,\, g^{-1}d_t g - 
u\,dt - \sum_i \xi_i \circ dW^i(t) \Big\rangle_\mathfrak{g} \,\Big) dt\,,
\]
then take variations to derive the corresponding equations of motion for $u\in \mathfrak{g}$ by applying Poincar\'e's approach to the resulting stochastically constrained Hamilton's principle. In this way, we obtain a variational approach for deriving \textit{data-driven} models in the framework of stochastic geometric mechanics.

\subsection{Data-driven modeling of uncertainty} 

As opposed to theory-driven models such as Newtonian force laws and thermodynamic processes for the subgrid-scale dynamics, here we will use stochastic geometric mechanics as an opportunity to consider a new type of data-driven modeling. In data-driven modeling, one seeks to model properties of a subsystem of a given dynamical system which, for example, may be observable at length or time scales which are below the resolution of available initial and boundary conditions, or of numerical simulations of the dynamical system based on the assumed exact equations. 

The most familiar example of data-driven modeling occurs in numerical weather forecasting, where various numerically unresolvable, but observable, subgrid-scale processes are expected to have profound effects on the variability of the weather; so they must be parameterized at the resolved scales of the numerical simulations. Of course, the accuracy of a given parameterization model often remains uncertain. In fact, even the possibility of modeling subgrid-scale properties in terms of resolved-scale quantities simulations may sometimes be questionable. However, if some information about the \textit{statistics} of the small-scale excitations is known, such as the spatial correlations of its observed transport properties at the resolved scales, one may arguably consider modeling the effects of the small scale dynamics on the resolved scales by a stochastic transport process whose spatial correlations match the observations, at the resolved scales. 
In this case, the eigenvectors of the correlation matrix of the observations may provide the modes of the sub-scale motion, to be modeled by applying stochasticity with those statistics at the resolved scales.  Although fluid dynamics is not considered in the present work, it falls within the purview of geometric mechanics and has been a source of inspiration in the previous development of stochastic geometric mechanics \cite{Ho2015}. 

\paragraph{Stochastic perturbations in finite dimensions.} As an example of data-driven modeling in finite dimensions, we consider the following situation. Suppose one notices an erratic ``jitter'' or ``wobble'' in the motion of an observable quantity, $q(t)=g(t)^{-1}q_0$ whose dynamics is governed by a subsystem of the full dynamics. For example, one might observe a jitter in the angular velocity of an orbiting  satellite, indicated, say, by a small antenna attached to it. Being only a subsystem quantity, and satisfying an auxiliary equation implying that it merely follows the rigid motions of the satellite, this observable quantity certainly does not determine the motion of the full system. However, the observation of its motion could still contain some useful information. For example, suppose its statistics can be measured. One may ask what dynamics of the full system would give rise to the observed statistics of the subsystem. In particular, one would be interested to know whether the observation of a perturbative wobble found in a subsystem could mean that the motion of the full system would eventually destabilize. If the dynamics of the unperturbed full system follows from Hamilton's principle, constrained by a deterministic auxiliary equation for an observable quantity, $q(t)$, then a reasonable procedure might be to take the variations, subject to the constraints determined from one's observations of the wobble in the subsystem, described as a stochastic perturbation of the original auxiliary equation for $q(t)$. Equivalently, given the observation of an apparently stochastic perturbation in a subsystem, one might ask, what motion equation gives rise to this stochastic wobble in the subsystem? In general, of course, this is not a well-posed question. However, for the geometric mechanics systems posed here, this question will have a definite answer. 

\paragraph{The rigid body example.} Euler's equation for stochastic motion for a rigid body provides a useful example in finite dimensions. For the Euler rigid body equations, the stochasticity introduced via the present approach enters the angular velocity and thereby provides a geometric mechanics description of stochastic motion of the angular momentum. In this type of problem, one asks, for example, whether an observed erratic perturbation in the angular velocity may destabilize a deterministic rigid body equilibrium. Indeed, it can, with positive probability. One also asks what the stochasticity does to the evolution of the energy and other conservation laws. Here the answer is interesting and suggestive of other potentially rich results. The first part of the answer is that the rigid body's energy is no longer conserved, but the magnitude of its angular momentum is still conserved, since the dynamics describes stochastic coadjoint motion. The rigid body example and the related heavy top example, when gravity is present, have been treated in \cite{ArCaHo2016a,ArCaHo2016b}.

\subsection{Stochastic Hamilton equations}
On the Hamiltonian side, the modern name for Poincar\'e's ``new form'' of dynamics is ``coadjoint motion''. 
The primary source of stochastic symplectic Hamilton equations is \cite{Bi1981}, which was recently reviewed and developed further from the geometric mechanics viewpoint in \cite{LCOr2008}.
In the present work, we are also interested in the situation where the motion is generated by applying a transformation group to a configuration manifold $Q$ with coordinates $q$, and then extending its action by cotangent lift to its entire phase space $T^*Q$ with coordinates $(q,p)$. The primary example occurs when the rotation group $G=SO(3)$ acts on $Q=\mathbb{R}^3$ and Poincar\'e's new form of the motion equation governs the angular momentum $J(q,p)\in\mathfrak{so}(3)^*\simeq \mathbb{R}^3$ of the rigid body, or heavy top, or spherical pendulum. This situation requires the noise to be present in both the $q$ and $p$ equations. Bismut's 1981 book \cite{Bi1981} discusses the Hamiltonian dynamics of stochastic particle motion, in which  
\begin{equation}
            dq = \{q\,,\, H(q,p) dt + \sum_i h_i(q,p) \circ dW^i(t) \}
\,,\qquad
            dp = \{p\,,\, H(q,p) dt + \sum_i h_i(q,p) \circ dW^i(t) \}
\,,
\label{Ham-stoch}
\end{equation}
for the canonical Poisson bracket $\{\, \cdot \,,\, \cdot\, \}$. If the stochastic Hamiltonians $h_i(q)$ happen to depend only on position $q$, then stochasticity appears only in the canonical momentum equation, as a Newtonian force,
\begin{equation}
           dq = \{q\,,\, H(q,p) dt\}
\,,\qquad
            dp = \{p\,,\, H(q,p) dt + \sum_i h_i(q) \circ dW^i(t) \}
\,.
\label{Newton-stoch}
\end{equation}
In this restricted case, the difference between Stratonovich and It\^o noise is immaterial. 
However, for rotating motion in three dimensions $q\in \mathbb{R}^3$, for example, we will need the stochastic Hamiltonians $h_i(q,p)$ to depend on both $q$ and $p$, since $q$ and $p$ transform the same way under rotations. In particular, they both transform as vectors in $\mathbb{R}^3$. In this situation, the noise appears in both of the equations in \eqref{Ham-stoch}, and the difference between Stratonovich and It\^o noise is crucial. The distinction between Stratonovich and It\^o noise is important for all of the motion equations in Poincar\'e's form, since the transformation of the conjugate momentum $p$ is the cotangent lift of the transformation of coordinate $q$ in Poincar\'e's class of equations. 

\subsection{A distinction from other approaches}

Although Poincar\'e \cite{Po1901} used a version of what one would now call 
``reduction by symmetry'', here we use an earlier approach due to Clebsch 
\cite{Clebsch1859}, which introduces constrained variations into 
Hamilton's principle by imposing velocity maps corresponding in the 
deterministic case to the infinitesimal transformations of a Lie group.
(For up to date applications to mechanics of the Clebsch method, see
\cite{GBRa2011}.) In a certain sense, Clebsch \cite{Clebsch1859} presages 
the Pontryagin maximum principle in optimal control theory. In the 
present paper, however, the velocity maps will be made stochastic.

Thus, we will consider \textit{stochastic} Clebsch action principles whose variables are stochastic. The equations of motion derived will be stochastic ordinary (or partial) differential equations (SDEs, or SPDEs) for 
motion on coadjoint orbits of (finite or infinite dimensional) Lie algebras.

Now we comment further on the distinction between the 
stochastic Clebsch and reduced Lagrangian approaches.
A stochastic Lagrangian symmetry reduction process has been developed in 
\cite{ArChCr2014,ChCrRa2015}. In that case, the  Lagrangian curves in the 
configuration space are stochastic diffusion processes, which are critical 
states of the action functional.
In these works, the drift of the stochastic processes is essentially regarded as its (mean, generalized) time derivative.
In \cite{ArChCr2014}, the action functional is defined with the classical Lagrangian computed 
on that velocity;  the corresponding Euler-Poincar\'e equations of motion, satisfied by 
the velocity, are  \textit{deterministic} (ordinary differential 
equations when the configuration space is finite-dimensional, or 
partial differential equations in the infinite-dimensional case). In \cite{ChCrRa2015}
the same kind of reduction process for stochastically perturbed Lagrangians is considered and
corresponding stochastic differential equations of motion (stochastic partial differential equations
in the infinite dimensional case) are derived.

In the present paper, as in \cite{Ho2015}, the stochastic Clebsch 
approach is not equivalent to the reduced stochastic Lagrangian 
processes approach employed in \cite{ArChCr2014,ChCrRa2015}. 
 In particular, the velocities 
in the reduction approach of \cite{ArChCr2014,ChCrRa2015} are essentially
identified with the drift of the underlying diffusion processes, which, 
as is well known, requires the computation of a conditional expectation.
In addition, in the reduced stochastic Lagrangian approach of 
\cite{ArChCr2014,ChCrRa2015}, it is not possible to take arbitrary
variations; instead, a particular form for the variations is required 
and the final resulting equations of motion depend on this choice. 
However, in the present work and in \cite{Ho2015}, the variations 
are quite arbitrary.

Therefore, the present stochastic Clebsch action principle cannot 
be regarded as a formulation of the Euler-Poincar\'e variational 
principle  obtained in \cite{ChCrRa2015}. In 
order to consider the present variational principle approach from the 
viewpoint of reduction by symmetry, one would need to interpret the 
velocity as an It\^o derivative of the  underlying stochastic curves, in 
which case the resulting stochastic action functional would be divergent. 
This divergence was avoided in \cite{ArChCr2014, ChCrRa2015} via the 
``renormalization" achieved by taking conditional expectations.

\paragraph{Outline of the paper.}
Following the Clebsch approach to the Euler-Poincar\'e equations, 
in Section \ref{VP-coadmotion-sec} we introduce a stochastic velocity 
map in the Stratonovich sense as a constraint in Hamilton's principle 
for motion on a manifold acted upon by infinitesimal transformations 
of a Lie algebra. With hindsight, we see that the stochasticity in the 
velocity map is coupled to the motion by the momentum map which arises 
from the variation of the Lagrangian function and the deterministic 
part of the velocity map. The resulting   stationarity conditions 
generalize the classical deterministic formulations of motion 
on coadjoint  orbits of Lie algebras in Poincar\'e \cite{Po1901} and 
Hamel \cite{Ha1904}, by making them stochastic. In Section 
\ref{Ito-form-coadmotion-sec}, we present the It\^o formulations of 
the stationary variational conditions. Three alternative routes 
are taken in calculating the It\^o double-bracket forms of the 
variational equations for stochastic coadjoint motion. In
Section \ref{Properties-stochcoadmotion-sec}, we  discuss the Poisson 
structure of the Stratonovich-Hamiltonian formulation of the stochastic 
motion equations. We also give the It\^o interpretation of the Casimir 
functions for the Lie-Poisson part of the bracket in this formulation, and 
derive the associated Lie-Poisson Fokker-Planck equation for the motion of 
the probability density function on
the level sets of Casimir functions.

\section{Variational principle for Stratonovich stochastic coadjoint motion}\label{VP-coadmotion-sec}

\subsection{Deterministic formulation}

In \cite{Po1901}, Poincar\'e begins by considering the transitive action 
of a Lie group $G$ of smooth transformations of a manifold $Q$, whose 
points in local coordinates are written as $q =(q^1, \ldots, q^n)$ and 
whose infinitesimal transformations are represented by the vector field 
obtained at linear order in the Taylor series. Let $\{e_1, \ldots, e_r\}$
be a basis of $\mathfrak{g}$ and $\alpha=1, \ldots, r$ 
the indices of the local coordinates in this basis. 
Denote by $A_{\alpha}[f]$ any infinitesimal transformation of this group, 
and express its action on a smooth function $f$ as
\begin{equation}
A_{\alpha}[f]
:=\sum_{i=1}^{n}A^{i}_{\alpha}\pp{f}{q^{i}}
=A_{\alpha}^{1}\pp{f}{q^{1}}+A_{\alpha}^{2}\pp{f}{q^{2}}+\cdots+
A_{\alpha}^{n}\pp{f}{q^{n}}
\,,\quad \alpha=1,\ldots, r,
\label{InfXform}
\end{equation}
where $A_\alpha^i$ are functions of 
$(q^1, \ldots, q^n)$.

Throughout this paper, Greek indices enumerate 
Lie algebra basis elements, Latin indices denote coordinates on the 
manifold, and the standard Einstein summation convention is assumed.
Since these transformations form a Lie algebra, 
Poincar\'e remarks that
\begin{equation} 
A_{\alpha}[A_{\beta}]-A_{\beta}[A_{\alpha}]=
\sum_{\gamma=1}^r c_{\alpha\beta}{}^\gamma A_{\gamma} \quad 
\Longleftrightarrow \quad 
A_\alpha^s \frac{\partial A_\beta^k}{\partial q^s} -
A_\beta^s \frac{\partial A_\alpha^k}{\partial q^s} = 
c_{\alpha\beta}{}^\gamma(q) A_ \gamma^k, \;\; \forall k=1, \ldots, n, 
\;\alpha, \beta = 1, \ldots, r,   
\label{VFcomrel}
\end{equation}
where $c_{\alpha \beta}{}^\gamma(q) \in C^\infty(Q)$ are structure \textit{functions} for the Lie algebra of smooth vector fields on the manifold $Q$, in the basis associated with the Greek indices. When the $A_\alpha^i(q)$ are \textit{linear} functions of $(q^1, \ldots, q^n)$, then the $c_{\alpha \beta}{}^\gamma \in\mathbb{R}$ are the usual structure \textit{constants} of a matrix Lie algebra. In this regard, Poincar\'e \cite{Po1901} presages Hamel \cite{Ha1904}, cf. also Marle \cite{Ma2013}.

\paragraph{Geometric setup.}
We give now a glimpse of the global formulation. Poincar\'e \cite{Po1901}
does not really use a transformation group, only its associated Lie 
algebra action. In \cite{Po1901} Poincar\'e takes a configuration $n$-manifold $Q$ of a mechanical system
and a Lie algebra morphism $\mathfrak{g}\ni u \mapsto u_Q \in  
\mathfrak{X}(Q)$
of a given Lie algebra $\mathfrak{g}$, $\dim \mathfrak{g} =:r<\infty$, 
to the Lie algebra $\mathfrak{X}(Q)$
of vector fields on $Q$, endowed with the usual Lie bracket $[X,Y][f]:=
X[Y[f]] - Y[X[f]]$, where $X,Y \in \mathfrak{X}(Q)$, $f\in C^{\infty}(Q)$, 
and $X[f]$ is the differential of $f$ in the direction $X$, given in 
coordinates by \eqref{InfXform}. The coordinate expression 
\begin{equation}
\label{inf_gen}
u_Q (q) =:u_Q^i (q)\frac{\partial}{\partial q^i}=:A_\alpha^i (q)u^\alpha \frac{\partial}{\partial q^i}
\end{equation}
of $u_Q \in \mathfrak{X}(Q)$, relative to a coordinate system 
$(q^1, \ldots, q^n)$ on the chart domain $U \subset Q$ and a basis 
$\{e_1, \ldots, e_r\}$ of 
$\mathfrak{g}$, is thus determined by the functions 
$A_\alpha^i \in C^{\infty}(U)$ and the basis expansion 
$u =: u^\alpha e_\alpha$ of $u \in \mathfrak{g}$. Since $[u_Q, v_Q] = 
[u,v]_Q$ for any $u, v \in \mathfrak{g}$, the functions $A_\alpha: =
\left[A_\alpha^i \right] \in C^{\infty}(U, \mathbb{R}^n)$, defined 
by $(e_\alpha)_Q =: A_\alpha^i\frac{\partial}{\partial q^i}$, satisfy
\eqref{VFcomrel}\footnote{Thus, Poincar\'e works with a \textit{right}
action of the underlying Lie group on the manifold; we adopt his 
index conventions in \cite{Po1901},  also used in \cite{BlMaZe2009}.
For \textit{left} actions,
$\mathfrak{g} \rightarrow \mathfrak{X}(Q)$ is a Lie algebra  
\textit{anti-homomorphism}, i.e., $[u_Q, v_Q] = - [u,v]_Q$.},
which is equivalent to saying that the local vector fields
$A_\alpha, A_\beta \in \mathfrak{X}(U)$ satisfy
\begin{equation}
\label{A_alpha_commutator}
[A_\alpha, A_\beta] \stackrel{\eqref{VFcomrel}}= 
c_{\alpha\beta}{}^\gamma(q) A_ \gamma
\,.
\end{equation}
The action is assumed to be \textit{transitive}
in \cite{Po1901}, which means that \textit{any} tangent vector 
$v_q \in T_qQ$ is of the form $v_q = u_Q(q)$ for some $u \in \mathfrak{g}$, 
and hence if $u = a^\alpha e_\alpha$ for some $a^\alpha \in\mathbb{R}$, 
then $v_q$ can be written locally as
$v_q = a^\alpha (e_\alpha)_Q(q) = a^\alpha A_\alpha^i(q)
\frac{\partial}{\partial q^i}$.

If $(q^1, \ldots, q^n)$ are local coordinates on $Q$, the corresponding 
standard coordinates on the tangent bundle $TQ$ and the cotangent bundle
$T^*Q$ are, respectively, $(q^1, \ldots, q^n, \dot{q}^1, \ldots, \dot{q}^n)$ and 
$(q^1, \ldots, q^n,p_1, \ldots, p_n)$, where 
$v_q =\dot q^i \frac{\partial}{\partial q^i}$ and $p_q = p_i dq^i$
for any $v_q \in T_qQ$ and $p_q \in T^*_qQ$ (the cotangent space at
$q \in Q$, the dual of $T_qQ$). Throughout the paper, we use these
naturally induced coordinates. The sign convention for the canonical
Poisson bracket on $T^*Q$ adopted in this paper is, in standard coordinates,
\begin{equation}
\label{can_Poisson_bracket}
\{f, g\} = \frac{\partial f}{\partial q^k}
\frac{\partial g}{\partial p_k} - \frac{\partial g}{\partial q^k}
\frac{\partial f}{\partial p_k}\,, \quad \text{for any} \quad 
f, g \in C^{\infty}(T^*Q).
\end{equation}
If $h \in C^{\infty}(T^*Q)$, its Hamiltonian vector field is denoted
by $X^{T^*Q}_h \in \mathfrak{X}(T^*Q)$. Hamilton's equations for
a curve $c(t) \in T^*Q$ in Poisson
bracket form are $\frac{d}{dt}f(c(t)) = \{f, h\}(c(t))$ for any $f \in 
C^{\infty}(T^*Q)$. 

When working with a general Poisson manifold $(P, \{\cdot, \cdot \})$,
the Hamiltonian vector field $X^P_h \in \mathfrak{X}(P)$ of $h \in C^{\infty}(P)$ is defined by 
$\mathbf{d}f\left(X^P_h\right) := \{f, h\}$. For a symplectic manifold, $(P, \omega)$, this is equivalent to the usual definition, $\mathbf{i}_{X^P_h}\omega = X^P_h \contract \omega = \mathbf{d}h$.

\paragraph{Pairing notation.} 
For any manifold $Q$, 
finite or infinite dimensional, we denote by
$\left\langle \cdot , \cdot \right\rangle_Q:
T^*Q \times TQ\rightarrow\mathbb{R}$ the natural (weakly, in
the infinite-dimensional case) non-degenerate fiberwise duality pairing.
Given a Lie algebra $\mathfrak{g}$, which is always finite dimensional
in this paper, the non-degenerate duality pairing between its dual 
$\mathfrak{g}^\ast$ and $\mathfrak{g}$ is denoted by $\left\langle\cdot , 
\cdot \right\rangle_\mathfrak{g}: \mathfrak{g}^\ast \times \mathfrak{g}
\rightarrow \mathbb{R}$. 

Given $f \in C^{\infty}(\mathfrak{g}^\ast)$, the \textit{functional derivative} $\frac{\delta f}{\delta \mu} \in \mathfrak{g}$ of $f$
evaluated at $\mu \in \mathfrak{g}^\ast$ is defined by
\begin{equation}
\label{func_der}
\left.\frac{d}{d\epsilon}\right|_{\epsilon=0}
f(\mu + \epsilon \delta \mu) = 
\left\langle \delta\mu, \frac{\delta f}{\delta \mu}
\right\rangle_\mathfrak{g}, \quad \text{for all} \;\; 
\delta\mu \in \mathfrak{g}.
\end{equation}

\paragraph{The momentum map.} The \textit{momentum map} 
$\mathbf{J}_{T^*Q}: T^*Q \rightarrow 
\mathfrak{g}^\ast$ of the 
lifted $\mathfrak{g}$-action to $T^*Q$ is defined by 
\begin{equation}
\label{momentum_map_def}
u_{T^*Q} = X^{T^*Q}_{\mathbf{J}^u},\quad \text{for any}
\quad  u \in \mathfrak{g},
\end{equation}
where $\mathbf{J}_{T^*Q}^u(p_q) := \left\langle \mathbf{J}_{T^*Q}(p_q), u
\right\rangle_\mathfrak{g}$.
Its expression is, cf. \S12.1, formula (12.1.15) of \cite{MaRa1999},
\begin{equation}
\label{momentum_map}
\mathbf{J}_{T^*Q}^u(p_q) 
= \left\langle p_q, u_Q(q) \right\rangle_Q, \quad 
p_q\in T^*Q, \;\; u \in \mathfrak{g}, \quad \text{or, in coordinates,}
\quad \mathbf{J}_{T^*Q}(q^i, p_i) = p_j A_\alpha^j(q^i)e^\alpha,
\end{equation} 
where $\{e^1, \ldots, e^r\}$ is the basis of $\mathfrak{g}^\ast$ dual
to the basis $\{e_1, \ldots, e_r\}$ of $\mathfrak{g}$. This momentum 
map is infinitesimally equivariant. That is,\footnote{This is the infinitesimal
equivariance relation for \textit{right} actions. For \textit{left}
actions, the sign in the right hand side changes.} 
\begin{equation}
\label{inf_eqvar}
\mathbf{J}_{T^*Q}^{[u,v]}= -\{\mathbf{J}_{T^*Q}^u,\mathbf{J}_{T^*Q}^v\}\,,
\end{equation}
for all $u, v \in \mathfrak{g}$. A useful equivalent statement of
infinitesimal equivariance is (see, e.g., 
\cite[\S11.5, formula (11.5.6)]{MaRa1999} with a sign change because
we work with right actions) 
\begin{equation}
\label{inf_equ_der}
T_{p_q} \mathbf{J}_{T^*Q} \left(u_{T^*Q}(p_q) \right) = 
\operatorname{ad}_u^*\mathbf{J}_{T^*Q}(p_q)
\,.
\end{equation}
If we denote 
\[
m:=\mathbf{J}_{T^*Q} \in C^{\infty}(T^*Q, \mathfrak{g}^\ast)
\,,\quad
m_\alpha :=\mathbf{J}_{T^*Q}^{e_\alpha}
\,,\quad
m_{[\alpha, \beta]}:= 
\mathbf{J}_{T^*Q}^{[e_\alpha, e_\beta]}\in C^{\infty}(T^*Q)
\,,\]
we have
$m= m_\alpha e^\alpha$, 
with
\begin{equation}
\label{m_alpha}
m_\alpha(p_q) = p_i A_\alpha^i (q)
\end{equation}
and the infinitesimal equivariance is expressed in coordinates as
\begin{equation}
\label{m_alpha_beta}
-m_{[ \alpha, \beta]} = \{m_\alpha, m_\beta\}
\stackrel{\eqref{can_Poisson_bracket}}= 
-p_j\left[A_\alpha, A_\beta \right]^j 
\stackrel{\eqref{A_alpha_commutator}}= 
-p_jc_{\alpha\beta}{}^\gamma A_\gamma^j
\stackrel{\eqref{m_alpha}}= - c_{\alpha\beta}{}^\gamma m_\gamma.
\end{equation}

If $G$ is a Lie group with Lie algebra $\mathfrak{g}$ acting on the right
on $Q$, then the $\mathfrak{g}$-action on $Q$ is given by the infinitesimal
generator vector field $u_Q \in \mathfrak{X}(Q)$ defined at every
$q \in Q$ by $u_Q(q):= \left.\frac{d}{dt}\right|_{t=0}
q \cdot \exp(tu) \in T_qQ$, where $q \cdot g$ denotes the action of
$g \in G$ on the point $q \in Q$. The momentum map $\mathbf{J}: T^*Q 
\rightarrow \mathfrak{g}^\ast$ given in \eqref{momentum_map} is
equivariant relative to the given right $G$-action on $P$ and the
right coadjoint $G$-action on $\mathfrak{g}^\ast$, i.e., 
$\mathbf{J}_{T^*Q}(p_q \cdot g) = \operatorname{Ad}^\ast_g 
\mathbf{J}_{T^*Q}(p_q)$ for all $g \in G$. We use here the following
notations: $\operatorname{Ad}_g: \mathfrak{g}\rightarrow \mathfrak{g}$
is the adjoint action of $g \in G$, defined as the derivative
at the identity of the conjugation by $g$ in $G$; $\operatorname{Ad}_g$
is a Lie algebra isomorphism; $\operatorname{Ad}_g^*: \mathfrak{g}^\ast
\rightarrow\mathfrak{g}^\ast$ is the dual map of $\operatorname{Ad}_g$.

\paragraph{The Lie-Poisson bracket.} 
The dual $\mathfrak{g}^\ast$ 
of any finite dimensional Lie algebra $\mathfrak{g}$ is endowed with
the \textit{Lie-Poisson bracket} (see, e.g., 
\cite[\S13.1, p.416]{MaRa1999})
\begin{equation}
\label{LP}
\{f, h\}_{\pm} (\mu) = \pm \left\langle \mu, 
\left[\frac{\delta f}{ \delta\mu}, \frac{\delta h}{\delta\mu} \right]
\right\rangle_\mathfrak{g}, \qquad f, h \in C^{\infty}(\mathfrak{g}^\ast), \quad 
\mu\in\mathfrak{g}^\ast.
\end{equation}
We denote by $\mathfrak{g}^\ast_{\pm}$ the vector space $\mathfrak{g}^\ast$
endowed with the Poisson bracket \eqref{LP}. The Hamiltonian vector
field of $h \in C^{\infty}(\mathfrak{g}^\ast)$ defined by the
equation $\dot{f} = \{f, h\}$ for any $f \in C^{\infty}(\mathfrak{g}^\ast)$
has the expression $X_h^{\pm}(\mu) = \mp 
\operatorname{ad}_{\frac{\delta h}{\delta\mu}}^* \mu$ (the signs 
correspond). Given $\xi \in \mathfrak{g}$, $\operatorname{ad}_\xi^*: \
\mathfrak{g}^\ast\rightarrow \mathfrak{g}^\ast$ is the dual of the 
linear map $\mathfrak{g}\ni\eta \mapsto \operatorname{ad}_\xi \eta: =
[ \xi, \eta] \in \mathfrak{g}$.

If $G$ is a Lie group with Lie algebra $\mathfrak{g}$, left translation by
$g \in G$ is denoted as $G\ni h \mapsto L_g(h): = gh \in G$ and, likewise,  
$G\ni h \mapsto R_g(h): = hg \in G$ denotes right translation. 
Then the momentum
map $\mathbf{J}_R: T^*G \rightarrow \mathfrak{g}^\ast_-$, 
$\mathbf{J}_R(\alpha_g) = T_e^*L_g (\alpha_g)$, $\alpha_g \in T^*_gG$, 
of the lifted right translation on $G$ (i.e., $\mathbf{J}_R$ equals
$\mathbf{J}_{T^*G}$ given in \eqref{momentum_map} for the action of
$G$ on itself given by right translations) is a Poisson map, i.e., 
$\{f, h\}_- \circ \mathbf{J}_R = 
\{f \circ\mathbf{J}_R, h \circ \mathbf{J}_R\}$. Similarly, the
momentum map $\mathbf{J}_L: T^*G \rightarrow \mathfrak{g}^\ast_+$, 
$\mathbf{J}_L(\alpha_g) = T_e^*R_g (\alpha_g)$, $\alpha_g \in T^*_gG$,
for left translation is another Poisson map, i.e., 
$\{f, h\}_+ \circ \mathbf{J}_L = \{f \circ\mathbf{J}_L, 
h \circ \mathbf{J}_L\}$. (For the proof see, e.g., \cite[\S13.3]{MaRa1999}.) 
More generally, the momentum map $\mathbf{J}: T^*Q \rightarrow 
\mathfrak{g}^\ast_-$ of the lifted right $\mathfrak{g}$-action to $T^*Q$ 
is a Poisson map; the coordinate expression of this statement is 
\eqref{m_alpha_beta}.

A function $k \in C^{\infty}(\mathfrak{g}^\ast)$ such that 
$\{k,f\}_{\pm} =0$, for all $f \in  C^{\infty}(\mathfrak{g}^\ast)$, 
or, equivalently, $X^{\pm}_k =0$, is called a \textit{Casimir function}. 
This definition is valid for any Poisson manifold, not just 
$\mathfrak{g}^\ast$.

\subsection{Stochastic Clebsch formulation}

\paragraph{Introducing stochasticity into the Clebsch methodology.}
We assume that all stochastic processes are defined in the same filtered probability space $(\Omega, \mathbb P,  {\cal P}_t )$.
Let $t \mapsto W_t^k(\omega)$, $k = 1, \ldots N$, $\omega \in\Omega$, 
be $N$ independent real-valued Brownian motions, $\xi_1, \ldots, 
\xi_N \in \mathfrak{g}$, and $\Omega \ni \omega \mapsto 
(p_q)_\omega(t)\in T^*Q$ random variables for every $t$. 
The induced random variable on $Q$, the foot point of 
$(p_q)_\omega(t)$, is denoted by $\Omega \ni\omega
\mapsto q_\omega (t)\in Q$. 
Stratonovich differentiation is denoted by $X\circ dY$ and It\^o differentiation simply  by $XdY$. Then, given $\xi_1, \ldots, \xi_N 
\in \mathfrak{g}$ and a $\mathfrak{g}$-valued random curve 
$u(t)$,
$$
t \longmapsto 
\left\langle
(p_q)_\omega (t),\,\circ d q_\omega (t) -u_\omega (t)_Q (q_\omega(t))dt -
(\xi_k )_Q (q_\omega(t))\circ d W_t^k(\omega)
\right\rangle_Q
$$
is a process whose coordinate expression is
$$
(p_i)_\omega(t) \left(\circ d q_{\omega}^i (t) -
A_\alpha^i (q_{\omega}(t))u_{\omega}^\alpha(t) dt  -  
A_\alpha^i (q_\omega(t))\xi_k^\alpha 
\circ d W_t^k \right).
$$
We always assume that the stochastic processes are defined for all times $t\in [0,T]$, the coefficients are smooth, and that $u(t)$ is smooth in the time variable. Furthermore, we assume that the manifold $Q$ has no boundary.

\paragraph{Remark on notation.}
For simplicity in the notation, we no longer write the probability variable $\omega$ and, instead, we use symbols $\mathscr{p,q}$, etc., to denote semimartingales. 

Given the Lagrangian $\ell \in C^{\infty}(\mathfrak{g} \times Q)$, 
introduce the stochastic action, defined for random curves
$u \in C^{1}([0,T], \mathfrak{g})$, $\mathscr{q} \in  C([0,T], Q)$, $(\mathscr{p}_\mathscr{q}) \in C([0,T], T^* Q)$,
and define the constrained stochastic action integral $S(u,\mathscr{p}_\mathscr{q})$ by
\begin{align}
S( u,\mathscr{p}_\mathscr{q})  &=
\int_0^T  \ell (u(t),\mathscr{q}(t))  dt + 
\langle \mathscr{p}_{\mathscr{q}}(t)\,,\,  \circ d\mathscr{q}(t) -u(t)_Q (\mathscr{q}(t))dt -(\xi_k )_Q (\mathscr{q}(t)) \circ dW_t^k \rangle 
\,,\label{SVP1_a}
\end{align}
where the semimartingale $\mathscr{p}_\mathscr{q}$ is assumed to be regular enough for the above integrals to be  finite.
Indeed, all stochastic processes considered in this paper will be continuous semimartingales with regular coefficients.
In local coordinates, the stochastic action integral \eqref{SVP1_a} may be recognized as the sum of a Lebesgue integral and a Stratonovich integral
\begin{align}
S(u,\mathscr{p}_\mathscr{q})  &=
\underbrace{
\int_0^T \left(\ell(u(t),\mathscr{q}(t)) dt - 
\mathscr{p}_i(t) A^i_\alpha(\mathscr{q}(t))u^\alpha(t) \,dt \right) 
}_{\hbox{Lebesgue integral}}
+
\underbrace{
\int _0^T \mathscr{p}_i\left(\circ d\mathscr{q}^i (t) -
A^i_\alpha(\mathscr{q}(t))\xi_k^\alpha\circ dW^k_t 
\right) }_{\hbox{Stratonovich integral}}
\,.\label{SVP1_b}
\end{align}

For notational convenience, we introduce for every $t\in [0, T]$
the stochastic Lie algebra element whose components 
in the basis $\{e_1, \ldots, e_r\}$ of $\mathfrak{g}$ are
\begin{align}
d\mathscr{x}_t^\alpha := u^\alpha(t) dt + \xi_k^\alpha\circ dW^k_t
\,,\label{dx-def}
\end{align}
to convey that, when we integrate some stochastic process $\mathscr{X}_t$ with respect to  $d\mathscr{x}_t^\alpha$, we mean
$$
\int_0^T \mathscr{X}_t d\mathscr{x}_t^\alpha := 
\int _0^T\mathscr{X}_t u^\alpha(t) dt + 
\int_0^T \mathscr{X}_t \xi_k^\alpha\circ dW^k_t \,, \quad
\alpha = 1, \ldots, \dim \mathfrak{g}=r.
$$
In particular, we rewrite the action integral in \eqref{SVP1_b} in the abbreviated form
\begin{align}
S(u,\mathscr{p}_\mathscr{q})  &=
\int_0^T \ell(u(t),\mathscr{q}(t))dt + 
\mathscr{p}_i(t)\left({\circ d\mathscr{q}^i(t)} - 
A^i_\alpha( \mathscr{q}(t))d\mathscr{x}_t^\alpha \right)  
\,.\label{SVP2}
\end{align}
We assume that the Lagrangian $\ell(u,q)$, viewed as 
a function $\ell: \mathfrak{g} \times Q \rightarrow \mathbb{R}$, is 
\textit{hyperregular}, i.e., for every $q \in Q$, the
map $\mathfrak{g} \times \{q\} \ni (u,q) \mapsto n:=
\frac{\delta \ell}{\delta u} \in \mathfrak{g}^\ast \times \{q\}$ is 
a diffeomorphism. In particular,
$n$ is a function of $(u,q)$ and, conversely, $u$ is a function of
$(n, q)$. Thus, replacing the variables $u\in \mathfrak{g}$ and
$q \in Q$ by the random curves $u(t)$, $\mathscr{q}(t)$, we get 
the semimartingale $n(u(t),\mathscr{q}(t))$. 
 
Consider a random point $(q_\omega, p_\omega)$ in the manifold $T^* Q$ and 
$f\in C^\infty (T^*Q)$. The differential of
$f$ in the direction of the (deterministic) vector field $Z\in 
\mathfrak{X}(T^*Q)$ is given by
$$
\left\langle \mathbf{d}f, Z \right\rangle_Q (q_\omega ,p_\omega )=  
\left.\frac{d}{d\epsilon}\right|_{\epsilon=0} f(\gamma_\omega (\epsilon )),
$$
where $\gamma_\omega$ is a curve starting from $(q_\omega , p_\omega )$ 
with initial velocity $Z(q_\omega, p_\omega)$ and the limit is taken
in $L^2 (\Omega )$. 
Therefore $\left\langle\mathbf{d}f, Z\right\rangle_Q (q_\omega ,p_\omega )$ 
consists in evaluating  $\left\langle\mathbf{d}f, Z\right\rangle_Q$ at the 
random point $(q_\omega,p_\omega)$.

Consider  now a semimartingale of the form 
\begin{equation}
\label{semimform}
\mathscr{Y}_t (\mathscr{q},\mathscr{p})=\mathscr{Y}_0 +\int_0^t \phi_\alpha 
(\mathscr{q}(s),\mathscr{p}(s)) \xi_k^\alpha \circ dW^k_s +\int_0^t \psi (\mathscr{q}(s),
\mathscr{p}(s)) ds,
\end{equation}
where  $\mathscr{q}(t)$, $\mathscr{p}(t)$ are $Q$-, respectively, $T^*Q$-valued 
semimartingales, with $\mathscr{q}(t)$ the footpoint of $\mathscr{p}(t)$, 
$\phi_\alpha,\, \psi \in C^{\infty}(T^*Q)$ are 
deterministic smooth functions,
and $\xi_k = \xi_k^\alpha e_\alpha \in \mathfrak{g}$ are 
given (constant) elements. The \textit{{\rm(}Stratonovich{\rm)}  stochastic Poisson bracket} of $f(\mathscr{q}(t),\mathscr{p}(t))$ 
with $\mathscr{Y}_t$
is defined by
\begin{align}
\begin{split}
\{f(\mathscr{q}(t),\mathscr{p}(t)),\circ d_t \mathscr{Y}_t\} :=& 
\left\langle\mathbf{d}f, X_{\phi_\alpha}\right\rangle_Q  (\mathscr{q}(t),\mathscr{p}(t)) \xi_k^\alpha \circ dW^k_t + 
\left\langle\mathbf{d}f, X_\psi\right\rangle_Q(\mathscr{q}(t),\mathscr{p}(t)) dt\\
=& \{f, \phi_\alpha \} (\mathscr{q}(t),\mathscr{p}(t))\xi_k^\alpha 
\circ dW^k_t +\{f,\psi\} (\mathscr{q}(t),\mathscr{p}(t))dt.
\end{split}
\label{stochastic_PB}
\end{align}
where $X_\phi$, $X_\psi$ denote the Hamiltonian vector fields of $\phi $ and $\psi$.

If $g \in C^{\infty}(T^*Q)$, the Poisson bracket of the two semimartingales
$f(\mathscr{q}(t),\mathscr{p}(t))$ and  $g(\mathscr{q}(t),\mathscr{p}(t))$
is defined as 
\begin{align}
\label{standard_stochastic_PB}
\left\{f(\mathscr{q}(t),\mathscr{p}(t)), g(\mathscr{q}(t),\mathscr{p}(t)) \right\}: &=  \{f, g\}(\mathscr{q}(t),\mathscr{p}(t)),
\end{align}
i.e., it equals the semimartingale obtained by computing the function
$\{f, g\} \in C^{\infty}(T^*Q)$ and replacing its variables $(q,p)$ 
by the semimartingales $(\mathscr{q}(t),\mathscr{p}(t))$.

The constraint imposed by the pairing with the Lagrange multipliers 
$\mathscr{p}_i(t)$ defines the $i$th component of the \emph{stochastic velocity map}, 
\begin{align}
d\mathscr{q}^i(t) = A^i_\alpha(\mathscr{q}(t)) d\mathscr{x}_t^\alpha
\,.\label{SVP3}
\end{align}

To justify the computations that follow on manifolds 
and ensure that they are intrinsic, we provide a quick review of the 
basics of the Malliavin Calculus in the next subsection.

\subsection{Calculus of variations on path spaces}
In this subsection we give a brief summary of some definitions and 
results about the calculus of variations on (probability) path spaces 
known as \textit{Malliavin Calculus}, both in the case where the paths 
take values on Euclidean spaces and on  Riemannian manifolds. For this 
subject we refer to  \cite{Mall1997}.

\paragraph{Malliavin derivative: the Euclidean case.} Beginning with the 
Euclidean configuration space case, let 
$x_0 \in \mathbb{R}^n$ be given, and fixed throughout the
discussion below, and let $\mathbb P_{x_0} 
=\{ x:  [0,T]\rightarrow \mathbb R^n ,\; x ~\hbox{continuous},\; 
x(0)=x_0\}$ be the path space of continuous paths endowed with the law 
$\mu$ of the Brownian motion  on $\mathbb R^n$ starting from $x_0$ at time 
$0$   and with the usual past filtration ${\mathcal P}_t$. A variation of 
the paths $x$ is a map 
$z: [0,T]
\rightarrow \mathbb R^n$ of
bounded variation and such that $ \int_0^T |\frac{d}{dt} z(t)|^2 dt 
<\infty$, $z(0)=0$. These are the elements in the Cameron-Martin space, which is dense in 
$\mathbb P_{x_0}$ for the sup topology.
For a functional $F\in L^p_\mu (\mathbb P_{x_0} )$, the 
\textit{Malliavin derivative} of $F$ in the direction $z$ is defined as 
$$
D_z F(x)=\lim_{\epsilon \rightarrow 0} \frac{1}{\epsilon} 
\big( F(x+\epsilon z) -F(x) \big) ,
$$
the limit being taken in the $L^2_{\mu} $ sense. For a ``cylindrical" functional of the form
$F(x) =f(x(t))$, for each $t\in [0,T]$ fixed, where $f$ is a real valued 
smooth (at least $C^1$) function, we have
$$
D_z f(x(t))= \frac{\partial f}{\partial x^i} (x(t)) z^i (t).
$$

For a semimartingale $\zeta$ with values in $\mathbb{R}^n$ we can also define
$D_{\zeta} f(x(t))= \frac{\partial f}{\partial x^i} (x(t)) \zeta^i (t)$.

If we consider  pinned Brownian paths (or bridges), the corresponding path space will be
 $\mathbb P_{x_0 ,x_T} 
=\{ x:  [0,T]\rightarrow \mathbb R^n ,\; x ~\hbox{continuous},\; 
x(0)=x_0 ,x(T)=x_T\}$. In this case the variations must satisfy the condition $z(T)=0$.

\paragraph{The It\^o map.} Analogously, given a smooth $n$-dimensional manifold $Q$ 
acted upon by a Lie group $G$ with Lie algebra $\mathfrak{g}$, and a point
$q_0 \in Q$, let $\mathbb P_{q_0} (Q):=\{ q:[0,T]\rightarrow Q \mid q ~\hbox{continuous},\; q(0)=q_0\}$ denote the path space of continuous
paths starting at $q_0$, endowed with the law of the process 
\[
dq(t):=(\xi_k)_Q (q(t))\circ dW_t^k +(u(t))_Q (q(t))dt
\,,\]
with $q(0)=q_0$. As usual, $\eta_Q$ denotes the infinitesimal generator 
vector field on $Q$ induced by $\eta \in \mathfrak{g}$, $\xi_1, 
\ldots, \xi_N\in \mathfrak{g}$ are $N$ given Lie algebra elements, 
and $t \mapsto u(t) \in \mathfrak{g}$ is a given random path 
that we assume adapted and of bounded variation (we are actually assuming  smoothness).
In addition, we request that the process $dq(t)$ defined 
above does not explode in finite time and, in particular, is defined 
for all $t\in [0,T]$. 

It is possible to define in the space $\mathbb P_{q_0} (Q)$ a global 
chart, as follows. On the manifold $Q$ we consider the 
bilinear form 
\[
\bar{g}(q)(p_q, \beta_q): = \sum_{k=1}^N 
p_q\left(\left(\xi_k\right)_Q(q) \right)\
\beta_q\left(\left(\xi_k\right)_Q(q) \right), \qquad p_q, 
\beta_q \in T^*_qQ,
\]
whose coordinate expression is
$$
\bar{g}^{ij} =\sum_{k=1}^N A_\alpha^i 
\xi_k^\alpha A_\beta^j \xi_k^\beta =
\sum_{k=1}^N (\xi_k )^i_Q (\xi_k )^j_Q\,.
$$
It is assumed that this is a co-metric on $T^*Q$, i.e., $\bar{g}$ is positive
definite, so that  the diffusions $q$ are non singular, i.e., their generators are elliptic. The associated Riemannian metric on $Q$, denoted by $g$ or
$\left\langle\!\left\langle \cdot , \cdot\right\rangle\!\right\rangle$,
has a corresponding Levi-Civita connection $\nabla$.  A (stochastic) 
parallel transport over the paths $q(\cdot )$ can be defined (following It\^o, see for example, \cite{IkWat1989}). We denote it
by $t^q _{\tau_2 \leftarrow \tau_1}: T_{q(\tau_1 )} Q \rightarrow T_{q(\tau_2 )} Q$.

This amounts to solving the following stochastic system
$$
d_{\tau_1}[t^q _{\tau_2 \leftarrow \tau_1}]_i^j =
[t^q _{\tau_2 \leftarrow \tau_1}]_r^j 
 \Gamma_{si}^r \Big(  A^s_\alpha (q(t))\xi_k^\alpha \circ dW_t^k +
A^s_\alpha  (q(t)) u^\alpha dt \Big), \quad 
[t^q _{\tau_2 \leftarrow \tau_2}] =I,
$$
where $\Gamma$ are the Christoffel symbols associated with the Levi-Civita 
connection $\nabla$, $i,j,r,s=1, \ldots, n$, and 
$k=1, \ldots, N$. Parallel transport of vector fields can be lifted to 
parallel transport of orthonormal frames, namely isometries 
$r: \mathbb R^n \rightarrow T_q Q $ (c.f. \cite{IkWat1989}).
For a path  $W \in \mathbb P_{x_0}$, let $r_W (t)$ be the 
parallel transport of frames from time $0$ to time $t$ along $W(t)$, with 
$\pi (r_W (0))=q_0$ (where $\pi $ denotes the canonical projection from 
the orthonormal frame bundle over $Q$ to the underlying manifold $Q$).
We have
$t^q _{\tau_2 \leftarrow \tau_1} =r_W (\tau_2 ) r_W (\tau_1  )^{-1}$.

The
 It\^o map ${\mathcal I} : \mathbb P_{x_0} (\mathbb R^n )\rightarrow \mathbb P_{q_0} (Q)$, is defined by
$$
{\mathcal I}(W)(t)=\pi (r_W (t)).
$$
This map realizes an isomorphism of probability spaces, i.e., it is a bijective map that transports the law of the Brownian motion to the law of the process $q$ (\cite{Mall1997}).
It provides a global chart for the path space of the manifold that will be used throughout the paper.

\paragraph{Malliavin derivative: the manifold case.} A variation of the path $q\in \mathbb P_{q_0} (Q)$ is a map $Z_q (t)\in T_{q(t)}Q$ such that $z (t)=t^q_{0\leftarrow t} (Z_q (t))$ is a variation on the Euclidean path space $\mathbb P_{x_0} (\mathbb R^n )$
as defined above.  We have
$$
\frac{d}{dt} z (t) = t^q_{0\leftarrow t}(\nabla^q_t  Z (t))\,,
$$
where $\nabla^q$ is the covariant derivation $\nabla^q_t Z_q (t) =\lim_{\epsilon \rightarrow 0}\frac{1}{\epsilon} (t^q_{t\leftarrow t+\epsilon}(Z_q (t+\epsilon ))-Z_q (t))$.

Then, for a cylindrical functional $F$ defined on the path space 
$\mathbb{P}_{q_0} (Q)$ of the form $F(q)=f(q(t))$, for each
$t\in [0,T]$ fixed, we consider its directional derivative
\begin{equation}
\label{Malliavin_dir_dev_mfld}
D_{Z} F(q) =\int_0^T 1_{\tau <t}\big( t^q_{0 \leftarrow t} (\nabla f)(q(t))\big)^i ~\frac{d }{d\tau }
z_i (\tau)d\tau, 
\end{equation}
where $1_{\tau <t}$ denotes the characteristic function
of the open interval $(-\infty, t)$.

These Malliavin derivatives can be defined for more general functionals, 
but in this paper we only need those which are introduced above.
Notice that for $Q=\mathbb R^n$, we have $D_z F(x)=
\frac{\partial f}{\partial x^i} (x(t))z_i (t)$. That is, the directional 
derivative coincides with the one defined in the Euclidean path space 
setting.

\paragraph{The pull back of the Malliavin derivative to Euclidean space.}
We want to pull back derivatives on the path space of the manifold to the 
Euclidean path space. For this purpose, we invoke the following result  (\cite[ch. XI]{Mall1997}, 
or \cite[ch. II c]{Bi1984} 
and \cite{CrMa1998} for the case with drift):

$$
(D_{Z}F)\circ \mathcal I =D_{\zeta}(F\circ \mathcal I ),
$$
where the $\mathbb{R}^n$-valued semimartingale $\zeta$ satisfies the stochastic differential system
$$
\left\{\begin{matrix}
d\zeta(t)=\frac{d}{dt} z(t) dt -\left(\frac{1}{2} R_t +\mathcal D(t)\right) (z(t))dt -\rho(t) dW(t) \\
d\rho(t) = \Omega\left(\circ dW(t)-b(t)dt,z(t) \right).
\end{matrix}\right.
$$ 
Here, $z (t):=t^q_{0\leftarrow t} (Z_q (t))$ for the given the path 
$t\mapsto Z_q (t)\in T_{q(t)}Q$, $\Omega: TQ \times TQ \rightarrow TQ$ is the curvature tensor of the metric $g$ on $Q$,
$$
R_t = t^q_{0\leftarrow t} \circ {\rm Ricci}_{q(t)} \circ 
t^q_{t \leftarrow 0}
$$
is the  representation in the global chart  of the Ricci tensor on $Q$, and
$\mathcal D (t)= t^q_{0\leftarrow t} \circ \nabla b (t) \circ t^q_{t\leftarrow 0}$, where $b(t)$ is the time dependent 
smooth vector field on $Q$ defined by
\[
b(t) := (u(t))_Q + \frac{1}{2} \sum_{k=1}^N \nabla_{(\xi_k)_Q}
(\xi_k)_Q\,.
\]

Considering the variation $Z$ on the path space of the manifold corresponds to considering in the global chart given by the It\^o map a variation with respect to the semimartingale $\zeta$ defined above. The diffusion part of this semimartingale is given by the curvature of the manifold, which is antisymmetric; therefore this part corresponds to a rotation of the Brownian motion, which is again a Brownian motion. We neglect this rotation 
and understand that we may not be working with a fixed Brownian motion, but eventually with equivalent ones (identical in law). Therefore, our variations with respect to $Z$  will  be taken, using the global chart, with respect to directions $\bar z$ of the form
$$
\frac{d}{dt}\bar z =\frac{d}{dt} z(t)-
\left(\frac{1}{2} R +{\mathcal D}\right) (z(t))
\,.
$$
The map $z\mapsto \bar{z}$ can be inverted, through the resolvent equation
$$
\frac{d}{dt } Q_{t,s} =\left(\frac{1}{2} R+{\mathcal D} \right) Q_{t,s} ,\quad Q_{s,s}=Id
$$
and $z(t )=\int_0^t Q_{t,s} \left(\frac{d}{ds} \bar z (s)\right)d{s}$, as long as the previous resolvent equation has a solution for all $t\in [0,T]$.

We can then conclude that making an ``arbitrary" variation in the direction $Z$ in the path space of the manifold corresponds, using the global chart,  to making an ``arbitrary" variation in $\mathbb P_{x_0} (\mathbb R^n )$ in the direction  $z$ (an adapted random curve of bounded variation  with time derivative in $L^2_{\mu}$). To use a more intuitive notation we shall denote by $\delta q (t)=Z_q (t)$ a variation in $\mathbb P_{q_0}(Q)$.

All the considerations above still hold for pinned Brownian paths, namely those with a final condition $q(T)=q_T$. Then the variations are equal to zero at this final time. The corresponding sigma-algebra and filtration on the path space are the usual ones, generated by the coordinate maps and generated by the coordinate maps up to time $t$, respectively. We refer to \cite{Driver1994,ElLi2006} for more explanation.

From now on, we also assume a growth control on the vector fields $u(t)$ so that the laws of the corresponding diffusion processes $dq(t):=(\xi_k)_Q (q(t))\circ dW_t^k +(u(t))_Q (q(t))dt$, $q(0)=q_0$ are absolutely continuous with respect to the law of $d\tilde q(t):=(\xi_k)_Q (\tilde q(t))\circ dW_t^k$, $\tilde q(0)=\tilde q_0$. The standard assumption to ensure this is Novikov's condition (see, e.g., \cite{IkWat1989}); namely, that there exists a  $\mathbb R^N$-valued stochastic process $\theta(t)$ such that 
$(\xi_k)_Q \theta ^k (t)=u(t)_Q$ and a constant $0<\lambda <1$ such that  $E_{\mu} (\exp \lambda \int_0^T |\theta (t)|^2 dt) <\infty$. 
Also we assume growth control on $u$ and $b$ so that the resolvent equation above has a solution defined for all $t\in [0,T]$.

\subsection{Stochastic variational principles}

With these definitions and preparatory formulas, the following theorem 
holds.

\begin{theorem}{\rm[Stratonovich stochastic variational principle]}
\label{StratSVP-thm}$\,$
The stochastic variational principle $\delta S = 0$ with action integral 
$S$ given in equation \eqref{SVP2}, with $d\mathscr{x}_t^\alpha$ given by
\eqref{dx-def}, 
defined on diffusion processes with fixed diffusion coefficients, under the assumptions described before and such that $q(0)=q_0 , q(T)=q_T$, yields the following stochastic dynamical equations
\begin{align}
\label{SEP-eqns-thm}
\begin{aligned}
&d_tn_\alpha(\mathscr{q}(t),\mathscr{p}(t)) = 
\{n_\alpha(\mathscr{q}(t),\mathscr{p}(t))\,,\,
n_\beta(\mathscr{q}(t),\mathscr{p}(t))\} d\mathscr{x}_t^\beta 
+ \frac{\partial \ell}{\partial q^i}(u(t),\mathscr{q}(t))
A^i_\alpha(\mathscr{q}(t)) dt, \\
&d_t\mathscr{q}^i(t) = \{ \mathscr{q}^i(t),\, 
n_\beta(\mathscr{q}(t),\mathscr{p}(t))\} d\mathscr{x}_t^\beta
\,,\quad
d_t\mathscr{p}_i(t)  = \{\mathscr{p}_i(t),\, 
n_\beta(\mathscr{q}(t),\mathscr{p}(t))\} d\mathscr{x}_t^\beta 
+ \frac{\partial \ell}{\partial q^i}(u(t),\mathscr{q}(t))dt
\,,
\end{aligned}
\end{align}
for all $\alpha = 1, \ldots, \dim \mathfrak{g}$, $i = 1, \ldots, \dim Q$,
where $n:= \frac{\delta\ell}{\delta u} \in 
\mathfrak{g}^\ast$, i.e.,
$n_\alpha = \frac{\partial \ell}{\partial u^\alpha}$. Moreover, 
$n(\mathscr{q}(t),\mathscr{p}(t))=m(\mathscr{q}(t),\mathscr{p}(t))$ a.s., where $m(p_q) = \mathbf{J}_{T^*Q}(p_q) = 
m_\alpha(p_q) e^\alpha \in \mathfrak{g}^\ast$, $m_\alpha = p_i A_\alpha^i(q)$ and hence 
$m_\alpha(\mathscr{q}, \mathscr{p}) = 
\mathscr{p}_i A_\alpha^i (\mathscr{q})$.
\end{theorem}

The Poisson brackets in this formula need 
interpretation, since $n_\alpha$ depends only on the variables $u$ and 
$q^i$. First, the Poisson brackets on the right hand side of 
\eqref{SEP-eqns-thm} are taken in the sense of \eqref{standard_stochastic_PB} 
or its global version. Second, as
will be shown below, the stationarity condition $\delta S =0$ 
yields the relation $\frac{\delta \ell}{\delta u^\alpha}
(\mathscr{q}(t),\mathscr{p}(t)) = 
\mathscr{p}_i(t) A^i_\alpha(\mathscr{q}(t))$ almost surely,
which says that $n(\mathscr{q}(t),\mathscr{p}(t))=
m(\mathscr{q}(t),\mathscr{p}(t))$ almost surely and that the 
Lagrange multipliers $\mathscr{p}_i(t)$ also depend on the 
random curves $u(t)$ and $\mathscr{q}^i(t)$, as expected. By pretending now 
that the quantity $m_\alpha = p_i A^i_\alpha(q)$ depends 
on the variables $q^i$ and $p_i$, as if
they were \textit{independent} $T^*Q$-chart variables, one may compute
the Poisson brackets in \eqref{SEP-eqns-thm}
by using the derivative of the semimartingale $\{f,g\}(\mathscr{q}(t),\mathscr{p(t)})$.

\begin{remark}\rm
Upon looking back at the Stratonovich integral in the stochastic action
functional \eqref{SVP1_b}, we see that the stochasticity couples to the phase 
space variables through the \emph{momentum map} via 
the relations $\frac{\partial\ell}{\partial u^\alpha}(u(t),
\mathscr{q}(t)) = \mathscr{p}_i(t)A_\alpha^i(\mathscr{q}(t))$. 
\hfill $\lozenge$
\end{remark}

\begin{proof}
The first step in the proof of Theorem \ref{StratSVP-thm} is to take the variations of the action integral \eqref{SVP1_b}, thereby finding
the following equations, which hold almost surely,
\begin{align}
\begin{split}
\delta u(t):\quad&
\frac{\delta \ell}{\delta u^\alpha}(u(t),\mathscr{q}(t)) - \mathscr{p}_i(t) A^i_\alpha(\mathscr{q}(t)) = 0
\,,\\
\delta \mathscr{p}(t):\quad&
d_t\mathscr{q}^i(t) - A^i_\alpha(\mathscr{q}(t)) d\mathscr{x}_t^\alpha  = 0
\,,\\
\delta \mathscr{q}(t):\quad&
d_t\mathscr{p}_i(t) + \mathscr{p}_j(t) 
\frac{\partial A^j_\alpha}{\partial \mathscr{q}^i}(\mathscr{q}(t))
d\mathscr{x}_t^\alpha  - \,\frac{\partial \ell}{\partial q^i}
(u(t),\mathscr{q}(t))dt = 0 
\,,\end{split}
\label{var-eqns-proof}
\end{align}
after integrations by parts using the vanishing of the term 
$(\mathscr{p}_i(t)\, \delta \mathscr{q}^i(t))|_0^T$ at the endpoints 
in time, which follows from the assumption $\delta \mathscr{q}^i (0)= 0 
=\delta \mathscr{q}^i (T)$. 

Notice that we are not fixing a priori  a vector field $u(t)$ in this Theorem.
Nevertheless, by Novikov's condition above, the laws  of the corresponding diffusions are 
absolutely continuous with respect a fixed path space (strictly before the final time $T$) and all the 
admissible variations, as we have seen, do not depend on the form of $u(t)$. 
Admissible variations for the stochastic process 
$\mathscr{q}(\cdot)$, in the 
global chart given by the It\^o map, are variations on the corresponding path space as defined in \eqref{Malliavin_dir_dev_mfld}. Therefore one can take arbitrary (in this sense) horizontal variations of the process  $\mathscr{p}(\cdot)_{\mathscr{q}(\cdot )}$. Since also variations in all vertical directions are allowed, we can take arbitrary variations of
 $\mathscr{p}(\cdot )_{\mathscr{q}(\cdot )}$.

On the contrary, the admissible paths ($\mathscr{q}$ and 
$\mathscr{p}$) for which the action functional $S$ is defined have fixed diffusion coefficients, so that we can use Malliavin calculus on the corresponding path spaces. Since  the diffusion coefficient of $\mathscr{p}$ is not elliptic at every point, we may regularize it, deduce the corresponding dynamical equations and then pass to the limit.

In particular, we have 
$n_\alpha(\mathscr{q}(t),\mathscr{p}(t)) =
\frac{\delta \ell\ }{\delta u^\alpha}(u(t),\mathscr{q}(t))
=m_\alpha(\mathscr{q}(t),\mathscr{p}(t))$.
Therefore, taking the stochastic differential of the first equation, then using the second 
and third equations in \eqref{var-eqns-proof}, 
we get, dropping the $t$-dependence notation in the semimartingales, 
\begin{align}
d_tm_\alpha(\mathscr{q},\mathscr{p})&= 
d\left( \frac{\delta \ell\ }{\delta u^\alpha}\right) 
= d \left(\mathscr{p}_i A^i_\alpha(\mathscr{q})\right) = 
(\circ d\mathscr{p}_i) A^i_\alpha(\mathscr{q}) 
+ \mathscr{p}_i \frac{\partial A^i_\alpha}{\partial q^j}(\mathscr{q}) 
\circ  d\mathscr{q}^j
\nonumber \\
& \!\!\!\stackrel{\eqref{dx-def}}= A_\alpha^i(\mathscr{q})
\left(-\mathscr{p}_j 
\frac{\partial A_\beta^j}{\partial q^i}(\mathscr{q}) d\mathscr{x}^\beta_t  
+ \frac{\partial \ell}{\partial q^i}(u, \mathscr{q})dt\right) +\mathscr{p}_j\
\frac{\partial A_\alpha^j}{\partial q^k}(\mathscr{q})
 A_\beta^k(\mathscr{q}) d\mathscr{x}^\beta_t
\nonumber \\ 
& =\mathscr{p}_j \left(A_\beta^k(\mathscr{q})
\frac{\partial A_\alpha^j}{\partial q^k}(\mathscr{q}) -
A_\alpha^k(\mathscr{q})\frac{\partial A_\beta^j}{\partial q^k}
(\mathscr{q})\right) d\mathscr{x}^\beta_t
+ \frac{\partial \ell}{\partial q^i}(u,\mathscr{q})
A_\alpha^i(\mathscr{q}) dt
\nonumber \\
&\!\!\!\stackrel{\eqref{VFcomrel}}= 
\mathscr{p}_j c_{\beta\alpha}{}^\gamma \,\mathscr{p}_i A^j_\gamma(\mathscr{q})\,d\mathscr{x}_t^\beta
 + \frac{\partial \ell}{\partial q^{i}}A^i_\alpha(\mathscr{q}) dt
\nonumber \\ 
&= -c_{\alpha\beta}{}^\gamma m_\gamma(\mathscr{q},\mathscr{p}) 
\,d\mathscr{x}_t^\beta
 + \frac{\partial \ell}{\partial q^{i}}(u,\mathscr{q})
 A^i_\alpha(\mathscr{q}) dt
\nonumber  \\ 
& \!\!\!\stackrel{\eqref{m_alpha_beta}}= 
\{m_\alpha(\mathscr{q},\mathscr{p}), 
m_\beta(\mathscr{q},\mathscr{p})\}\,d\mathscr{x}_t^\beta + 
\frac{\partial \ell}{\partial q^{i}}(u,\mathscr{q})
A^i_\alpha(\mathscr{q}) dt
\label{var-eqns}
\end{align}
which is the first equation in \eqref{SEP-eqns-thm}.

The second equation $d_t\mathscr{q}^i = A^i_\alpha(\mathscr{q}) 
d\mathscr{x}_t^\alpha$ in
\eqref{var-eqns-proof} and the identity $\{\mathscr{q}^i, 
m_\beta(\mathscr{q},\mathscr{p})\} = A_\beta^i(\mathscr{q})$
yield the second equation in \eqref{SEP-eqns-thm}. Finally, the third
equation $ d_t\mathscr{p}_i=-\mathscr{p}_j 
\frac{\partial A^j_\alpha}{\partial q^i}(\mathscr{q})d\mathscr{x}_t^\alpha 
+\frac{\partial \ell}{\partial q^i}(u,\mathscr{q})dt$ in
\eqref{var-eqns-proof} and the identity $\{\mathscr{p}_i, 
m_\beta(\mathscr{q},\mathscr{p})\} = 
-\mathscr{p}_j \frac{\partial A_\beta^j}{\partial q^i}(\mathscr{q})$
yield the third equation in \eqref{SEP-eqns-thm}. 
\end{proof}

The first variational equation in \eqref{var-eqns-proof} captures the momentum map relation \eqref{m_alpha}, and the latter two equations in \eqref{var-eqns-proof} produce the corresponding equations in \eqref{SEP-eqns-thm}, when expressed in terms of the canonical Poisson bracket $\{\,\cdot\,,\,\cdot\,\}$ on $T^*Q$. The second equation in 
\eqref{var-eqns-proof} recovers the velocity map in \eqref{SVP2}, and the third equation determines the evolution of the dual canonical momentum variable, the Lagrange multiplier $\mathscr{p}_i$.

The penultimate equality in \eqref{var-eqns} yields the following
result.

\begin{corollary}
\label{cor_sotchastic_coadjoint_motion}
Hamilton's principle $\delta S = 0$ for the constrained action integral 
in \eqref{SVP2} recovers stochastic coadjoint motion equation in the 
following form,
\begin{align}
d_t \left(\frac{\partial \ell}{\partial u^\alpha}(u(t),
\mathscr{q}(t))\right)
=\left({\rm ad}^*_{d\mathscr{x}_t} \frac{\partial \ell}{\partial u}
(u(t),\mathscr{q}(t))\right)_\alpha
 + \frac{\partial \ell}{\partial q^{i}}(u(t),\mathscr{q}(t))
 A^i_\alpha(\mathscr{q}(t))dt, 
 \label{Hamel-eqn}
\end{align}
where $\left({\rm ad}^*_{d\mathscr{x}_t} 
 \frac{\partial \ell}{\partial u}
 (u(t),\mathscr{q}(t))\right)_\alpha
 :=
 -\,c_{\alpha\beta}{}^\gamma\frac{\partial \ell\ }{\partial u^\gamma}
 (u(t),\mathscr{q}(t))
 d\mathscr{x}_t^\beta$. 
\end{corollary}

\subsection{The stochastic equations of motion on $\mathfrak{g}^\ast 
\times Q$} 

The presence of the Poisson brackets in \eqref{SEP-eqns-thm} suggests
the existence of a Hamiltonian version of these equations. This will
be explored in detail in Section \ref{Properties-stochcoadmotion-sec}.
Here we just introduce a stochastic version of the Legendre
transform and derive certain equations on $\mathfrak{g}^\ast \times Q$
whose geometric structure will be investigated in Section 
\ref{Properties-stochcoadmotion-sec}.

In the classical deterministic case, recall that the Legendre 
transform of a Lagrangian $\mathfrak{L}:
\mathfrak{g}\rightarrow \mathbb{R}$ to a Hamiltonian $\mathfrak{H}:
\mathfrak{g}^\ast \rightarrow \mathbb{R}$, mapping the Euler-Poincar\'e 
equations $\frac{d}{dt}\frac{\delta \mathfrak{L}}{\delta \xi} = 
-\operatorname{ad}^*_\xi\frac{\delta \mathfrak{L}}{\delta \xi}$ 
to the Lie-Poisson equations $\frac{d}{dt} \mu = 
-\operatorname{ad}^*_{\frac{\delta\mathfrak{H}}{\delta \mu}}\mu$ (and, 
conversely, if the map is 
a diffeomorphism), is given by (see, e.g., 
\cite[\S 13.5, p. 437]{MaRa1999})
\[
\mu: = \frac{\delta\mathfrak{L}}{\delta \xi}, \quad \mathfrak{H}(\mu) : =
\left\langle\mu, \xi \right\rangle_\mathfrak{g} - \mathfrak{L}(\xi), 
\qquad \xi \in \mathfrak{g}, \quad \mu\in \mathfrak{g}^\ast.
\]
If the map $\mathfrak{g} \ni \xi \mapsto \mu = 
\frac{\delta\mathfrak{L}}{\delta \xi} \in \mathfrak{g}^\ast$ is a 
diffeomorphism, the Lagrangian $\mathfrak{L}$ and Hamiltonian 
$\mathfrak{H}$ given above, are called \textit{hyperregular}.
We define below a stochastic version of this Legendre transform, 
depending on a parameter, 
replacing the Lie algebra element by the stochastic
vector field \eqref{dx-def} and the element in the dual of the Lie algebra
by a semimartingale.

We proceed in the following way. We say that the Lagrangian 
$\ell: \mathfrak{g} \times Q \rightarrow \mathbb{R}$ is 
\textit{hyperregular} if the function $\ell(\cdot , q): \mathfrak{g}
\rightarrow \mathbb{R}$ is hyperregular for every $q \in Q$. We work
with hyperregular Lagrangians from now on. Define, as in Theorem 
\ref{StratSVP-thm}, $n:=n(u,q) := \frac{\delta \ell}{\delta u} \in 
\mathfrak{g}^\ast$, invert this relation for every $q \in Q$ to
get $u = u(n, q)$, and introduce the Hamiltonian function $h(n, q) :=
\left\langle n, u(n,q) \right\rangle_ \mathfrak{g} - \ell(u(n, q), q)$. 
Next, recalling that $m = \mathbf{J}_{T^*Q}(p_q) = 
p_i A_\alpha^i(q) e^\alpha\in \mathfrak{g}^\ast$, consider the function 
$h(m,q)$, i.e., we replace the first variable $n$ of $h$ by the 
expression $m$. Now, replace the  variables $(u,q) \in\mathfrak{g}\times Q$ 
by random curves $(u(t), \mathscr{q}(t))$ and form a
semimartingale 
$\mathscr{h}(m(\mathscr{q}(t), \mathscr{p}(t)),\mathscr{q}(t))$ (which 
corresponds to a stochastic Hamiltonian, as explained in subsection 
\ref{sec_stoch_ham_equ}), by imposing,
in analogy with the deterministic case, the stochastic
derivative of the semimartingale 
$\mathscr{h}(m(\mathscr{q}(t), \mathscr{p}(t)),\mathscr{q}(t))$ to
equal 
\begin{equation}
\label{LegXform1}
d_t\mathscr{h} (m(\mathscr{q}(t), \mathscr{p}(t)),\mathscr{q}(t))= 
\mathscr{p}_i(t) A_\alpha^i(\mathscr{q}(t)) d\mathscr{x}_t^\alpha  - \ell(u(t),\mathscr{q}(t))\,dt
\stackrel{\eqref{m_alpha}}=  m_\alpha (\mathscr{q}(t), \mathscr{p}(t))\,  
d\mathscr{x}_t^\alpha - \ell(u(t),\mathscr{q}(t)) dt, 
\end{equation}
where  $d\mathscr{x}_t^\alpha = u^\alpha dt + \xi_k^\alpha\circ dW^k_t$.

This semimartingale is of the form \eqref{semimform}, namely
\[
d\mathscr{h} (m(\mathscr{q}(t), \mathscr{p}(t)),\mathscr{q}(t)) =   
(h^1)_\alpha(m(\mathscr{q}(t), \mathscr{p}(t)),\mathscr{q}(t)) 
\xi_k^\alpha \circ dW_t^k +
h^2(m(\mathscr{q}(t), \mathscr{p}(t)),\mathscr{q}(t))dt.
\]

In agreement with  our previous definitions, we shall use the notation
\begin{equation}
\label{semim_der1}
d_t \left(\frac{\partial \mathscr{h}}{\partial m_\beta}\right)
(m(\mathscr{q}(t), \mathscr{p}(t)),\mathscr{q}(t)) = 
\frac{\partial (h^1)_\alpha}{\partial m_\beta} (m(\mathscr{q}(t), \mathscr{p}(t)),\mathscr{q}(t)) \xi_k^\alpha \circ dW_t^k +
 \frac{\partial h^2}{\partial m_\alpha} (m(\mathscr{q}(t), 
 \mathscr{p}(t)),\mathscr{q}(t))dt,
\end{equation}
\begin{equation}
\label{semim_der2}
d_t \left( \frac{\partial \mathscr{ h}}{\partial q^j}\right)
(m(\mathscr{q}(t), \mathscr{p}(t)),\mathscr{q}(t)) = 
\frac{\partial (h^1)_\alpha}{\partial q^j}(m(\mathscr{q}(t), 
\mathscr{p}(t)),\mathscr{q}(t))\xi_k^\alpha \circ dW_t^k +
 \frac{\partial h^2}{\partial q^j} (m(\mathscr{q}(t), \mathscr{p}(t)),\mathscr{q}(t))dt.
\end{equation}

\begin{theorem}
\label{m_q_thm}
The stochastic variational principle $\delta S=0$, with action
integral defined in \eqref{SVP2} and semimartingale 
$\mathscr{h}(m(\mathscr{q}(t), \mathscr{p}(t)),\mathscr{q}(t))$ introduced above, implies the equations
\begin{equation}
\label{Hamel}
\begin{aligned}
\!\!\!\!
d_tm_\alpha(\mathscr{q}(t), \mathscr{p}(t))  
&=  \{m_\alpha(\mathscr{q}(t), \mathscr{p}(t))\,,\, 
m_\beta(\mathscr{q}(t), \mathscr{p}(t)) \}
d \mathscr{x}_t^\beta
-  A^j_\alpha(\mathscr{q}(t)) \frac{\partial \mathscr{h}}{\partial q^j}
(m(\mathscr{q}(t), \mathscr{p}(t)),\mathscr{q}(t)) \,dt
\\&=  \{m_\alpha(\mathscr{q}(t), \mathscr{p}(t)),\, 
m_\beta(\mathscr{q}(t), \mathscr{p}(t)) \}\circ 
d_t\left( \frac{\partial \mathscr{h}}{\partial m_\beta}\right)
(m(\mathscr{q}(t), \mathscr{p}(t)),\mathscr{q}(t))\\
& \qquad 
+ \{m_\alpha(\mathscr{q}(t), \mathscr{p}(t)),\, \mathscr{q}^j \}
\frac{\partial \mathscr{h}}{\partial q^j}
(m(\mathscr{q}(t), \mathscr{p}(t)),\mathscr{q}(t))dt 
\,,\\ 
d_t\mathscr{q}^i (t) &= 
\{ \mathscr{q}^i ,m_\beta(\mathscr{q}(t), \mathscr{p}(t)) \} d\mathscr{x}_t^\beta = A^i_\beta(\mathscr{q}(t)) d\mathscr{x}_t^\beta=  
A^i_\beta(\mathscr{q}(t)) \circ d_t \left(\frac{\partial \mathscr{h}}{\partial m_\beta}\right)(m(\mathscr{q}(t), 
\mathscr{p}(t)),\mathscr{q}(t)),
\end{aligned}
\end{equation}
with the convention that the Poisson brackets are computed as in Theorem
\ref{StratSVP-thm}.
\end{theorem}

\begin{proof} In the computations below, we shall drop the notational $t$-dependence of the semimartingales. By Theorem \ref{StratSVP-thm}, we know that $n(\mathscr{q}, \mathscr{p}) = m(\mathscr{q}, \mathscr{p})$ a.s. and that \eqref{SEP-eqns-thm} hold.
Next, we take the differential of condition 
\eqref{LegXform1}. Thus, if $\delta m_\alpha(\mathscr{q}, \mathscr{p})$ 
and $\delta \mathscr{q}^i$ are arbitrary variations (namely random 
curves of bounded variation in $t$)
of the semimartingales $m_\alpha(\mathscr{q}, \mathscr{p})$ and 
$\mathscr{q}^i$, respectively, we get 
\begin{align*}
&d \left(\frac{\partial  \mathscr{h}}{\partial m_\alpha} \right)
(m(\mathscr{q}, \mathscr{p}), \mathscr{q}) \delta m_\alpha
(\mathscr{q}, \mathscr{p})
+ d \left(\frac{\partial  \mathscr{h}}{\partial q^i} \right)
(m(\mathscr{q}, \mathscr{p}), \mathscr{q}) \delta \mathscr{q}^i \\
& \qquad =
\delta m_\alpha(m(\mathscr{q}, \mathscr{p})) d\mathscr{x}_t^\alpha 
+ \delta u^\alpha
\left(m_\alpha(m(\mathscr{q}, \mathscr{p})) - 
\frac{\partial \ell}{\partial u^\alpha}(u, \mathscr{q})\right)dt 
- \delta \mathscr{q}^i  \left(\frac{\partial \ell}{\partial q^i}\right)
(u, \mathscr{q})dt,
\end{align*}
which is equivalent a.s. to
\begin{equation}
\label{partial-derivs2}
\begin{aligned}
&m_\alpha(\mathscr{q}, \mathscr{p}) - 
\frac{\partial \ell}{\partial u^\alpha}(u, \mathscr{q})=0
\,,\qquad
d \left(\frac{\partial  \mathscr{h}}{\partial m_\alpha}\right)
(m(\mathscr{q}, \mathscr{p}), \mathscr{q}) = d\mathscr{x}_t^\alpha\,,
\quad\hbox{and}\\\
&d \left( \frac{\partial  \mathscr{h}}{\partial q^i}\right)
(m(\mathscr{q}, \mathscr{p}), \mathscr{q}) = 
-  \left( \frac{\partial \ell}{\partial q^i}\right)
(u, \mathscr{q}) \,dt\,.
\end{aligned}
\end{equation}
Note that the first equation implies, as expected from Theorem
\ref{StratSVP-thm}, the a.s. equality of the semimartingales 
$n(\mathscr{q}, \mathscr{p})=m(\mathscr{q}, \mathscr{p})$. 

Using the identities \eqref{partial-derivs2}
and the equations \eqref{SEP-eqns-thm}, we compute 
$dm_\alpha(\mathscr{q}, \mathscr{p})$
and $d\mathscr{q}^i$ to find,
\begin{align*}
\begin{split}
dm_\alpha (\mathscr{q}, \mathscr{p}) 
&=  \{m_\alpha (\mathscr{q}, \mathscr{p}),\, 
m_\beta (\mathscr{q}, \mathscr{p})\} d \mathscr{x}_t^\beta
-  A^j_\alpha(\mathscr{q}) 
\frac{\partial  \mathscr{h}}{\partial q^j} 
((\mathscr{q}, \mathscr{p}), \mathscr{q})dt \\
&\!\!\!\stackrel{\eqref{partial-derivs2}}=  
\{m_\alpha (\mathscr{q}, \mathscr{p}),\, 
m_\beta (\mathscr{q}, \mathscr{p}) \}\circ d\left( 
\frac{\partial \mathscr{h}}{\partial m_\beta}
((\mathscr{q}, \mathscr{p}), \mathscr{q})\right)
+ \{m_\alpha (\mathscr{q}, \mathscr{p}),\, \mathscr{q}^j \}
\frac{\partial  \mathscr{h}}{\partial q^j}
((\mathscr{q}, \mathscr{p}), \mathscr{q})dt 
\,,\\ 
d\mathscr{q}^i  &= 
\{ \mathscr{q}^i ,m_\beta (\mathscr{q}, \mathscr{p})\} d\mathscr{x}_t^\beta = A^i_\beta(\mathscr{q}) d\mathscr{x}_t^\beta\stackrel{\eqref{partial-derivs2}}=  
A^i_\beta(\mathscr{q}) \circ d \left( 
\frac{\partial\mathscr{h}}{\partial m_\beta}\right)
((\mathscr{q}, \mathscr{p}), \mathscr{q})
\,,
\end{split}
\end{align*}
which recover equations \eqref{Hamel}.
\end{proof}

\begin{remark} \rm
The defining relation $m_\alpha(\mathscr{q}, \mathscr{p}) =
\mathscr{p}_i A_\alpha^i (\mathscr{q})$ and the second equation in 
\eqref{SEP-eqns-thm} imply
\begin{align*}
dm_\alpha (\mathscr{q}, \mathscr{p})&=
 d \left(\mathscr{p}_i A^i_\alpha(\mathscr{q})\right) = 
(\circ d\mathscr{p}_i) A^i_\alpha(\mathscr{q}) 
+ \mathscr{p}_i \frac{\partial A^i_\alpha}{\partial q^j}
(\mathscr{q}) \circ d\mathscr{q}^j
\nonumber \\
&\!\!\!\stackrel{\eqref{SEP-eqns-thm}}= 
(\circ d\mathscr{p}_i) A^i_\alpha(\mathscr{q}) 
+ \mathscr{p}_i \frac{\partial A^i_\alpha}{\partial q^j}(\mathscr{q}) 
A_\beta^j (\mathscr{q})d\mathscr{x}_t^\beta\,.
\end{align*}
By Theorem \ref{StratSVP-thm}, we know that $n(\mathscr{q},\mathscr{p}) 
= m(\mathscr{q}, \mathscr{p})$
a.s. and hence the first equation in \eqref{SEP-eqns-thm} yields 
\begin{align*}
dm_\alpha (\mathscr{q}, \mathscr{p})&= 
\{m_\alpha (\mathscr{q}, \mathscr{p}), 
m_\beta (\mathscr{q}, \mathscr{p})\}\,d\mathscr{x}_t^\beta + 
\frac{\partial \ell}{\partial q^{i}}(u, \mathscr{q})
A^i_\alpha(\mathscr{q}) dt \\
& \!\!\!\stackrel{\eqref{var-eqns}}=
\mathscr{p}_j \left(A_\beta^k(\mathscr{q})
\frac{\partial A_\alpha^j}{\partial q^k}(\mathscr{q}) -
A_\alpha^k(\mathscr{q})\frac{\partial A_\beta^j}{\partial q^k}
(\mathscr{q})\right) \,
d\mathscr{x}^\beta_t
+ \frac{\partial \ell}{\partial q^i}(u, \mathscr{q})
A_\alpha^i(\mathscr{q}) dt.
\end{align*}
Comparing these two expressions, we conclude the a.s. equality
\begin{equation}
\label{stochastic_strange_equ}
(\circ d\mathscr{p}_i) A^i_\alpha(\mathscr{q}) =
-\mathscr{p}_jA_\alpha^k(\mathscr{q})
\frac{\partial A_\beta^j}{\partial q^k}(\mathscr{q})
d\mathscr{x}^\beta_t
+ \frac{\partial \ell}{\partial q^i}(u, \mathscr{q}) 
A^i_\alpha(\mathscr{q}) dt 
= \left(-\mathscr{p}_j\frac{\partial A_\beta^j}{\partial q^i}(\mathscr{q})
d\mathscr{x}^\beta_t
+ \frac{\partial \ell}{\partial q^i}(u, \mathscr{q})  dt 
\right)A^i_\alpha(\mathscr{q}).
\end{equation}
Note that this identity is clearly implied by the third equation in
\eqref{SEP-eqns-thm}.
\hfill $\lozenge$
\end{remark}

\section{It\^o formulation of stochastic coadjoint motion}
\label{Ito-form-coadmotion-sec}

As before, $t \mapsto W_t^k(\omega)$, $k = 1, \ldots N$, 
$\omega \in\Omega$, are $N$ independent real-valued Brownian motions 
and $\xi_1, \ldots, \xi_N \in \mathfrak{g}$. For each $\xi_k$, $k=1,\ldots, N$, define the Hamiltonian vector field $X_{\xi_k} \in \mathfrak{X}(T^*Q)$ by
\begin{align}
X_{\xi_k} := \{\,\cdot\,,\, m_\alpha(q, p) \xi^\alpha_k \} 
\stackrel{\eqref{var-eqns-proof}}= 
\{\,\cdot\,,\, p_i A^i_\alpha(q) \xi^\alpha_k \}  
\,,\label{Xsubxi-def-thm}
\end{align}
i.e., $X_{\xi_k}$ is the Hamiltonian vector field on $T^*Q$ with Hamiltonian
function $T^*Q\ni p_q \mapsto \left\langle\mathbf{J}_{T^*Q}(p_q), 
\xi_k \right\rangle_\mathfrak{g} \in\mathbb{R}$, $k=1, \ldots, N$. As in the previous
section, we denote interchangeably points in $T^*Q$ by $p_q$ or $(q, p)$.  

Define the operator on semimartingales of the form $f(\mathscr{q}(t), 
\mathscr{p}(t))$, where $f \in C^{\infty}(T^*Q)$ by
\begin{equation}
\label{def_semi_X_xi_k}
(\mathscr{X}_{\xi_k}f)(\mathscr{q}(t), \mathscr{p}(t)): 
= \left\{f(\mathscr{q}(t), \mathscr{p}(t)),\, m_\alpha(\mathscr{q}(t), 
\mathscr{p}(t))\xi_k^\alpha \right\} =  \left\{f(\mathscr{q}(t), 
\mathscr{p}(t)) ,\,\mathscr{p}_i(t)
A_\alpha^i(\mathscr{q}(t))\xi_k^i \right\},
\end{equation}
where the brackets in the right hand side are those of semimartingales, 
as in \eqref{standard_stochastic_PB}.
Note that the result of the operation $(\mathscr{X}_{\xi_k}f)
(\mathscr{q}(t), \mathscr{p}(t))$, defined in \eqref{def_semi_X_xi_k},
is again a semimartingale.

In analogy with \eqref{dx-def}, define the It\^o stochastic element 
$d\widehat{\mathscr{x}}_t^\beta\in \mathfrak{g}$  by
\begin{align}
d\widehat{\mathscr{x}}_t^\alpha := u^\alpha(t) dt + \xi_k^\alpha dW^k_t
\,.\label{dxhat-def-thm}
\end{align}
The It\^o stochastic Hamiltonian vector field 
$\mathscr{X}_{d\widehat{\mathscr{x}}_t}$  
is also defined  by the Poisson bracket operation 
\begin{align}
(\mathscr{X}_{d\widehat{\mathscr{x}}_t} f) 
(\mathscr{q}(t), \mathscr{p}(t)):= 
\{f(\mathscr{q}(t), \mathscr{p}(t)), \,
m_\beta(\mathscr{q}(t), \mathscr{p}(t)) \} 
d\widehat{\mathscr{x}}_t^\beta
:= \{f(\mathscr{q}(t), \mathscr{p}(t)),\,\mathscr{p}_i(t)
A^i_\beta(\mathscr{q}(t))\} 
d\widehat{\mathscr{x}}_t^\beta
\,,\label{X-PBop-thm}
\end{align}
for any $f \in C^{\infty}(T^*Q)$,
where, again, the brackets in the right hand side are those of 
semimartingales, \eqref{standard_stochastic_PB}. The result of the operation 
$(\mathscr{X}_{d\widehat{\mathscr{x}}_t}f)(\mathscr{q}(t),\mathscr{p}(t))$, 
defined in \eqref{X-PBop-thm}, is again a semimartingale. In both
\eqref{def_semi_X_xi_k} and \eqref{X-PBop-thm}, in agreement with
the conventions in Section \ref{VP-coadmotion-sec}, we define the
semimartingales $\mathscr{q}^i(t):= q^i(\mathscr{q}(t), \mathscr{p}(t))$ and 
$\mathscr{p}_i(t):=p_i(\mathscr{q}(t), \mathscr{p}(t))$.

With these notations, we have the following result.

\begin{corollary}{\rm[It\^o stochastic variational conditions]}
\label{Ito-SVP-corollary}
\label{Lemma-m-eqn-peakon}
The corresponding It\^o forms of the Stratonovich stochastic variational equations \eqref{SEP-eqns-thm} are given by 
\begin{align}
\!\!
d_t m_\alpha(\mathscr{q}(t), \mathscr{p}(t)) &= 
\left(\mathscr{X}_{d\widehat{\mathscr{x}}_t} m_\alpha\right)
(\mathscr{q}(t), \mathscr{p}(t))
 + \frac{1}{2}{\sum_{k=1}^N}
 \left(\mathscr{X}_{\xi_k}(\mathscr{X}_{\xi_k} m_\alpha)\right)
 (\mathscr{q}(t), \mathscr{p}(t))dt + 
 \frac{\partial \ell}{\partial q^i}(u(t), \mathscr{q}(t))
 A_\alpha^i(\mathscr{q}(t))dt
\,,\\
d_t\mathscr{q}^i(t) &= \mathscr{X}_{d\widehat{\mathscr{x}}_t} 
\mathscr{q}^i(t) + \frac{1}{2}{\sum_{k=1}^N}
  \mathscr{X}_{\xi_k}(\mathscr{X}_{\xi_k} \mathscr{q}^i )(t)\,dt
\,,\\
d_t\mathscr{p}_i(t)  &= \mathscr{X}_{d\widehat{\mathscr{x}}_t} 
\mathscr{p}_i (t) + \frac{1}{2} {\sum_{k=1}^N}
\mathscr{X}_{\xi_k}(\mathscr{X}_{\xi_k} \mathscr{p}_i )(t)\,dt + \frac{\partial \ell}{\partial q^i}(u(t), \mathscr{q}(t))dt
\,.
\label{Ito-EP-eqns-X-thm}
\end{align}
\end{corollary}

\begin{remark}\rm
Remarkably, the It\^o interpretation for the coadjoint dynamics of the 
momentum map defined by $m_\alpha(\mathscr{q}(t), \mathscr{p}(t)):=
\frac{\partial \ell}{\partial u^\alpha}(u(t), \mathscr{q}(t)) = 
\mathscr{p}_i(t) A^i_\alpha(\mathscr{q}(t))$ has the 
same double bracket structure as the individual equations for the phase 
space variable $(q,p)$. Several perspectives of how this preservation 
of structure in Corollary \ref{Lemma-m-eqn-peakon} occurs, can be seen by considering three different direct proofs of it. \hfill $\lozenge$
\end{remark}

\noindent \textit{First proof.} In all the proofs below, we ignore
the $t$-dependence notation on the semimartingales.
The first proof of Corollary \ref{Lemma-m-eqn-peakon} begins by 
streamlining the notation in the Stratonovich stochastic equations  
\eqref{SEP-eqns-thm} of Theorem \ref{StratSVP-thm}, to write them simply as 
\begin{align}
dm_\alpha(\mathscr{q}, \mathscr{p}) = (\mathscr{X}_{d\mathscr{x}_t}m_\alpha)
 (\mathscr{q}, \mathscr{p}) + 
\frac{\partial \ell}{\partial q^i}(u, \mathscr{q})
A_\alpha^i(\mathscr{q})\,dt
\,,\quad
d\mathscr{q}^i = \mathscr{X}_{d\mathscr{x}_t} \mathscr{q}^i
\,,\quad
d\mathscr{p}_i  = \mathscr{X}_{d\mathscr{x}_t} \mathscr{p}_i + \frac{\partial \ell}{\partial q^i}(u, \mathscr{q})dt
\,,
\label{Strat-EP-eqns-X-proof}
\end{align}
in terms of the following Poisson bracket operator (analogous to \eqref{X-PBop-thm})
\begin{align}
\mathscr{X}_{d\mathscr{x}_t} := \{\,\cdot\,,\,
m_\beta(\mathscr{q}, \mathscr{p}) \}  d\mathscr{x}_t^\beta
:= \{\,\cdot\,,\, \mathscr{p}_iA^i_\beta(\mathscr{q})\} d\mathscr{x}_t^\beta
\,.\label{X-PBop}
\end{align}

We want to write these expressions in It\^o form. 
For this, we recall  It\^o's formula: if $\mathscr{X}_t$ is a 
semimartingale with regular coefficients and $f$ a smooth function 
(c.f., for example, \cite{IkWat1989}), then
\begin{equation}
\label{ito_formula}
d_tf(\mathscr{X}_t)=\partial_i f (\mathscr{X}_t )\circ d_t\mathscr{X}_t^i =
\partial_i f (\mathscr{X}_t ) d_t\mathscr{X}_t^i +\frac{1}{2} \partial^2_{i,j} f(\mathscr{X}_t )d_t\mathscr{X}^i_t . d_t\mathscr{X}^j_t\,.
\end{equation} 

The corresponding It\^o forms of the latter Stratonovich expressions in 
\eqref{SEP-eqns-thm} are then written equivalently as
\begin{align}
\begin{split}
&d\mathscr{q}^i = \mathscr{X}_{d\widehat{\mathscr{x}}_t} \mathscr{q}^i +  
\frac{1}{2}\sum_{k=1}^N \mathscr{X}_{\xi_k}(\mathscr{X}_{\xi_k} \mathscr{q}^i )\,dt
\,,\\
&d\mathscr{p}_i  = \mathscr{X}_{d\widehat{\mathscr{x}}_t} \mathscr{p}_i + \frac{1}{2}
\sum_{k=1}^N \mathscr{X}_{\xi_k}(\mathscr{X}_{\xi_k} \mathscr{p}_i )\,dt + 
\frac{\partial \ell}{\partial q^i}(u, \mathscr{q})dt
\,.
\end{split}
\label{Ito-EP-eqns-X}
\end{align}

We prove the first relation in \eqref{Ito-EP-eqns-X}, as the other one is similarly derived. To simplify notation, we write simply $A_\alpha^i$ instead
of $A_\alpha^i(\mathscr{q})$. Recall that
$$
d\mathscr{q}^i =A_{\alpha}^i u^{\alpha} dt + 
A_{\alpha}^i  \xi_k^{\alpha} \circ dW_t^k
$$
By It\^o's formula \eqref{ito_formula}, the only term which is 
not of bounded variation in the expression for $dA_{\alpha}^i$
is equal to
$ \frac{ \partial }{\partial q^j} (A_{\alpha}^i ) A_{\beta}^j
\xi_k^{\beta} dW_t^k$
and we conclude that 
$$
d(A_{\alpha}^i \xi_k^{\alpha} \xi_k^{\alpha} ). dW_t^k =
\sum_{k=1}^N \frac{ \partial A_\alpha^i}{\partial q^j} 
A_\beta^j\xi_k^\beta \xi_k^\alpha dt.
$$
Thus, \eqref{ito_formula} yields
$$
d\mathscr{q}^i =A_\alpha^i u^\alpha dt + (A_\alpha^i \xi_k^\alpha ) dW_t^k +
\frac{1}{2}\sum_{k=1}^N \frac{ \partial A_\alpha^i}{\partial q^j} 
A_\beta^j \xi_k^\beta \xi_k^\alpha dt,
$$
which is the expanded version of the first equation in 
\eqref{Ito-EP-eqns-X}.

Having introduced this streamlined notation for $d\mathscr{q}^i$ and $\mathscr{dp}_i$ in 
the It\^o equations \eqref{Ito-EP-eqns-X}, we calculate the It\^o 
equation for the components of the momentum map 
$m_\alpha(\mathscr{q}, \mathscr{p}):=\mathscr{p}_iA^i_\alpha( \mathscr{q})$, 
by using the It\^o rule for the 
derivative of a product of a pair of It\^o semimartingales, 
$\mathscr{X}$ and $\mathscr{Y}$, given by 
\begin{align}
d(\mathscr{X}\mathscr{Y}) = \mathscr{X}d\mathscr{Y} + \mathscr{Y}d\mathscr{X} + d\mathscr{X}.d\mathscr{Y}
\,,\quad\hbox{and for}\quad
d\mathscr{X}=\mathscr{\sigma} dW \,,\quad d\mathscr{Y}=\tilde{\mathscr{\sigma}} dW
\quad\hbox{we have}\quad
d\mathscr{X}.d\mathscr{Y} = \mathscr{\sigma}\tilde{\mathscr{\sigma}}dt
\,,\label{Ito-product-rule}
\end{align}
where $d\mathscr{X}.d\mathscr{Y}$ is the co-variation, or It\^o contraction.
According to the It\^o product rule, the It\^o contraction in computing 
$dm_\alpha(\mathscr{q}, \mathscr{p}) = d(\mathscr{p}_iA^i_\alpha (\mathscr{q}))$ from equation \eqref{Ito-EP-eqns-X} is 
\begin{align}
d\mathscr{p}_i \,.\, dA^i_\alpha 
= {\sum_{k=1}^N}
(\mathscr{X}_{\xi_k}\mathscr{p}_i)\,.\,(\mathscr{X}_{\xi_k}
A^i_\alpha(\mathscr{q}))
\,.\label{Ito-deriv-rule}
\end{align}
Indeed, this It\^o contraction  expression comes from the fact that the martingale parts of the processes
$\mathscr{p}_i$  and $A_\alpha^i (\mathscr{q})$ are given, respectively, by
$$
d\mathscr{p}_i \simeq-\mathscr{p}_j 
\frac{\partial A_{\alpha}^j }{\partial q^i}(\mathscr{q})
\xi_k^{\alpha} dW_t^k
\qquad \text{and} \qquad dA_\alpha^i (\mathscr{q})\simeq\frac{\partial A_{\alpha}^i }{\partial q^j} (\mathscr{q}) 
A_{\beta}^j (\mathscr{q}) \xi_k^{\beta} dW_t^k \,
$$ 
where $\Phi \simeq \Psi $ means that $\Phi - \Psi$ is a process of 
bounded variation.

Remarkably, this It\^o contraction 
\eqref{Ito-deriv-rule} turns out to be exactly 
what we need to show by direct calculation from \eqref{Ito-EP-eqns-X} that 
\begin{align}
\begin{split}
dm_\alpha(\mathscr{q}, \mathscr{p}) &= 
\left(\mathscr{X}_{d\widehat{\mathscr{x}}_t} m_\alpha\right)
(\mathscr{q}, \mathscr{p}) + 
\frac{1}{2} {\sum_{k=1}^N}
\left(\mathscr{X}_{\xi_k}(\mathscr{X}_{\xi_k} m_\alpha) \right)
(\mathscr{q}, \mathscr{p})dt + 
 \frac{\partial \ell}{\partial q^i}(u, \mathscr{q})
 A_\alpha^i(\mathscr{q})dt
\\&=
\{m_\alpha(\mathscr{q}, \mathscr{p}),\,
\mathscr{m}_\beta(\mathscr{q}, \mathscr{p})\}
d\widehat{\mathscr{x}}_t^\beta 
+ \frac{1}{2}\{ \{m_\alpha(\mathscr{q}, \mathscr{p}),\,
m_\beta(\mathscr{q}, \mathscr{p})\}\,,\, m_\gamma(\mathscr{q}, 
\mathscr{p})\}
{\sum_{k=1}^N}
\xi_k^\beta\xi_k^\gamma\,dt \\
& \qquad + 
\frac{\partial \ell}{\partial q^i}(u, \mathscr{q})
A_\alpha^i(\mathscr{q}) \,dt
\,.
\end{split}
\label{Ito-EP-momap-eqn-X}
\end{align}
In the direct calculation, the It\^o contraction is  cancelled by a cross term arising from applying the second-order derivative operator 
$\frac12 \mathscr{X}_{\xi_k}(\mathscr{X}_{\xi_k} \cdot\,)$ from equation \eqref{Ito-EP-eqns-X} to the quadratic product 
$m_\alpha(\mathscr{q}, \mathscr{p})=\mathscr{p}_iA^i_\alpha(\mathscr{q})$. This completes the first proof of Corollary \ref{Ito-SVP-corollary}. 
\hfill $\square$
\medskip

\noindent \textit{Second proof of the first equation in 
\eqref{Ito-EP-eqns-X-thm}.} In the statement of the Corollary 
\ref{Ito-SVP-corollary}, the first It\^o equation in 
\eqref{Ito-EP-eqns-X-thm} may also be verified by an even more direct 
calculation than in \eqref{Ito-EP-momap-eqn-X}, as follows. To
simplify notations in the computations below, we again temporarily
suppress the dependence of $A_\alpha^i$ on the semimartingale $\mathscr{q}$, 
of $\frac{\partial \ell}{\partial q^i}$ on the semimartingales 
$(u, \mathscr{q})$, and of $m_\alpha(\mathscr{q}, \mathscr{p})$
on the semimartingales $(\mathscr{q}, \mathscr{p})$. We also suppress
the $k$-index.

By equations 
\eqref{m_alpha_beta}, \eqref{SEP-eqns-thm}, the definition 
\eqref{dxhat-def-thm} of $d\widehat{\mathscr{x}}_t^\beta$, the It\^o 
product rule in \eqref{Ito-product-rule}, and Theorem 1, we have, 
\begin{align}
\begin{split}
dm_\alpha &=  \{m_\alpha\,,\,m_\beta\} d\mathscr{x}_t^\beta + 
\frac{\partial \ell}{\partial q^i}A_\alpha^i\,dt
\,,\\
d(\mathscr{p}_i A^i_\alpha)  & = -\, \mathscr{p}_j [\,A_\alpha\,,\, A_\beta\,]^j 
d\widehat{\mathscr{x}}_t^\beta 
- \frac12\, d\Big( \mathscr{p}_j [\,A_\alpha\,,\, A_\beta\,]^j \,\xi^\beta \Big).dW_t 
+ \frac{\partial \ell}{\partial q^i}A_\alpha^i\,dt
\,,\end{split}
\label{Ito-dm-calc1}
\end{align}
where we have substituted the momentum map definition $m_\alpha=
\mathscr{p}_i A^i_\alpha$ from equation \eqref{m_alpha}.
We have,
\begin{align}
\begin{split}
d\Big( \mathscr{p}_j [\,A_\alpha\,,\, A_\beta\,]^j \,\xi^\beta \Big).\,dW_t  
& = \Big((d \mathscr{p}_j) [\,A_\alpha\,,\, A_\beta\,]^j + \mathscr{p}_j d[\,A_\alpha\,,\, 
A_\beta\,]^j \Big)\xi^\beta . dW_t \\
&= \Big((d \mathscr{p}_j) [\,A_\alpha\,,\, A_\beta\,]^j 
+ \mathscr{p}_j \frac{\partial}{\partial q^l}[\,A_\alpha\,,\, A_\beta\,]^j\,
(d\mathscr{q}^l) \Big) \xi^\beta . dW_t  \\ 
&\!\!\!\!\stackrel{\eqref{var-eqns-proof}}=  \Big(
-\,\mathscr{p}_l \frac{\partial A^l_\gamma}{\partial q^j} [\,A_\alpha\,,\, 
A_\beta\,]^j \xi^\gamma
+ \mathscr{p}_j \frac{\partial}{\partial q^l}[\,A_\alpha\,,\, A_\beta\,]^j\, 
A^l_\gamma\, \xi^\gamma  \Big)\xi^\beta dt \\ 
&=  \Big(-\,\mathscr{p}_j  [\,A_\alpha\,,\, A_\beta\,]^l   
\frac{\partial A^j_\gamma}{\partial q^l}
+\mathscr{p}_j A^l_\gamma \frac{\partial}{\partial q^l}[\,A_\alpha\,,\, 
A_\beta\,]^j  \Big)\xi^\beta\xi^\gamma dt \\ 
&=  \mathscr{p}_j   \big[\,A_\gamma \,,\, [ A_\alpha\,,\, A_\beta\,]\,\big]^j\,
\xi^\beta\xi^\gamma dt \\ 
&\!\! \stackrel{\eqref{m_alpha_beta}}=  
\big\{ \mathscr{m}_\gamma\,,\,\{ \mathscr{m}_\alpha\,,\, \mathscr{m}_\beta \}\big\} \,
\xi^\beta\xi^\gamma dt
\,.\end{split}
\label{Ito-dm-calc3}
\end{align}
Consequently, we may write the entire equation \eqref{Ito-dm-calc1} as 
(reinstating the $k$-indices)
\begin{align}
\begin{split}
dm_\alpha &=  \{m_\alpha\,,\, m_\beta\} d\widehat{\mathscr{x}}_t^\beta
- \frac{1}{2}\big\{m_\gamma\,,\,\{m_\alpha\,,\,m_\beta \}\big\} \,
\left(\sum_{k=1}^N \xi_k^\beta\xi_k^\gamma\right) dt + 
\frac{\partial \ell}{\partial q^i}A_\alpha^i\,dt   \\ 
&= \{m_\alpha\,,\, m_\beta\} (u^\beta dt + \xi_k^\beta dW^k_t)
+ \frac{1}{2}\big\{ \{m_\alpha\,,\,m_\beta \} \,,\,m_\gamma\big\} \,
\left(\sum_{k=1}^N \xi_k^\beta\xi_k^\gamma \right)dt + 
\frac{\partial \ell}{\partial q^i}A_\alpha^i\,dt   \\ 
&= \mathscr{X}_{d\widehat{\mathscr{x}}_t} m_\alpha + \frac{1}{2} 
\left(\sum_{k=1}^N \mathscr{X}_{\xi_k}(\mathscr{X}_{\xi_k}m_\alpha)\right)\,dt + 
\frac{\partial \ell}{\partial q^i}A_\alpha^i\,dt
\quad \text{using the notations}\quad \eqref{Xsubxi-def-thm}\quad 
\text{and}\quad   \eqref{X-PBop} 
\,,\end{split} 
\label{Ito-dm-calc4}
\end{align}
in agreement with the first equation in \eqref{Ito-EP-eqns-X-thm}.
\hfill $\square$
\medskip

\noindent \textit{Third proof of the first equation in 
\eqref{Ito-EP-eqns-X-thm}.} 
Let $\mathscr{\eta} : [0,T]\rightarrow \mathfrak{g}$ be an arbitrary random curve of bounded variation. We begin the third proof by computing
\begin{align*}
\begin{split}
\left< d\left( \frac{\delta \ell }{\delta u}
(u, \mathscr{q})\right),\eta \right> 
& = d\Big( \mathscr{p}_i \,A_\alpha^i (\mathscr{q})\Big) \eta^\alpha 
\\
&=\left(A_\alpha^i(\mathscr{q}) \eta^\alpha \right)\circ d\mathscr{p}_i +
\left(\mathscr{p}_i \frac{\partial A_\alpha^i}{\partial q^j}(\mathscr{q})
 \eta^\alpha \right) \circ d\mathscr{q}^j
\\
& \!\!\!\!
\mathrel{\operatorname*{=}_{\eqref{var-eqns-proof}}^{\eqref{dx-def}}}
-\,\mathscr{p}_j \frac{\partial A_\alpha^j }{\partial q^i}
(\mathscr{q}) u^\alpha A_\beta^i(\mathscr{q})
\eta^\beta dt
- \left(\mathscr{p}_j \frac{\partial A_\alpha^j }{\partial q^i}
(\mathscr{q}) \xi_k^\alpha  
A_\beta^i(\mathscr{q})  \eta^\beta \right)\circ dW_t^k 
\\
& \qquad +A_\beta^i(\mathscr{q}) \eta^\beta 
\frac{\partial \ell}{\partial q^i}(u,\mathscr{q}) dt
\\
&\qquad + \mathscr{p}_i \frac{\partial A_\alpha^i }{\partial q^j}
(\mathscr{q}) \eta^\alpha 
A_\beta^j (\mathscr{q}) u^\beta dt
 + \left(\mathscr{p}_i \frac{\partial A_\alpha^i}{\partial q^j}
 (\mathscr{q}) \eta^\alpha  
 A_\beta^j (\mathscr{q})\xi_k^\beta \right)\circ dW_t^k
\\
&\!\!\!\stackrel{\eqref{inf_gen}}=
\mathscr{p}_j \left( u_Q^i \frac{\partial \eta_Q^j}{\partial q^i} -
\eta_Q^i \frac{\partial u_Q^j}{\partial q^i}\right)(\mathscr{q})\,dt +
\mathscr{p}_j \left((\xi_k)_Q^i \frac{\partial \eta_Q^j}{\partial q^i} -
\eta_Q^i \frac{\partial (\xi_k)_Q^j}{\partial q^i}\right)(\mathscr{q})
\circ dW_t^k \\
& \qquad 
+ A_\beta^i(\mathscr{q}) \eta^\beta \frac{\partial \ell}{\partial q^i}
(u,\mathscr{q}) dt
\\
& = \mathscr{p}_j \left[u_Q, \eta_Q \right]^j(\mathscr{q})dt + 
\mathscr{p}_j \left[(\xi_k)_Q, \eta_Q \right]^j(\mathscr{q}) \circ dW_t^k
+ A_\beta^i(\mathscr{q}) \eta^\beta \frac{\partial \ell}{\partial q^i}
(u,\mathscr{q}) dt
\\
& =\mathscr{p}_j [u,\eta]_Q^j(\mathscr{q}) dt +
\mathscr{p}_j [\xi_k ,\eta ]_Q^j (\mathscr{q})\circ dW_t^k +
A_\beta^i(\mathscr{q})\eta^\beta \frac{\partial \ell}{\partial q^i}
(u, \mathscr{q}) dt 
 \,.
\end{split}
\end{align*}

Next, we compute the It\^o contraction term, namely the difference between the Stratonovich integral above and the corresponding It\^o one. Consequently, we find
\[
d\mathscr{p}_j \simeq -\left(\mathscr{p}_i  
\frac{\partial A_\alpha^i}{\partial q^j}(\mathscr{q}) 
\xi_k^\alpha \right)dW_t^k 
\]
and
\[
d [\xi_k ,\eta ]_Q^j(\mathscr{q})= 
d( A_\alpha^j(\mathscr{q})  [\xi_k ,\eta ]^\alpha(\mathscr{q}) )
\simeq 
\left( \frac{\partial A_\alpha^j}{\partial q^i}(\mathscr{q})  
[\xi_k ,\eta ]^\alpha(\mathscr{q}) A_\beta^i(\mathscr{q}) \xi_k^\beta \right) dW_t^k 
\,. 
\]
 Therefore,
\begin{align*}
d(\mathscr{p}_j [\xi_k ,\eta]_Q^j(\mathscr{q}) )&\simeq 
-\left(\mathscr{p}_i  A_\beta^j(\mathscr{q}) [\xi_k ,\eta]^\beta (\mathscr{q}) 
\frac{\partial A_\alpha^i}{\partial q^j}(\mathscr{q}) 
\xi_k^\alpha \right)dW_t^k
+ \left( \mathscr{p}_j  \frac{\partial A_\alpha^j }{\partial q^i}(\mathscr{q}) 
[\xi_k ,\eta ]^\alpha(\mathscr{q}) A_\beta^i(\mathscr{q}) 
\xi_k^\beta \right) dW_t^k  \\
& = \mathscr{p}_i \left( (\xi_k)_Q^j(\mathscr{q})
\frac{\partial [\xi_k, \eta]_Q^i}{\partial q^j}(\mathscr{q})
- \left[\xi_k, \eta\right]_Q^j(\mathscr{q}) 
\frac{\partial (\xi_k)_Q^i}{\partial q^j} (\mathscr{q})
\right)dW_t^k\\
& = \mathscr{p}_i\left[(\xi_k)_Q, \left[\xi_k, \eta\right]_Q \right]^i
(\mathscr{q}) dW_t^k\\
& = \mathscr{p}_i\left[\xi_k,\left[\xi_k, \eta\right] \right]_Q^i
(\mathscr{q}) dW_t^k
\end{align*}
and we obtain the It\^o contraction term
\[
d\big(\mathscr{p}_j [\xi_k ,\eta ]_Q^j(\mathscr{q}) \big)\,.\,dW_t^k =
\sum_{k=1}^N \mathscr{p}_j \big[\xi_k ,\,[\xi_k ,\eta ]  \big]_Q^j
(\mathscr{q})dt
= \sum_{k=1}^N \mathscr{p}_j \left(\operatorname{ad}_{\xi_k} 
\operatorname{ad}_{\xi_k}  \eta \right)_Q^j (\mathscr{q}) dt
\,.\] 
Thus, the first Stratonovich equation in 
\eqref{SEP-eqns-thm} reads, in the It\^o version,
\begin{align*}
d\left(\frac{\delta l}{\delta u}(u,\mathscr{q})\right) &= 
-\operatorname{ad}^*_{u} \left(\frac{\delta l}{\delta u}
(u,\mathscr{q}) \right)dt  -
\operatorname{ad}^*_{\xi_k} \left(\frac{\delta l}{\delta u}
(u,\mathscr{q}) \right) dW_t^k \\
& \qquad +\frac{1}{2}\sum_{k=1}^N\operatorname{ad}^*_{\xi_k}
\operatorname{ad}^*_{\xi_k}  \left(
\frac{\delta l}{\delta u}(u,\mathscr{q}) \right) dt + 
\frac{\partial \ell}{\partial q^i}(u, \mathscr{q})
A_\alpha^i(\mathscr{q})e^\alpha \in \mathfrak{g}^\ast,
\end{align*}
which is an explicit version of the first equation in 
\eqref{Ito-EP-eqns-X-thm}.
This finishes the third proof of Corollary \ref{Ito-SVP-corollary}.
\hfill $\square$

\section{Stochastic Hamiltonian formulation}
\label{Properties-stochcoadmotion-sec}

The goal of this section is to present the Hamiltonian version
of Theorem \ref{StratSVP-thm} and analyze its consequences.

In Section \ref{VP-coadmotion-sec}, we found the stochastic equations
of motion \eqref{SEP-eqns-thm} on $\mathfrak{g}^\ast\times T^*Q$ and 
\eqref{Hamel} on $\mathfrak{g}^\ast \times Q$.   
We want to deduce these equations 
in a purely Hamiltonian manner, without any reference to variational
principles or the Lagrangian formulation of Sections 
\ref{VP-coadmotion-sec} and \ref{Ito-form-coadmotion-sec}. Thus,
we need stochastic Hamiltonians 
$\widetilde{h}(\mathscr{m}, \mathscr{p}_\mathscr{q})$ and
$h(\mathscr{m},\mathscr{q})$. The latter Hamiltonian was already defined in 
the last paragraph of Section \ref{VP-coadmotion-sec}. We also need the
Poisson brackets on $\mathfrak{g}^\ast
\times T^*Q$ and $\mathfrak{g}^\ast\times Q$.

\subsection{The deterministic Hamilton equations}
\label{sec_det_ham_equ}

We first recall the Poisson structure on 
$\mathfrak{g}^\ast \times P$ introduced in \cite{KrMa1987}, where 
the Lie group $G$, whose Lie algebra is $\mathfrak{g}$, acts on
the right on the Poisson manifold $P$ by Poisson diffeomorphisms.

\paragraph{The Poisson manifold $\mathfrak{g}^\ast \times P$.}
In this paragraph, the entire discussion is non-stochastic. We recall
below the results in \cite{KrMa1987} relevant to our development and 
expand on it in certain directions we will need later.
The framework studied in \cite{KrMa1987}, when adapted to our situation, 
is the following. Let a Lie group $G$ act on the right
by Poisson diffeomorphisms on the Poisson manifold $P$. Endow 
$T^*G \times P$ with the Poisson bracket equal to the sum of the 
canonical bracket $\{\cdot , \cdot \}$ on $T^*G$ and the given Poisson bracket $\{\cdot , \cdot \}_P$ on $P$.
Define the free proper left $G$-action by Poisson diffeomorphisms on 
$(T^*G \times P, \{\cdot , \cdot \}$+ $\{\cdot , \cdot \}_P)$  by 
$h \cdot (\alpha_g, p) : = 
\left(T_{hg}^* L_{h^{-1}}(\alpha_g), p \cdot h^{-1} \right)$, 
where $g,h \in G$, 
$\alpha_g\in T_g^*G$, $p \in P$, and $p\cdot h^{-1}$ denotes the given
right action of $h^{-1}$ on the point $p$. Then the map 
$\phi:T^*G \times P \ni
(\alpha_g, p) \mapsto \left(T_e^*L_g(\alpha_g),  p\cdot g \right) \in 
\mathfrak{g}^\ast \times P$ is $G$-invariant and induces a
diffeomorphism $\phi/G: (T^*G \times P)/G \rightarrow 
\mathfrak{g}^\ast \times P$. The push forward of the quotient Poisson
bracket on $(T^*G \times P)/G$ by $\phi/G$ yields the Poisson
bracket
\begin{equation}
\label{KMPB}
\{f, h\}_{\mathfrak{g}^\ast \times P} (\mu, p)= \{f^p, h^p\}_-(\mu)  + 
\left\langle \mathbf{d}f^\mu(p), \left(\frac{\delta h^p}{\delta \mu} \right)_P(p) \right\rangle_P - 
\left\langle \mathbf{d}h^\mu(p), \left(\frac{\delta f^p}{\delta \mu} \right)_P(p) \right\rangle_P + 
\{f^\mu, h^\mu\}_P(p)\,,
\end{equation}
for all $f, h \in C^{\infty}(\mathfrak{g}^\ast \times P)$, where 
$f^\mu, h^\mu \in C^{\infty}(P)$ and $f^p, h^p \in 
C^{\infty}(\mathfrak{g}^\ast)$ are defined by $f^\mu(p): = f^p(\mu) :=
f(\mu, p)$, for all $\mu \in\mathfrak{g}^\ast$, $p\in P$, and 
similarly for $h$ (\cite[Proposition 2.1]{KrMa1987}). Thus, Hamilton's equations for $h \in 
C^{\infty}( \mathfrak{g}^\ast\times P)$ are
\begin{equation}
\label{KMHam_eq}
\frac{d}{dt}\mu= \operatorname{ad}^*_{\frac{\delta h^p}{ \delta \mu}} \mu 
- \mathbf{J}_{T^*P}\left(\mathbf{d}h^\mu(p) \right), \qquad
\frac{d}{dt}p = \left(\frac{\delta h^p}{\delta \mu} \right)_P(p) + 
X^P_{h^\mu}(p),
\end{equation}
where $X^P_{h^\mu}$ denotes the Hamiltonian vector field of $h^\mu \in 
C^{\infty}(P)$ on the Poisson manifold $P$ and $\mathbf{J}_{T^*P}:
T^*P \rightarrow  \mathfrak{g}^\ast$ is the momentum map of the 
cotangent lifted action (see \eqref{momentum_map} with $Q$ replaced by $P$).

Suppose now that the right $\mathfrak{g}$-action on $P$ has a momentum 
map $\mathbf{J}_P: P \rightarrow \mathfrak{g}^\ast$, 
which means that $\xi_P[f] = \left\{f, \mathbf{J}_P^ \xi\right\}_P$ 
for all $f \in C^{\infty}(P)$ and all 
$\xi \in \mathfrak{g}$, where $\mathbf{J}_P^\xi(p) : = 
\left\langle \mathbf{J}_P(p), \xi \right\rangle_\mathfrak{g}$. Suppose
also that $\mathbf{J}_P$ is infinitesimally equivariant i.e.,  
$\mathbf{J}_P^{[\xi, \eta]} = -\left\{ \mathbf{J}_P^\xi, 
\mathbf{J}_P^\eta\right\}_P$,
for all $\xi, \eta \in \mathfrak{g}$. We recall that the existence of
a momentum map on $P$ for a connected Lie group action forces the group 
orbits to be included in the symplectic leaves of $P$, which is 
a rather stringent condition. There are many examples of Poisson
Lie group actions that do not admit a momentum map. (See, e.g., 
\cite[Chapters 4 and 5]{OrRa2004} for a discussion of this problem.) However, in the presence of an 
equivariant momentum
map $\mathbf{J}_P: P \rightarrow \mathfrak{g}^\ast$, the
diffeomorphism $\psi: \mathfrak{g}^\ast\times P\ni (\mu, p) \mapsto 
(\mu - \mathbf{J}_P(p), p) \in \mathfrak{g}^\ast\times P$ pushes 
forward the Poisson bracket $\{\cdot,\cdot\}_{\mathfrak{g}^\ast \times P}$,
given by \eqref{KMPB}, to the sum Poisson bracket
\begin{equation}
\label{sum_PB_g_star_q}
\{f, h\}_{\rm sum}(\mu, p): = \{f^p, h^p\}_-(\mu) + \{f^\mu, h^\mu\}_P(p)
\end{equation}
on $\mathfrak{g}^\ast \times P$.
This is proved for left actions in \cite[Proposition 2.2]{KrMa1987}; 
although our formulas in \eqref{KMPB} and the definition of $\psi$ have relative 
sign changes because we work with a right $G$-action on $P$. The proof is
a direct verification. Hamilton's equations $\frac{d}{dt}f=\{f,h\}$ for the
sum Poisson bracket \eqref{sum_PB_g_star_q} are
\begin{equation}
\label{Ham_equ_sum_general}
\frac{d}{dt}\mu = \operatorname{ad}_{\frac{\delta h^p}{\delta \mu}}^* \mu, 
\qquad \frac{d}{dt}p = X^P_{h^\mu}(p).
\end{equation}

Using  \eqref{sum_PB_g_star_q}, it follows that if $k\in 
C^{\infty}(\mathfrak{g}^\ast\times P)$ is a Casimir function on 
$\left(\mathfrak{g}^\ast \times P, \{\cdot , \cdot \}_{\rm sum} \right)$, 
then $k \circ \psi$ is a Casimir function on $\left(\mathfrak{g}^\ast 
\times P, \{\cdot , \cdot \}_{\mathfrak{g}^\ast \times P} \right)$. In
particular, if $k_P \in C^{\infty}(P)$ is a Casimir function, then 
the function $(\mu, p) \mapsto k_P(p)$ is a Casimir function 
for  $\left(\mathfrak{g}^\ast 
\times P, \{\cdot , \cdot \}_{\mathfrak{g}^\ast \times P} \right)$. This
can also be easily checked directly using \eqref{KMHam_eq}. More 
interestingly, if $k_{\mathfrak{g}^\ast} \in C^{\infty}(\mathfrak{g}^\ast)$ 
is a Casimir function on $\mathfrak{g}^\ast$, then 
$(\mu, p) \mapsto k_{\mathfrak{g}^\ast}( \mu - \mathbf{J}_P(p))$ is a 
Casimir function for  $\left(\mathfrak{g}^\ast 
\times P, \{\cdot , \cdot \}_{\mathfrak{g}^\ast \times P} \right)$ (\cite[Corollary 2.3]{KrMa1987}).

Since the projections $\pi_{\mathfrak{g}^\ast}:
\left(\mathfrak{g}^\ast \times P, \{\cdot , \cdot \}_{\rm sum} \right) 
\rightarrow \mathfrak{g}^\ast_-$ and $\pi_P:
\left(\mathfrak{g}^\ast \times P, \{\cdot , \cdot \}_{\rm sum} \right)
\rightarrow P$ are Poisson maps, their compositions 
\begin{equation}
\label{Poisson_projections}
\pi_{\mathfrak{g}^\ast}
\circ\psi:\left(\mathfrak{g}^\ast\times P,\{\cdot,\cdot\}_{\mathfrak{g}^*
\times P} \right)\ni (\mu, p) \longmapsto \mu - \mathbf{J}_P(p)\in
\mathfrak{g}^\ast_-, \quad 
\pi_P\circ\psi:\left(\mathfrak{g}^\ast\times P,
\{\cdot,\cdot\}_{\mathfrak{g}^*\times P} \right) \ni (\mu, p) 
\longmapsto p\in P
\end{equation} 
with the Poisson diffeomorphism $\psi: \left(\mathfrak{g}^\ast\times P,
\{\cdot,\cdot\}_{\mathfrak{g}^*\times P} \right) \rightarrow
\left(\mathfrak{g}^\ast \times P, \{\cdot , \cdot \}_{\rm sum} \right)$
are also Poisson maps. 

Remarkably, the projection $\pi_{\mathfrak{g}^\ast}:
\left(\mathfrak{g}^\ast\times P,\{\cdot,\cdot\}_{\mathfrak{g}^*
\times P} \right) \rightarrow \mathfrak{g}^\ast_-$
is also a Poisson map, as an easy direct verification shows, using
for $\mathsf{f} \in C^{\infty}(\mathfrak{g}^\ast)$
the identities $(\mathsf{f} \circ \pi_{\mathfrak{g}^\ast})^p = \mathsf{f}$ 
for every $p \in P$ and $(\mathsf{f} \circ \pi_{\mathfrak{g}^\ast})^\mu = 
\mathsf{f}(\mu)$, a constant function on $P$, for every  
$\mu \in \mathfrak{g}^\ast$. In
particular, this means that Hamilton's equations \eqref{KMHam_eq} for
a Hamiltonian of the form $h:=\mathsf{h} \circ \pi_{\mathfrak{g}^\ast}$,
where $\mathsf{h} \in C^{\infty}(\mathfrak{g}^\ast)$ (i.e.,
$h$ does not depend on $p \in P$), are the Lie-Poisson equations
for $\mathsf{h}$ on $\mathfrak{g}^\ast_-$ which completely decouple from
the second equation in \eqref{KMHam_eq}. The second equation is
given by an infinitesimal generator at every instance of time, namely,
if $\mu(t)$ is a solution of the Lie-Poisson equation $\frac{d}{dt}\mu = 
\operatorname{ad}_{\frac{\delta \mathsf{h}}{\delta \mu}}^* \mu$, 
then the second equation in \eqref{KMHam_eq} is the \textit{time-dependent
infinitesimal generator equation} 
\begin{equation}
\label{second_KMHam_eq_decoupled}
\frac{d}{dt}p(t) = \left(\frac{\delta \mathsf{h}}{\delta \mu(t)} \right)_P
(p(t)).
\end{equation}

Similarly, the projection $\pi_P: \mathfrak{g}^\ast \times P 
\rightarrow P$ is a Poisson map relative to both Poisson brackets
$\{\cdot , \cdot \}_{\mathfrak{g}^\ast \times P}$ and 
$\{\cdot , \cdot \}_{\rm sum}$, because if $\overline{\mathsf{f}} \in 
C^{\infty}(P)$, then $(\overline{\mathsf{f}} \circ \pi_P)^p = 
\overline{\mathsf{f}}(p)$, a constant on $\mathfrak{g}^\ast$, and 
$(\overline{\mathsf{f}} \circ \pi_P)^\mu = \overline{\mathsf{f}}$,
for any $\mu \in \mathfrak{g}^\ast$.

Hamilton's equations \eqref{KMHam_eq} and \eqref{Ham_equ_sum_general} show
that the manifolds $\{\mu\} \times P$ and $\mathfrak{g}^\ast_-\times\{p\}$
for any $\mu \in \mathfrak{g}^\ast$, $p \in P$, are not Poisson
submanifolds of $\mathfrak{g}^\ast \times P$ endowed with either
Poisson bracket $\{\cdot , \cdot \}_{\mathfrak{g}^\ast \times P}$
or $\{\cdot , \cdot \}_{\rm sum}$.

\paragraph{The Poisson brackets on $\mathfrak{g}^\ast \times Q$ and
$\mathfrak{g}^\ast \times T^*Q$.} We specialize the results of
the previous paragraph to the following Poisson manifolds: $Q$, endowed with 
the zero Poisson structure, and $T^*Q$, endowed with the canonical
Poisson structure (whose local expression is \eqref{can_Poisson_bracket}).
We continue to work in the non-stochastic context.

For any $f, h \in C^{\infty}(\mathfrak{g}^\ast\times Q)$, the Poisson
bracket \eqref{KMPB} reads 
\begin{align}
\label{PB_q}
\{f,h\}_{\mathfrak{g}^\ast_-\times Q}(m,q):&=
\begin{bmatrix}
\partial f / \partial m_\alpha \\ \partial f / \partial q^i 
\end{bmatrix}^\mathsf{T}
\begin{bmatrix}
-\,c_{\alpha\beta}{}^\gamma\,m_{\gamma} & -\,A^j_\alpha 
\\
A^i_\beta & 0 
\end{bmatrix}
\begin{bmatrix}
\partial h / \partial m_\beta \\ \partial h / \partial q^j 
\end{bmatrix}  \nonumber \\
&= \left\{f^q,h^q \right\}_-(m)  
+ \left\langle \mathbf{d}f^m(q), 
\left(\frac{\delta h^q}{\delta m}\right)_Q(q) \right\rangle_Q
- \left\langle \mathbf{d}h^m(q), 
\left(\frac{\delta f^q}{\delta m}\right)_Q(q) \right\rangle_Q,
\end{align}
where $f^q \in C ^{\infty}(\mathfrak{g}^\ast)$ and $f^\mu\in 
C ^{\infty}(Q)$ are defined by $f^q (\mu ) : = f^m (q) : = f(m,q)$, 
for all $m\in \mathfrak{g}^\ast$, $q \in Q$, and $\{ \cdot ,\cdot \}_-$ 
is the minus Lie-Poisson
bracket \eqref{LP} on $\mathfrak{g}^\ast_-$.

Similarly, for any $\widetilde{f}, \widetilde{h} \in 
C^{\infty}(\mathfrak{g}^\ast\times T^*Q)$, the Poisson
bracket \eqref{KMPB} reads 
\begin{align}
\label{PB_t_star_q}
&\left\{\widetilde{f},\widetilde{h}\right\}_{\mathfrak{g}^\ast_- 
\times T^*Q}\left(m, p_q\right):=
\begin{bmatrix}
\partial \widetilde{f} / \partial m_\alpha \\ 
\partial \widetilde{f} / \partial q^i \\ 
\partial \widetilde{f} / \partial p_i
\end{bmatrix}^\mathsf{T}
\begin{bmatrix}
-\,m_{[\,\alpha\,,\, \beta\,]} & -\,A^j_\alpha & 
p_k\frac{\partial A^k_\alpha}{\partial q^j}
\\
A^i_\beta & 0 & \delta^i_j
\\
-\,p_k\frac{\partial A^k_\beta}{\partial q^i} & -\,\delta^j_i & 0
\end{bmatrix}
\begin{bmatrix}
\partial \widetilde{h} / \partial m_\beta \\ 
\partial \widetilde{h} / \partial q^j \\ 
\partial \widetilde{h} / \partial p_j
\end{bmatrix} \nonumber  \\
& = \left\{\widetilde{f}^{p_q}, \widetilde{h}^{p_q} \right\}_- (m) 
+ \left\langle \mathbf{d} f^m(p_q), 
\left(\frac{\delta \widetilde{h}^{p_q}}{\delta m}\right)_{T^*Q}(p_q)
\right\rangle_Q  
- \left\langle \mathbf{d} h^m(p_q), 
\left(\frac{\delta \widetilde{f}^{p_q}}{\delta m}\right)_{T^*Q}(p_q)
\right\rangle_Q
+\left\{ \widetilde{f}^m, \widetilde{h}^m\right\}(p_q),
\end{align}
where $\widetilde{f}^{p_q} \in C^{\infty}(\mathfrak{g}^\ast)$ and
$\widetilde{f}^m \in C^{\infty}(T^*Q)$ are defined by 
$\widetilde{f}^{p_q}(m):=\widetilde{f}^m(p_q): = \widetilde{f}(m,p_q)$,
for all $m \in \mathfrak{g}^\ast$, $p_q \in T^*Q$,
and $\{\cdot , \cdot \}$ is 
the canonical Poisson bracket \eqref{can_Poisson_bracket} on $T^*Q$.
This proves the first statement in the following theorem.

\begin{theorem}
\label{two_PB_thm}
The brackets \eqref{PB_q} and \eqref{PB_t_star_q} are Poisson brackets
on $\mathfrak{g}^\ast \times Q$ and $\mathfrak{g}^\ast \times T^*Q$, 
respectively. Hamilton's equations on $\mathfrak{g}^\ast \times T^*Q$
are given by \eqref{KMHam_eq} with $P$ replaced by $T^*Q$. In standard
coordinates, for $\widetilde{h} \in C^{\infty}(\mathfrak{g}^\ast \times 
T^*Q)$, these equations are
\begin{equation}
\label{KM_star_q_eq_coord}
\frac{d}{dt}m_\alpha = -m_{[\alpha, \beta]}
\frac{\partial \widetilde{h}}{\partial m_\beta} -
A^j_\alpha\frac{\partial \widetilde{h}}{ \partial q^j} + p_k
\frac{\partial A_\alpha ^k}{ \partial q^j} 
\frac{ \partial \widetilde{h}}{ \partial p_j}\,, \qquad 
\frac{d}{dt}q^i = A^i_\beta \frac{\partial \widetilde{h}}{\partial m_\beta} + 
\frac{\partial \widetilde{h}}{ \partial p_i}\,,\qquad
\frac{d}{dt}p_i = -p_k \frac{\partial A_\beta^k}{ \partial q^i}
\frac{\partial \widetilde{h}}{\partial m_\beta} -
\frac{\partial \widetilde{h}}{ \partial q^i}\,.
\end{equation}
The diffeomorphism $\psi: \mathfrak{g}^\ast\times T^*Q\ni (\mu, p_q) 
\mapsto (\mu - \mathbf{J}_{T^*Q}(p), p_q) \in \mathfrak{g}^\ast\times T^*Q$ 
pushes  forward the Poisson bracket \eqref{PB_t_star_q} to the sum Poisson bracket. If $k_{\mathfrak{g}^\ast} \in 
C^{\infty}(\mathfrak{g}^\ast)$ is a Casimir function on 
$\mathfrak{g}^\ast$, then $(\mu, p_q) \mapsto k_{\mathfrak{g}^\ast}
(\mu - \mathbf{J}_{T^*Q}(p_q))$ is a 
Casimir function for  $\left(\mathfrak{g}^\ast 
\times T^*Q, \{\cdot , \cdot \}_{\mathfrak{g}^\ast \times T^*Q} \right)$.

Hamilton's equations \eqref{KMHam_eq} on $\mathfrak{g}^\ast \times Q$
for $h \in C^{\infty}(\mathfrak{g}^\ast\times Q)$
{\rm(}with $P$ replaced by the trivial Poisson manifold $Q${\rm)} are 
Hamel's equations {\rm\cite{Ha1904}:}
\footnote{For a modern formulation, see, e.g., {\rm\cite[\S3.8, p.144]{Bloch2015}} or 
{\rm\cite{BlMaZe2009}}}
\begin{equation}
\label{KM_q_eq_coord}
\begin{aligned}
&\frac{d}{dt}m = \operatorname{ad}_{\frac{\delta h^q}{\delta m}}^* m - 
\mathbf{J}_{T^*Q}(\mathbf{d}h^m(q)), \qquad \quad  
\frac{d}{dt}q = \left(\frac{\delta h^q}{\delta m} \right)_Q(q) \qquad 
\Longleftrightarrow \\
&\frac{d}{dt}m_\alpha = -c_{\alpha \beta}{}^\gamma m_\gamma
\frac{\partial h}{ \partial m_\beta} - 
A^j_\alpha\frac{\partial h}{\partial q^j}\,, \qquad\;\; 
\frac{d}{dt}q^i = A^i_\beta \frac{\partial h}{\partial m_\beta}\,.
\end{aligned}
\end{equation}
If $\widetilde{h}$ does not depend on $p_q \in T^*Q$
{\rm(}respectively, $h$ does not depend on $q \in Q${\rm)}, 
then Hamilton's equations \eqref{KM_star_q_eq_coord} {\rm(}respectively, 
\eqref{KM_q_eq_coord}{\rm)} decouple into the Lie-Poisson equations on 
$\mathfrak{g}^\ast_-$ and the time-dependent infinitesimal 
generator equations for 
$\frac{\delta \widetilde{h}}{\delta m(t)}\in \mathfrak{g}$ on $T^*Q$
{\rm(}respectively, $\frac{\delta h}{\delta m(t)} \in\mathfrak{g}$ 
on $Q${\rm)}.

The four projections of $\mathfrak{g}^\ast \times Q$ and 
$\mathfrak{g}^\ast \times T^*Q$ on every factor are Poisson 
{\rm(}$\mathfrak{g}^*$ has the minus Lie-Poisson structure{\rm)}.
The map $\mathfrak{g}^\ast \times T^*Q \ni (\mu, p_q) 
\mapsto \mu - \mathbf{J}_{T^*Q} (p_q) \in\mathfrak{g}^\ast_-$
is Poisson. The embedding $\mathfrak{g}^\ast \times Q\ni (m, q) \mapsto 
(m, 0_q) \in \mathfrak{g}^\ast \times T^*Q$ is not Poisson. The map
$\rho: \mathfrak{g}^\ast \times T^*Q \ni (m, p_q) \longmapsto
(m, q) \in \mathfrak{g}^\ast \times Q$ is Poisson. 
\end{theorem}

\begin{proof} 
Formulas \eqref{KM_star_q_eq_coord} and \eqref{KM_q_eq_coord} are obtained by calculating
\eqref{KMHam_eq} for these two cases. 
The statements about the Poisson character of the five projections 
and the diffeomorphism $\psi$, the decoupling of the equations for
Hamiltonians depending only on $m \in \mathfrak{g}^\ast$, as well as the assertion about
the Casimir functions, were
proved in the previous paragraph for a general Poisson manifold $P$.

Setting all coordinates $p_i=0$ 
in \eqref{KM_star_q_eq_coord} 
does not yield \eqref{KM_q_eq_coord}, i.e., a Hamiltonian vector
field on the Poisson manifold $\mathfrak{g}^\ast \times T^*Q$, restricted 
to $\mathfrak{g}^\ast\times Q$ is, in general, not tangent to 
$\mathfrak{g}^\ast\times Q$. This proves that $\mathfrak{g}^\ast
\times Q$ is not a Poisson submanifold of $\mathfrak{g}^\ast \times T^*Q$.

Let $\pi: T^*Q \rightarrow Q$ be the cotangent bundle projection. 
The map, $\rho: \mathfrak{g}^\ast \times T^*Q \rightarrow \mathfrak{g}^\ast
\times Q$ is Poisson. This is a direct verification using the
formulas $(f \circ \rho)^m = f^m\circ \pi$, 
$\mathbf{d}(f \circ \rho)^m(p_q) = \mathbf{d} f^m(q) \circ T_{p_q}\pi$, 
$(f \circ \rho)^{p_q} = f^q$, $\frac{\delta(f \circ\rho)^{p_q}}{\delta m}
= \frac{\delta f^q}{\delta m}$, and the fact that the infinitesimal
generators $u_{T^*Q}$ (of the lifted $G$-action on $T^*Q$) and $u_Q$ 
(of the $G$-action on $Q$) are $\pi$-related.

The last statement is obtained by setting $\partial h/\partial p_i = 0$
in \eqref{KM_star_q_eq_coord}.
\end{proof}

\begin{remark}\rm[Collective Lie-Poisson momentum map dynamics]
\label{rem_coll}
In \eqref{KM_q_eq_coord}, note that if
the Hamiltonian depends only on $m \in \mathfrak{g}^\ast$, 
i.e., the Hamiltonian is of the form $h:=\mathsf{h} \circ 
\pi_{\mathfrak{g}^\ast}$, where
$\mathsf{h} \in C^{\infty}(\mathfrak{g}^\ast)$ and $\pi_{\mathfrak{g}^\ast}:
\left(\mathfrak{g}^\ast \times Q, \{ \cdot , \cdot \}_{
\mathfrak{g}^\ast \times Q}\right) \ni (m,q) 
\mapsto m \in \left(\mathfrak{g}^\ast, \{\cdot , \cdot \}_- \right)$,
Hamel's equations in \eqref{KM_q_eq_coord} become the Lie-Poisson equations on $\mathfrak{g}^\ast_-$.  
In this case, we say the motion \textit{collectivizes} (see 
\cite{GuSt1980}) since $\pi_{\mathfrak{g}^\ast}$ is a Poisson map. 
\hfill $\lozenge$
\end{remark}

\begin{remark}\rm
The manifold $Q$ endowed with the zero Poisson structure does not
admit a momentum map $\mathbf{J}_Q:Q \rightarrow \mathfrak{g}^\ast$. 
Indeed, if $\mathbf{J}_Q$ existed, we would have $\xi_Q[f] = 
\left\{f, \mathbf{J}_Q^\xi\right\}_Q = 0$ for all $f \in 
C^{\infty}(Q)$ and all $\xi \in \mathfrak{g}$, which would imply the false
statement that all smooth functions on $Q$ are $\mathfrak{g}$-invariant.
As a consequence, the statement in the previous paragraph about the
Poisson bracket on $\mathfrak{g}^\ast \times Q$ being isomorphic to the
sum Poisson bracket, which in this case would be just the minus 
Lie-Poisson bracket, does not apply. Similarly, Casimir functions
on $\mathfrak{g}^\ast$ do not induce Casimir functions on 
$\mathfrak{g}^\ast \times Q$. 
\hfill $\lozenge$
\end{remark}

For the statement of the next corollary, we need to introduce
the \textit{fiber translation vector field} $T_\alpha \in 
\mathfrak{X}(T^*Q)$ associated to a one-form $\alpha\in\Omega^1(Q)$.
The map $T^*Q \ni p_q \mapsto p_q - t \alpha(q) \in T^*Q$, 
$t \in\mathbb{R}$, is a one-parameter group. Define $T_\alpha$ to
be the vector field with this flow, i.e.,
\[
T_\alpha(p_q) : = \left.\frac{d}{dt}\right|_{t=0}\left(p_q - t \alpha(q) \right)\in T_{p_q} (T^* Q).
\]
This vector field is identical to the vertical lift operation by 
$-\alpha \in\Omega^1(T^*Q)$.

\begin{corollary}
\label{decoupling_cor}
Hamilton's equations \eqref{KMHam_eq} {\rm(}with $P = T^*Q${\rm)} on 
$\mathfrak{g}^\ast \times T^*Q$ for $\widetilde{h}:=h \circ \rho$, where $h \in 
C^{\infty}(\mathfrak{g}^\ast \times Q)$, i.e., $\widetilde{h}(m, p_q): = h(m,q)$, take the form
\begin{equation}
\label{intrinsic_equ_on_tstarq}
\frac{d}{dt}m = \operatorname{ad}_{\frac{\delta h^q}{\delta m}}^* m 
- \mathbf{J}_{T^*Q} \left(\mathbf{d}h^m(q)\right), \qquad 
\frac{d}{dt}p_q = \left(\frac{\delta h^q}{\delta m} \right)_{T^*Q} + 
T_{\mathbf{d}h^m} (p_q).
\end{equation}
In addition, equations \eqref{intrinsic_equ_on_tstarq} imply both \eqref{KM_q_eq_coord}
and the non-homogeneous Lie-Poisson equations
\begin{equation}
\label{final_strange_equ_gen}
\frac{d}{dt}
\mathbf{J}_{T^*Q}(p_q(t)) = 
\operatorname{ad}_{\frac{\delta h^{q(t)}}{\delta m(t)}}^\ast
\mathbf{J}_{T^*Q}(p_q(t))
- \mathbf{J}_{T^*Q}\left(\mathbf{d}h^{m(t)} (q(t)\right)
\end{equation}
for $\mathbf{J}_{T^*Q}(p_q(t))$,
where $(m(t), q(t))$ is
the solution of Hamel's equations \eqref{KM_q_eq_coord}.
\end{corollary} 

\begin{proof}
We have $\widetilde{h}^{p_q} = h^q \in C ^{\infty}(\mathfrak{g}^\ast)$ 
and $\widetilde{h}^m = h^m \circ \pi \in C^{\infty}(T^*Q)$, where 
$\pi:T^*Q \rightarrow Q$ is the cotangent bundle projection. 

We compute $\mathbf{J}_{T^*(T^*Q)}\left(\mathbf{d}(h^m \circ \pi)(p_q) 
\right)$, the second summand on the right hand side of the first
equation in \eqref{KMHam_eq} for $P=T^*Q$. To do this, we note 
that since the $G$-action on $T^*Q$
is the cotangent lifted $G$-action on $Q$, the cotangent bundle
projection $\pi: T^*Q \rightarrow Q$ is equivariant and thus
the infinitesimal generators of the two actions for the same Lie
algebra element are $\pi$-related, i.e., $T\pi \circ v_{T^*Q} = 
v_Q \circ \pi$ for any $v \in \mathfrak{g}$. Therefore,
\begin{align}
\label{first_equ_second_term}
\left\langle\mathbf{J}_{T^*(T^*Q)}\left(\mathbf{d}(h^m \circ \pi)(p_q)
\right), v \right\rangle_\mathfrak{g}
& \stackrel{\eqref{momentum_map}}= \left\langle 
\mathbf{d}(h^m \circ \pi)(p_q), v_{T^*Q}(p_q) \right\rangle_\mathfrak{g}
=\left\langle \mathbf{d}h^m(q), T_{p_q} \pi\left(v_{T^*Q}(p_q)\right) 
\right\rangle_\mathfrak{g}  \nonumber \\
& \;\,= \left\langle \mathbf{d}h^m(q), v_Q(q) \right\rangle_\mathfrak{g}
\stackrel{\eqref{momentum_map}}= \left\langle\mathbf{J}_{T^*Q} 
\left( \mathbf{d}h^m(q)\right), v \right\rangle_\mathfrak{g}.
\end{align}

Next, we compute $X^{T^*Q}_{\widetilde{h}^m}(p_q) = 
X^{T^*Q}_{h^m \circ \pi}(p_q)$, the second summand on the right 
hand side of the second equation in \eqref{KMHam_eq} for $P=T^*Q$.
Since this affects only the dynamics on $T^*Q$, we prove, in general, that
\begin{equation}
\label{second_equ_second_term}
X^{T^*Q}_{k \circ \pi} = T_{\mathbf{d}k}, \quad \text{for all}
\quad k \in C ^{\infty}(Q).
\end{equation}
To prove \eqref{second_equ_second_term}, it is easier to work in
local coordinates. Hamilton's equations for $k\circ \pi$ are
\[
\frac{dq^i}{dt} =\frac{ \partial(k \circ \pi) }{\partial p_i} = 0,
\qquad \frac{dp_i}{dt} = - \frac{ \partial (k \circ \pi) }{\partial q^i} 
= - \frac{ \partial k}{\partial q^i}\,,
\]
whose solution is $q^i(t) = q^i_0$, $p_i(t) = p_i^0 - 
t \frac{ \partial k}{\partial q^i}(q_0)$, where $(q^1_0, \ldots,
q^n_0, p_1^0, \ldots p_n^0)$ is the initial condition. Thus, the
flow of $X^{T^*Q}_{k \circ \pi}$ is $p_q \mapsto p_q - t \mathbf{d}k(q)$
which coincides with the flow of $T_{\mathbf{d}k}$, thereby proving
\eqref{second_equ_second_term}.

Using the identities \eqref{first_equ_second_term} and \eqref{second_equ_second_term}, equations \eqref{KMHam_eq} become
\eqref{intrinsic_equ_on_tstarq}. 

Since $\rho: \left(\mathfrak{g}^\ast \times T^*Q, 
\{ \cdot , \cdot \}_{\mathfrak{g}^\ast \times T^*Q}\right) \ni (m, p_q) 
\longmapsto (m, q) \in \left(\mathfrak{g}^\ast \times Q, 
\{\cdot , \cdot \}_{\mathfrak{g}^\ast \times Q}\right)$ is a Poisson
map by Theorem \ref{two_PB_thm}, we have 
$T\rho \circ X^{\mathfrak{g}^\ast \times T^*Q}_{h \circ \rho} = 
X^{\mathfrak{g}^\ast \times Q}_h \circ \rho$ for any $h \in 
C^{\infty}(\mathfrak{g}^\ast\times Q)$, which is equivalent to saying
that \eqref{intrinsic_equ_on_tstarq} (the equations of motion defined
by the Hamiltonian vector field 
$X^{\mathfrak{g}^\ast \times T^*Q}_{h \circ \rho}$) project to 
\eqref{KM_q_eq_coord} (the equations of motion defined
by the Hamiltonian vector field 
$X^{\mathfrak{g}^\ast \times Q}_{h}$).

Finally, we prove \eqref{final_strange_equ_gen}. We have
\begin{align}
\label{first_int_funny_equ}
\frac{d}{dt}\mathbf{J}_{T^*Q}(p_q(t)) &=
T_{p_q(t)}\mathbf{J}_{T^*Q} \left(\frac{d}{dt}p_q(t) \right)
\stackrel{\eqref{intrinsic_equ_on_tstarq}}= 
T_{p_q(t)}\mathbf{J}_{T^*Q} \left(\left(\frac{\delta h^{q(t)}}{\delta m(t)} \right)_{T^*Q} (p_q(t))\right) + 
T_{p_q(t)}\mathbf{J}_{T^*Q} \left(
T_{\mathbf{d}h^{m(t)}} (p_q(t)) \right) \nonumber \\
& \mathrel{\operatorname*{=}_{\eqref{second_equ_second_term}}^{\eqref{inf_equ_der}}}
\operatorname{ad}_{\frac{\delta h^{q(t)}}{\delta m(t)}}^\ast
\mathbf{J}_{T^*Q}(p_q(t)) + 
T_{p_q(t)}\mathbf{J}_{T^*Q} \left(
X^{T^*Q}_{h^{m(t)} \circ \pi} (p_q(t)) \right).
\end{align}
One may compute the second summand by pairing it with any $v \in \mathfrak{g}$
to obtain
\begin{align}
\label{second_int_funny_equ}
\left\langle T_{p_q(t)}\mathbf{J}_{T^*Q} \left(
X^{T^*Q}_{h^{m(t)} \circ \pi} (p_q(t)) \right), v 
\right\rangle_ \mathfrak{g} & = 
\left\langle \mathbf{d}\mathbf{J}_{T^*Q}^v(p_q(t)), 
X^{T^*Q}_{h^{m(t)} \circ \pi} (p_q(t)) \right\rangle_{T^*Q}
= \{ \mathbf{J}_{T^*Q}^v, h^{m(t)}\circ \pi\}(p_q(t)) \nonumber \\
& = - \{ h^{m(t)}\circ \pi,  \mathbf{J}_{T^*Q}^v\}(p_q(t))
= - \left\langle\mathbf{d}(h^{m(t)}\circ \pi)(p_q(t)),
X^{T^*Q}_{\mathbf{J}_{T^*Q}^v}(p_q(t)) \right\rangle_Q  \nonumber \\
&\stackrel{\eqref{momentum_map_def}}=
- \left\langle \mathbf{d}(h^{m(t)}\circ \pi)(p_q(t)),
v_{T^*Q} (p_q(t)) \right\rangle_Q \nonumber \\
& \stackrel{\eqref{inf_equ_der}}=
- \left\langle \mathbf{d} h^{m(t)}(q(t)), v_Q(q(t)) \right\rangle_Q
\nonumber \\
&\stackrel{\eqref{momentum_map}}=
- \left\langle\mathbf{J}_{T^*Q}\left(\mathbf{d}h^{m(t)} (q(t)\right), 
v \right\rangle_ \mathfrak{g}.
\end{align}
Formulas \eqref{first_int_funny_equ} and \eqref{second_int_funny_equ}
now yield \eqref{final_strange_equ_gen}.
\end{proof}

\begin{corollary}
\label{lagrangian_decoupling_cor}
If the Hamiltonian $\widetilde{h} \in C^{\infty}(\mathfrak{g}^\ast 
\times T^*Q)$ is of the form $\widetilde{h} = h \circ \rho$ and 
$h \in C^{\infty}(\mathfrak{g}^\ast \times Q)$ is hyperregular, i.e.,
the parameter dependent reduced Legendre transformation
$\mathfrak{g} \times Q\ni (u,q) \mapsto\left(m(u, q), q\right) := 
\left(\frac{\delta \overline{\ell}^q}{\delta u}, q \right) \in 
\mathfrak{g}^\ast \times Q$ is a diffeomorphism,
where $\overline{\ell}\in C^{\infty}(\mathfrak{g}\times Q)$ 
and $h(m,q): = \left\langle m, u(m,q) \right\rangle_\mathfrak{g} - \overline{\ell}(u(m,q),q)$, equation \eqref{final_strange_equ_gen} takes
the form
\begin{equation}
\label{final_strange_equ_gen_lagrangian}
\frac{d}{dt}
\mathbf{J}_{T^*Q}(p_q(t)) = \operatorname{ad}_{u(t)}^* \mathbf{J}_{T^*Q}(p_q(t))
+ \mathbf{J}_{T^*Q}(\mathbf{d} \ell^{u(t)}(q(t)),
\end{equation}
where $(u(t), q(t))$ is the solution of the Lagrangian version of
Hamel's equations
\begin{equation}
\label{KM_q_eq_Hamel}
\begin{aligned}
&\frac{d}{dt}\frac{\delta \overline{\ell}^q}{\delta u} = 
\operatorname{ad}_u^* \frac{\delta \overline{\ell}^q}{\delta u} + 
\mathbf{J}_{T^*Q}(\mathbf{d}\overline{\ell}^u(q)), \qquad \quad  
\frac{d}{dt}q = u_Q(q) \qquad 
\Longleftrightarrow \\
&\frac{d}{dt}\frac{\partial\overline{\ell}^q}{\partial u^\alpha} 
= -c_{\alpha \beta}{}^\gamma 
\frac{\partial\overline{\ell}^q}{\partial u^\gamma} u^\beta + 
A^j_\alpha\frac{\partial \overline{\ell}}{\partial q^j}\,, \qquad\;\; 
\frac{d}{dt}q^i = A^i_\beta u^\beta\,.
\end{aligned}
\end{equation}
\end{corollary}

\begin{proof} 
By hyperregularity, we can solve for $u$ to get $u(m,q) \in \mathfrak{g}$ and we have
\[
\left\langle \delta m, \frac{\delta h^q}{\delta m} 
\right\rangle_\mathfrak{g} = \left\langle \delta m,  
u(m,q) \right\rangle_\mathfrak{g} +\left\langle m, 
\mathbf{D}u^q(m) \cdot \delta m \right\rangle_\mathfrak{g} -
\left\langle \frac{ \delta \overline{\ell}^q}{\delta u(m,q)}, 
\mathbf{D}u^q(m) \cdot \delta m \right\rangle_ \mathfrak{g}
=  \left\langle \delta m,  u(m,q) \right\rangle_\mathfrak{g}
\]
because $\frac{\delta \overline{\ell}^q}{\delta u(m,q)} = m$ by definition
of $m$ and hyperregularity. Thus $u(m, q) =\frac{\delta h^q}{\delta m}$.
Since $\overline{\ell}(u, q) = \left\langle m(u,q), u 
\right\rangle_ \mathfrak{g} - h(m(u, q), q)$, we get
\begin{align*}
\mathbf{d}\ell^u(q) = \left\langle \mathbf{d}m^u(q), 
u \right\rangle_\mathfrak{g} - \left\langle \mathbf{d}m^u(q), 
\frac{\delta h^q}{\delta m(u,q)}\right\rangle_\mathfrak{g} 
- \mathbf{d} h^{m(u,q)}(q) = - \,\mathbf{d} h^{m(u,q)}(q),
\end{align*}
since, as we just saw and invoking hyperregularity, we have
$\frac{\delta h^q}{\delta m(u,q)} = u$. Equations 
\eqref{final_strange_equ_gen_lagrangian} and \eqref{KM_q_eq_Hamel}
now follow from \eqref{final_strange_equ_gen} and 
\eqref{KM_q_eq_coord}, respectively.
\end{proof}

\subsection{The stochastic Hamilton equations}
\label{sec_stoch_ham_equ}

\paragraph{The Stratonovich stochastic Hamilton equations for
semimartingales.} 
We begin by defining Stratonovich stochastic Hamilton equations. Let 
$(P, \{\cdot , \cdot \})$ be a Poisson manifold. For any $f \in 
C^{\infty}(P)$, form the semimartingale $f(\mathscr{p}(t))$ obtained
by replacing the point $p \in P$ by a $P$-valued semimartingale
$\mathscr{p}(t)$. Consider a semimartingale
$\mathscr{Y}_t (\mathscr{p}):=\mathscr{Y}_0 +\int_0^t \phi_\alpha
(\mathscr{p}(s)) \xi_k^\alpha \circ dW^k_s +\int_0^t \psi(\mathscr{p}(s)) ds$, where 
$\phi_\alpha$, $\psi \in C^{\infty}(P)$ are deterministic smooth functions
and $\xi_k : = \xi_k^\alpha e_\alpha \in \mathfrak{g}$ are constant
elements.

In analogy with Section \ref{VP-coadmotion-sec}, the 
\textit{{\rm(}Stratonovich{\rm)}  stochastic Poisson bracket} is defined 
by
\begin{equation}
\label{stochastic_PB_general}
\{f(\mathscr{p}(t)),\circ d_t \mathscr{Y}_t\} := \{f(\mathscr{p}(t)) , \phi_\alpha(\mathscr{p}(t)) \}\xi_k^\alpha \circ dW^k_t 
+ \{f(\mathscr{p}(t)) , \psi(\mathscr{p}(t)) \}dt.
\end{equation}
where
$$
\{f (\mathscr{p}(t)) , \phi (\mathscr{p}(t)) \} 
=\{f,\phi \} (\mathscr{p}(t)).
$$

\begin{definition}
\label{def_stoch_ham_equ}
The {\rm(}Stratonovich{\rm) stochastic Hamilton equations} 
for $P$-valued semimartingales with stochastic semimartingale Hamiltonian 
$\mathscr{Y}_t (\mathscr{p}):=\mathscr{Y}_0 +\int_0^t \phi_\alpha 
(\mathscr{p}(s))\xi_k^\alpha \circ dW^k_s +\int_0^t \psi(\mathscr{p}(s)) ds$ are 
\begin{equation}
\label{stochastic_Ham_equ_general}
d_tf(\mathscr{p}(t)) = \{f(\mathscr{p}(t)),\circ d_t \mathscr{Y}_t\} := \{f, \phi_\alpha \}(\mathscr{p}(t)) \xi_k^\alpha \circ dW^k_t 
+ \{f, \psi\}(\mathscr{p}(t)) dt, \quad \text{for any} \quad f \in C^{\infty}(P),
\end{equation}
where the Poisson bracket semimartingales on the right hand side are
defined in \eqref{standard_stochastic_PB} for variations as in \eqref{stochastic_PB}.
\end{definition}

If $(p^1, \ldots, p^n)$ are coordinates on $P$, the Stratonovich stochastic
Hamilton equations thus take the form
\begin{equation}
\label{stochastic_Ham_equ_general_coordinates}
d_t \mathscr{p}^i(t) = \{\mathscr{p}^i(t),\circ d_t \mathscr{Y}_t\} = 
\{p^i, \phi_\alpha \}(\mathscr{p}(t)) \xi_k^\alpha \circ dW^k_t 
+ \{p^i, \psi\}(\mathscr{p}(t)) dt.
\end{equation}

Let $k \in C^{\infty}(P)$ be a Casimir function. Then, for
the semimartingale $k(\mathscr{p}(t))$ we have, by It\^o's formula,
\[
d_tk(\mathscr{p}(t))= 
\{k(\mathscr{p}(t)), \circ d_t\mathscr{Y}\} =
\{k, \phi_\alpha \}(\mathscr{p}(t))\xi_k^\alpha \circ dW_t^k + 
\{k, \psi\}(\mathscr{p}(t)) dt = 0,
\]
i.e., \textit{the semimartingale $k(\mathscr{p}(t))$ is conserved along the
stochastic flow of the stochastic Hamiltonian semimartingale 
$\mathscr{Y}_t (\mathscr{p})$.} Clearly, $k(\mathscr{p}(t))$ is also
conserved in the It\^o representation.

\paragraph{The Stratonovich stochastic Hamilton equations on 
$\mathfrak{g}^\ast \times Q$ and $\mathfrak{g}^\ast\times T^*Q$.} 
We continue to denote the semimartingales
$\mathscr{q}^i(t):=q^i(\mathscr{q}(t), \mathscr{p}(t))$ and
$\mathscr{p}_i(t):=p_i(\mathscr{q}(t), \mathscr{p}(t))$. 
With this definition, the information in Theorem \ref{two_PB_thm}, in 
particular, having the Poisson bracket \eqref{PB_q} on 
$\mathfrak{g}^\ast\times Q$, we form the semimartingale 
\[
d_t\mathscr{h}_t = (h^1)_\alpha (m(\mathscr{q}(t), \mathscr{p}(t)),\mathscr{q}(t)) \xi^\alpha_k
\circ dW_t^k + h^2  (m(\mathscr{q}(t), \mathscr{p}(t)),\mathscr{q}(t)) dt,
\]
 where $(h^1)_\alpha, h^2 
\in C^{\infty}(\mathfrak{g}^\ast \times Q)$ and $\xi_k = 
\xi_k ^\alpha e_\alpha\in \mathfrak{g}$. 
By \eqref{stochastic_Ham_equ_general}, the Stratonovich
stochastic Hamilton equations are 
\begin{align*}
d_tm_\alpha(\mathscr{q}(t), \mathscr{p}(t)) &= 
\{m_\alpha(\mathscr{q}(t), \mathscr{p}(t)),\, 
\circ d_t\mathscr{h}_t(m(\mathscr{q}(t), \mathscr{p}(t)),
\mathscr{q}(t))\}_{\mathfrak{g}^\ast \times Q},\\ 
d_t\mathscr{q}^i(t) &= \{\mathscr{q}^i(t),\, 
\circ d_t\mathscr{h}_t(m(\mathscr{q}(t), \mathscr{p}(t)),
\mathscr{q}(t))\}_{\mathfrak{g}^\ast \times Q}.
\end{align*} 
That is, 
\begin{align}
\begin{split}
d_t m_\alpha(\mathscr{q}(t), \mathscr{p}(t)) &=  
\{m_\alpha(\mathscr{q}(t), \mathscr{p}(t)),\, 
\circ d_t\mathscr{h}_t(m(\mathscr{q}(t), \mathscr{p}(t)),
\mathscr{q}(t)) \}_{\mathfrak{g}^\ast 
\times Q} \\
&=  \{m_\alpha(\mathscr{q}(t), \mathscr{p}(t)),\, 
m_\beta(\mathscr{q}(t), \mathscr{p}(t))\}_{\mathfrak{g}^\ast \times Q}
\circ d_t\left(\frac{\partial \mathscr{h}_t}{\partial m_\beta}\right)
((\mathscr{q}(t), \mathscr{p}(t)), \mathscr{q}(t))\\
& \qquad + \{m_\alpha(\mathscr{q}(t), \mathscr{p}(t)),\, 
\mathscr{q}^j(t) \}_{\mathfrak{g}^\ast \times Q}
\circ d_t\left( \frac{\partial \mathscr{h}_t}{\partial q^j}\right)
((\mathscr{q}(t), \mathscr{p}(t)), \mathscr{q}(t)) \\
&= -\, c_{\alpha\beta}{}^\gamma\,m_{\gamma}(\mathscr{q}(t), 
\mathscr{p}(t))\circ d_t\left( 
\frac{\partial \mathscr{h}_t}{\partial m_\beta}\right)
((\mathscr{q}(t), \mathscr{p}(t)), \mathscr{q}(t)) \\
& \qquad 
-  A^j_\alpha(\mathscr{q}(t))\circ d_t\left( 
\frac{\partial \mathscr{h}_t}{\partial q^j} \right) 
((\mathscr{q}(t), \mathscr{p}(t)), \mathscr{q}(t))\\
&=: \left[{\rm ad}^*_{d\left(\frac{\delta \mathscr{h}_t}{\delta m}\right)
((\mathscr{q}(t), \mathscr{p}(t)), \mathscr{q}(t))}\,
m(\mathscr{q}(t), \mathscr{p}(t))\right]_\alpha 
-  A^j_\alpha(\mathscr{q}(t))\circ d_t\left(\frac{\partial \mathscr{h}_t}{\partial q^j}\right)
((\mathscr{q}(t), \mathscr{p}(t)), \mathscr{q}(t))\,,\\ 
d_t \mathscr{q}^i(t) &=  \{\mathscr{q}^i(t),\, 
\circ d_t\mathscr{h}_t(m(\mathscr{q}(t), \mathscr{p}(t)),\,
\mathscr{q}(t)) \}_{\mathfrak{g}^\ast \times Q}\\
&= \{\mathscr{q}^i(t),\,m_\beta(\mathscr{q}(t),\,
 \mathscr{p}(t))\}_{\mathfrak{g}^\ast \times Q}
\circ d\left( \frac{\partial \mathscr{h}_t}{\partial m_\beta}\right)
((\mathscr{q}(t), \mathscr{p}(t)), \mathscr{q}(t)) \\
&=   A^i_\beta(\mathscr{q}(t))\circ d\left(
\frac{\partial \mathscr{h}_t}{\partial m_\beta}\right) (
(\mathscr{q}(t), \mathscr{p}(t)), \mathscr{q}(t))\,,
\end{split}
\label{LPB-Ham-form}
\end{align}
which are identical to the stochastic equations of
motion \eqref{Hamel} in Theorem \ref{m_q_thm},
once we observe that for the functional $\mathscr{h}_t$ considered there, 
the explicit $\mathscr{q}$-dependence comes only
from its bounded variation part (defined by $h^2$) and therefore
$\circ d_t \left(\frac{\partial \mathscr{h}_t}{\partial q^j}\right) = 
\left(\frac{\partial \mathscr{h}_t}{\partial q^j}\right)dt $.

Note that equations \eqref{LPB-Ham-form} comprise the stochastic version of Hamel's equations \eqref{KM_q_eq_coord}. As in the 
deterministic case (see Remark \ref{rem_coll}), note that if
$\mathscr{h}_t$ depends only on the $\mathfrak{g}^\ast$-valued
semimartingale $m(\mathscr{q}(t), \mathscr{p}(t))$, then equations 
\eqref{LPB-Ham-form} decouple into the stochastic Lie-Poisson
equations on $\mathfrak{g}^\ast_-$ and the stochastic infinitesimal
generator equation for $\frac{\delta  \mathscr{h}_t}{\delta m}
((\mathscr{q}(t), \mathscr{p}(t)), \mathscr{q}(t)) \in \mathfrak{g}$.

\medskip

Our goal is to derive \eqref{LPB-Ham-form} purely from a Hamiltonian
point of view and, similarly, Stratonovich stochastic Hamilton equations
on $\mathfrak{g}^\ast$ to $T^*Q$. In particular, this means that the
semimartingale $m(\mathscr{q}(t), \mathscr{p}(t))$ needs to be
replaced by a semimartingale $\mathscr{m}$ in order not to appeal
to the Legendre transformation of the Lagrangian $\ell$. So, the 
setup is the following general situation.  

Let $\mathscr{h}_t$ be a semimartingale of the form 
$$
d_t \mathscr{h}_t=(h^1)_{\alpha} (\mathscr{m}(t),\mathscr{q}(t))
\xi_k^\alpha \circ dW_t^k +h^2 (\mathscr{m}(t),\mathscr{q}(t)) dt\,,
$$ 
where $(h^1)_\alpha, h^2$ are (deterministic) smooth functions 
evaluated on $(\mathfrak{g}^\ast \times Q)$-valued semimartingales 
$(\mathscr{m}(t),\mathscr{q}(t))$. Similarly,
denote by $ \mathscr{\widetilde h}_t$ a semimartingale of the form
$$
d_t \mathscr{\widetilde h}_t=(\widetilde h^1)_{\alpha} (\mathscr{m}(t),
\mathscr{p}_{\mathscr q}(t))
\xi_k^\alpha \circ dW_t^k +\widetilde h^2 (\mathscr{m}(t),
\mathscr{p}_\mathscr{q}(t)) dt\,,
$$
where $(\widetilde h^1)_\alpha ,\widetilde h^2 \in 
C^\infty(\mathfrak{g}^\ast \times T^\ast Q)$ are  
evaluated on $(\mathfrak{g}^\ast \times T^*Q)$-valued semimartingales
$(\mathscr{m}(t),\mathscr{p}_\mathscr{q}(t))$.
Consider the Poisson brackets defined in \eqref{PB_q} and 
\eqref{PB_t_star_q}. According to Definition 
\ref{def_stoch_ham_equ},  the corresponding stochastic Hamilton equations 
on $\mathfrak{g}^\ast \times Q$  are \textit{defined} to be 
\[
d_t f(\mathscr{m}(t), \mathscr{q}(t)) =
\{f(\mathscr{m}(t), \mathscr{q}(t)),\,
 \circ d_t \mathscr{h}\}_{\mathfrak{g}_-^\ast\times Q},
 \]
for any $f \in C^{\infty}(\mathfrak{g}^\ast\times Q)$,  
respectively on $\mathfrak{g}^\ast \times T^*Q$, 
\[
d_t \widetilde{f}(\mathscr{m}(t), 
\mathscr{p}_\mathscr{q}(t)) =\{\widetilde{f}(\mathscr{m}(t), 
\mathscr{p}_\mathscr{q}(t)),\, \circ d_t 
\mathscr{\widetilde h}\}_{\mathfrak{g}_-^\ast\times T^*Q},
\]
for any $\widetilde{f}\in C^{\infty}(\mathfrak{g}^\ast
\times T^*Q)$. 
Notice that, by the form of the Hamiltonian functionals $\mathscr{h}$ 
and $\mathscr{\widetilde h}$, these equations are
equivalent, respectively, to
\begin{align*}
d_t f(\mathscr{m}(t), \mathscr{q}(t)) &=\{f(\mathscr{m}(t),\, 
\mathscr{q}(t)), (h^1)_\alpha(\mathscr{m}(t),
\mathscr{q}(t)) \}_{\mathfrak{g}_-^\ast\times Q}\,\xi_k^\alpha 
\circ dW_t^k +
\{f(\mathscr{m}(t), \mathscr{q}(t)), 
h^2(\mathscr{m}(t), \mathscr{q}(t))\}_{\mathfrak{g}_-^\ast\times Q} dt\\
d_t \widetilde f(\mathscr{m}(t), \mathscr{p}_\mathscr{q}(t)) &=
\{\widetilde f(\mathscr{m}(t), \mathscr{p}_\mathscr{q}(t)), 
(\widetilde{h}^1 )_\alpha(\mathscr{m}(t), 
\mathscr{p}_\mathscr{q}(t))\}_{\mathfrak{g}_-^\ast\times T^*Q}\,\xi_k^\alpha 
\circ dW_t^k \\
& \qquad +\{\widetilde{f}(\mathscr{m}(t), 
\mathscr{p}_\mathscr{q}(t)),\, 
\widetilde h^2(\mathscr{m}(t), 
\mathscr{p}_\mathscr{q}(t))\}_{\mathfrak{g}_-^\ast\times T^*Q} dt.
\end{align*}

We now define the right hand sides of these equations involving the
Poisson bracket.

For $f\in C^{\infty}(\mathfrak{g}^\ast\times Q)$, the Poisson
bracket \eqref{PB_q} of the two semimartingales 
$f(\mathscr{m}(t), \mathscr{q}(t))$ and $\mathscr{h}_t$ then reads 
\begin{align}
\label{q_stoch_bracket}
&\{f(\mathscr{m}(t), \mathscr{q}(t)),
\circ d_t\mathscr{h}_t\}_{\mathfrak{g}^\ast_-\times Q}:=
\begin{bmatrix}
\frac{\partial f}{\partial m_\alpha}(\mathscr{m}(t), \mathscr{q}(t)) 
\vspace{2mm} \\ 
\frac{\partial f}{\partial q^i}(\mathscr{m}(t), \mathscr{q}(t)) 
\end{bmatrix}^\mathsf{T}
\begin{bmatrix}
-\,c_{\alpha\beta}{}^\gamma\,\mathscr{m}_{\gamma}(t) 
& -\,A^j_\alpha (\mathscr{q}(t))
\\
A^i_\beta (\mathscr{q}(t)) & 0 
\end{bmatrix}
\begin{bmatrix}
\circ d_t\left(\frac{\partial \mathscr{h}_t}{\partial m_\beta} 
\right) \vspace{2mm}\\ 
\circ d_t\left(\frac{\partial \mathscr{h}_t}{\partial q^j} \right)
\end{bmatrix}  \nonumber \\
&\qquad = 
\left\{f^q(\mathscr{m}(t), \mathscr{q}(t)),\,
\circ d_t\mathscr{h}^q_t\right\}_-  
+ \left\langle \mathbf{d}f^m(\mathscr{m}(t), \mathscr{q}(t)), 
\left(\circ d_t\left(\frac{\delta  \mathscr{h}_t^q}{\delta m}
\right)\right)_Q(\mathscr{m}(t),\mathscr{q}(t)) \right\rangle_Q \nonumber \\
& \qquad \qquad 
- \left\langle \circ d_t\mathbf{d}\mathscr{h}_t^m(\mathscr{m}(t),
\mathscr{q}(t)), 
\left(\frac{\delta f^q}{\delta m}\right)_Q(\mathscr{m}(t),\mathscr{q}(t)) \right\rangle_Q .
\end{align}
In this formula, $\partial f / \partial m_\alpha$ and 
$\partial f / \partial q^i$ are evaluated on the semimartingales 
$\mathscr{m}(t)$ and $\mathscr{q}(t)$ and, according to \eqref{semim_der1} and \eqref{semim_der2},
\begin{align*}
d_t \left(\frac{\partial \mathscr{h}_t}{\partial m_\beta}
\right)&:= \frac{\partial(h^1)_\alpha}{\partial m_\beta} 
(\mathscr{m}(t), \mathscr{q}(t))
\xi_k^\alpha \circ dW_t^k +\frac{\partial h^2}{\partial m_\beta} 
(\mathscr{m}(t),\mathscr{q}(t)) dt,  \\  
 d_t \left(\frac{\partial \mathscr{h}_t}{\partial q^j} \right)&:=
\frac{\partial(h^1)_\alpha}{\partial q^j} (\mathscr{m}(t),\mathscr{q}(t))
\xi_k^\alpha \circ dW_t^k +\frac{\partial h^2}{\partial q^j} 
(\mathscr{m}(t),\mathscr{q}(t)) dt,
\end{align*}
$\mathbf{d}(h^1)_\alpha^m(\mathscr{m}(t), \mathscr{q}(t))$, 
$\mathbf{d}(h^2)^m(\mathscr{m}(t), \mathscr{q}(t))$,
$\mathbf{d}f^m(\mathscr{m}(t), \mathscr{q}(t))$,
$\frac{\delta(h^1)_\alpha^q}{\delta m}(\mathscr{m}(t), \mathscr{q}(t))$,
$\frac{\delta (h^2)^q}{\delta m}(\mathscr{m}(t), \mathscr{q}(t))$, 
$\frac{\delta f^q}{\delta m}(\mathscr{m}(t), \mathscr{q}(t))$, and 
$\left(\frac{\delta(h^1)_\alpha^q}{\delta m}\right)_Q(\mathscr{m}(t),
\mathscr{q}(t))$,
$\left(\frac{\delta(h^2)^q}{\delta m}\right)_Q(\mathscr{m}(t),
\mathscr{q}(t))$,
$\left(\frac{\delta f^q}{\delta m}\right)_Q(\mathscr{m}(t), \mathscr{q}(t))$ are the covectors $\mathbf{d}(h^1)_\alpha^m(q)$,
$\mathbf{d}(h^2)^m(q)$, $\mathbf{d}f^m(q)
\in T_q^*Q$, the elements $\delta(h^1)_\alpha^q/\delta m$, 
$\delta(h^2)^q/\delta m$, $\delta f^q/\delta m\in \mathfrak{g}$, 
and the tangent vectors $\left(\delta(h^1)_\alpha^q/\delta m\right)_Q(q)$,
$\left(\delta(h^2)^q/\delta m\right)_Q(q)$, 
$\left(\delta f^q/\delta m\right)_Q(q) \in T_qQ$ 
with the variables $(m,q)$ replaced by the semimartingales 
$(\mathscr{m}(t),\mathscr{q}(t))$,
\begin{align*}
 d_t\mathbf{d}\mathscr{h}_t^m(\mathscr{m}(t), \mathscr{q}(t))&:=
\mathbf{d}(h^1)_\alpha^m(\mathscr{m}(t),\mathscr{q}(t))
\xi_k^\alpha \circ dW_t^k + \mathbf{d}(h^2)^m(\mathscr{m}(t),
\mathscr{q}(t)) dt 
,\\
\left(\circ d\left(\frac{\delta \mathscr{h}_t^q}{\delta m}
\right)\right)_Q(\mathscr{m}(t),\mathscr{q}(t))&:= 
\left(\frac{\delta(h^1)_\alpha^q}{\delta m}
\right)_Q(\mathscr{m}(t),\mathscr{q}(t))\xi_k^\alpha \circ dW_t^k + 
\left(\frac{\delta(h^2)^q}{\delta m}\right)_Q(\mathscr{m}(t),\mathscr{q}(t)) dt, \\
\left\{f^q(\mathscr{m}(t),\mathscr{q}(t)),\,d_t\mathscr{h}_t^q\right\}_-&:=
\left\{f^q(\mathscr{m}(t),\mathscr{q}(t)),\, (h^1)_{\alpha}^q
(\mathscr{m}(t),\mathscr{q}(t))\right\}_- 
\xi_k^\alpha \circ dW_t^k \\
& \qquad +
\left\{f^q(\mathscr{m}(t),\mathscr{q}(t)),\, 
(h^2)^q(\mathscr{m}(t),\mathscr{q}(t))\right\}_- dt.
\end{align*}

Similarly, for  $\widetilde{f}\in C^{\infty}(\mathfrak{g}^\ast\times T^*Q)$, 
the Poisson bracket \eqref{PB_t_star_q} computed for these semimartingales 
is given by
\begin{align}
\label{t_star_q_stoch_bracket}
&\left\{\widetilde{f}\left(\mathscr{m}(t), \mathscr{p}_{\mathscr{q}}(t)
\right), \circ d \widetilde{\mathscr{h}}_t
\right\}_{\mathfrak{g}^\ast_- \times T^*Q}
\nonumber \\
&\qquad :=
\begin{bmatrix}
\frac{\partial \widetilde{f}}{\partial m_\alpha}\left(\mathscr{m}(t), 
\mathscr{p}_{\mathscr{q}}(t) \right) \vspace{1mm} \\ 
\frac{\partial \widetilde{f}}{\partial q^i} \left(\mathscr{m}(t), 
\mathscr{p}_{\mathscr{q}}(t)\right) \vspace{1mm}\\ 
\frac{\partial \widetilde{f}}{\partial p_i} \left(\mathscr{m}(t), 
\mathscr{p}_{\mathscr{q}}(t)\right)
\end{bmatrix}^\mathsf{T}
\begin{bmatrix}
-c_{\alpha \beta}{}^\gamma\mathscr{m}_\gamma(t) & -\,
A^j_\alpha( \mathscr{q}(t)) & 
\mathscr{p}_k(t)\frac{\partial A^k_\alpha}{\partial q^j}(\mathscr{q}(t))
\\
A^i_\beta( \mathscr{q}(t)) & 0 & \delta^i_j
\\
-\,\mathscr{p}_k(t)\frac{\partial A^k_\beta}{\partial q^i}
( \mathscr{q}(t))  & -\,\delta^j_i & 0
\end{bmatrix}
\begin{bmatrix}
\vspace{1mm}
\circ d_t\left(\frac{\partial \widetilde{\mathscr{h}}_t}{\partial
m_\beta} \right) \\ 
\circ d_t\left(\frac{\partial \widetilde{\mathscr{h}}_t}{\partial q^j} 
\right) \vspace{1mm}\\ 
\circ d_t\left(\frac{\partial \widetilde{\mathscr{h}}_t}{\partial p_j} \right)
\end{bmatrix} \nonumber  \\
&\qquad\; = \left\{\widetilde{f}^{p_q}\left(\mathscr{m}(t), 
\mathscr{p}_{\mathscr{q}}(t)\right), 
\circ d_t\widetilde{\mathscr{h}}_t^{p_q} 
\right\}_-  
+ \left\langle \mathbf{d} f^m\left(\mathscr{m}(t), 
\mathscr{p}_{\mathscr{q}}(t)\right), 
\left(\circ d_t\left(\frac{\delta\widetilde{\mathscr{h}}_t^{p_q}}{\delta m}
\right)\right)_{T^*Q}
\left(\mathscr{m}(t), \mathscr{p}_{\mathscr{q}}(t)
\right) \right\rangle_Q \nonumber \\
& \qquad \qquad 
- \left\langle \circ d_t\mathbf{d} \widetilde{\mathscr{h}}_t^m
\left(\mathscr{m}(t), \mathscr{p}_{\mathscr{q}}(t)\right), 
\left(\frac{\delta \widetilde{f}^{p_q}}{
\delta m}\right)_{T^*Q}\left(\mathscr{m}(t), \mathscr{p}_{\mathscr{q}}(t)
\right)
\right\rangle_Q
+\left\{\widetilde{f}^m \left(\mathscr{m}(t), \mathscr{p}_{\mathscr{q}}(t)
\right), 
\circ d_t\widetilde{\mathscr{h}}_t^m \right\},
\end{align}
with the same notational conventions as for the bracket 
\eqref{q_stoch_bracket} and where the last Poisson bracket of
semimartingales is defined in \eqref{stochastic_PB_general}.

Since $d_tf(\mathscr{m}(t), \mathscr{q}(t)) =
\left\langle \frac{\delta f^q}{\delta m}(\mathscr{m}(t), \mathscr{q}(t)), 
\circ d_t\mathscr{m}(t) \right\rangle_\mathfrak{g} +
\left\langle \mathbf{d}f^m (\mathscr{m}(t), \mathscr{q}(t)),
\circ d_t\mathscr{q}(t) \right\rangle_Q$, the stochastic Hamilton equations 
(i.e., the stochastic versions of
equations \eqref{KM_q_eq_coord} and \eqref{KM_star_q_eq_coord})  are, respectively, the stochastic Hamel equations
\begin{align}
\label{ham_stoch_hamel}
&d_t  \mathscr{m}(t) = \operatorname{ad}_{\circ d_t\left(
\frac{\delta \mathscr{h}_t^q}{\delta m}(\mathscr{m}(t), 
\mathscr{q}(t))\right)}^* \mathscr{m} (t) - 
\mathbf{J}_{T^*Q}(\circ d_t\mathbf{d}\mathscr{h}_t^m(\mathscr{m}(t),
\mathscr{q}(t))),  \quad  
d_t \mathscr{q}(t) = \left( \circ d_t\left(\frac{\delta \mathscr{h}_t^q}{\delta m}\right) \right)_Q(\mathscr{m}(t),\mathscr{q}(t)) \quad 
\Longleftrightarrow  \nonumber \\
&d_t \mathscr{m}_\alpha(t) = -c_{\alpha \beta}{}^\gamma 
\mathscr{m}_\gamma (t)
\circ d_t\left(\frac{\partial \mathscr{h}_t}{\partial m_\beta}\right) - 
A^j_\alpha(\mathscr{q}(t)) \circ d_t\left( \frac{\partial \mathscr{h}_t}{\partial q^j}\right)\,, \qquad\;\; 
d_t \mathscr{q}^i (t)= A^i_\beta(\mathscr{q}(t)) \circ d_t\left(
\frac{\partial \mathscr{h}_t}{\partial m_\beta}\right)
\end{align}
and
\begin{equation}
\begin{aligned}
d_t \mathscr{m}_\alpha(t) &= -c_{\alpha \beta}{}^\gamma 
\mathscr{m}_\gamma(t) 
\circ d_t\left(\frac{\partial \widetilde{\mathscr{h}}_t}{\partial m_\beta}\right) -
A^j_\alpha(\mathscr{q}(t)) \circ d_t\left(
\frac{\partial \widetilde{\mathscr{h}}_t}{ \partial q^j} \right) + 
\mathscr{p}_k(t)
\frac{\partial A_\alpha ^k}{ \partial q^j}(\mathscr{q}(t)) 
\circ d_t\left(\frac{ \partial \widetilde{\mathscr{h}}_t}{ \partial p_j}
\right), \\ 
d_t \mathscr{q}^i(t) &= A^i_\beta(\mathscr{q}(t)) \circ d_t \left(\frac{\partial \widetilde{\mathscr{h}}_t}{\partial m_\beta}\right) + 
d_t \left(\frac{\partial \widetilde{\mathscr{h}}_t}{ \partial p_i}\right),\quad 
d_t \mathscr{p}_i(t) = -\mathscr{p}_k(t) 
\frac{\partial A_\beta^k}{ \partial q^i}(\mathscr{q}(t))
\circ d_t\left(\frac{\partial \widetilde{\mathscr{h}}_t}{\partial m_\beta}\right) -
d_t\left(\frac{\partial \widetilde{\mathscr{h}}_t}{ \partial q^i}\right).
\end{aligned}
\end{equation}
The last equations can be written intrinsically as
\begin{equation}
\begin{aligned}
d_t \mathscr{m}(t)&= \operatorname{ad}_{\circ d_t\left(
\frac{\delta \widetilde{\mathscr{h}}_t^q}{\delta m}\right)(\mathscr{m}(t), 
\mathscr{q}(t))}^* \mathscr{m} (t)
- \mathbf{J}_{T^* (T^* Q) }\left(\circ d_t (\mathbf{d}
\widetilde{\mathscr{h}}_t^m (\mathscr{m}(t),\mathscr{p}_\mathscr{q}(t))) \right), \\
d_t \mathscr{p}_{\mathscr{q}}(t) &=  \left(\circ d_t \left(
\frac{\delta \widetilde{\mathscr{h}}_t^{p_q}}{\delta m} \right) 
\right)_{T^* Q}(\mathscr{m}(t),\mathscr{p}_{\mathscr{q}}(t))+
X^{T^* Q}_{\circ d_t \widetilde{\mathscr{h}}_t^m}(\mathscr{m}(t), 
\mathscr{p}_{\mathscr{q}}(t))\,.
\end{aligned}
\end{equation}

By repeating the arguments
in the proof of Corollary 11,   we can derive the following stochastic Hamilton
 equations (with $P = T^*Q${\rm)} on 
$\mathfrak{g}^\ast \times T^*Q$ for these type of Hamiltonian functionals:
\begin{equation}
\begin{aligned}
d_t \mathscr{m}(t) &= \operatorname{ad}_{\circ d_t\left(
\frac{\delta \mathscr{h}_t^q}{\delta m}\right)(\mathscr{m}(t), 
\mathscr{q}(t))}^* \mathscr{m} (t)
- \mathbf{J}_{T^*Q}(\circ d_t\mathbf{d}\mathscr{h}_t^m(\mathscr{m}(t),
\mathscr{q}(t))), \\ 
d_t \mathscr{p}_\mathscr{q}(t) &=  
\left(d_t\left(\frac{\delta \mathscr{h}_t^q}{\delta m} \right)\right)_{T^*Q}
(\mathscr{m}(t), \mathscr{p}_\mathscr{q}(t))
+T_{d_t\mathbf{d} \mathscr{h}_t^m(\mathscr{m}(t), 
\mathscr{p}_\mathscr{q}(t))} (\mathscr{m}(t),\mathscr{p}_\mathscr{q}(t)),
\end{aligned}
\end{equation}
where 
$$
T_{d_t\mathbf{d} \mathscr{h}_t^m(\mathscr{m}(t), 
\mathscr{p}_\mathscr{q}(t))}(\mathscr{p}_\mathscr{q})= 
\left.\frac{d}{d\epsilon}\right|_{\epsilon=0}
(d_t \mathscr{p}_\mathscr{q} -\epsilon  ~ d_t\mathbf{d} \mathscr{h}_t^m(\mathscr{m}(t), 
\mathscr{p}_\mathscr{q}(t)))
$$
in which the limit is taken in $L^2 (\Omega )$.

In addition, we have the non-homogeneous stochastic Lie-Poisson equations
\begin{equation}
d_t\mathbf{J}_{T^*Q}(\mathscr{p}_{q}(t)) = 
\operatorname{ad}_{\circ d_t\left(
\frac{\delta \mathscr{h}_t^{q(t)}}{\delta m(t)}\right) (\mathscr{m}(t), 
\mathscr{q}(t))}^\ast
\mathbf{J}_{T^*Q}(\mathscr{p}_\mathscr{q}(t))
- \mathbf{J}_{T^*Q}\left(\circ d_t\mathbf{d}\mathscr{h}_t^{m(t)} (
\mathscr{m}(t),\mathscr{q}(t) \right)
\end{equation}
for $\mathbf{J}_{T^*Q}(\mathscr{p}_\mathscr{q}(t))$,
where $(\mathscr{m}(t), \mathscr{q}(t))$ 
is the solution of the stochastic  Hamel equations \eqref{ham_stoch_hamel}.

\begin{remark}\rm[The analog of the stochastic rigid body] 
In particular,
the stochastic free rigid body dynamics on $\mathfrak{g}^\ast$ is obtained
from \eqref{LPB-Ham-form} by assuming that the 
Hamiltonian $h : \mathfrak{g}^\ast \times Q \rightarrow \mathbb{R}$
is of the form $h(m,q): = \frac{1}{2} m_\alpha K^{\alpha\beta}m_\beta$, 
where $K^{\alpha\beta}$ are the components of the inner product 
on $\mathfrak{g}^\ast$ induced by an inner product on $\mathfrak{g}$.
This Hamiltonian $h$ is computed on semimartingales of the form
$d_t\mathscr{m}_\alpha =\mathscr{r}_\alpha dt +(\Phi_k)_\alpha \circ dW_t^k$ 
with $(\Phi_k )_\alpha K^{\alpha \beta} =\xi_k^\beta$ (notice that $K:=
\left[K^{\alpha\beta} \right]$ 
is an invertible matrix). Define $u^\beta:=
r_\alpha K^{\alpha \beta}$. In particular,
$\frac{\partial h}{\partial q^i}=0$ and
$d\Big(\frac{\partial h}{\partial m_\beta}\Big)=
r_\alpha K^{\alpha \beta} dt +
(\Phi_k )_\alpha K^{\alpha \beta}  \circ dW_t^k =
u^\beta dt +\xi_k^\beta \circ dW_t^k$. Thus,
from \eqref{LPB-Ham-form}, the stochastic free rigid body equations emerge as 
\begin{align}
\qquad \qquad\qquad\qquad\qquad
d\mathscr{m}_\alpha =   \{\mathscr{m}_\alpha\,,\, \mathscr{m}_\beta \}
(u^\beta\,dt + \xi_k^\beta \circ dW^k_t)
= -\,c_{\alpha\beta}{}^\gamma\,\mathscr{m}_{\gamma} \,dx_t^\beta\,. 
\qquad\qquad\qquad\qquad  \lozenge
\label{LPB-Ham-RBanalog}
\end{align}
\end{remark}

\subsection{The Kolmogorov equations}
\label{sec:Kolmogorov}

We start from the equations \eqref{stochastic_Ham_equ_general_coordinates}
for the Poisson manifold valued stochastic process $\mathscr{p}(t)$ written in coordinates, namely,
$$
d_t \mathscr{p}^i(t) =
\{p^i, \phi_\alpha \}(\mathscr{p}(t)) \xi_k^\alpha \circ dW^k_t 
+ \{p^i, \psi\}(\mathscr{p}(t)) dt.
$$

\begin{theorem} The generator of the process $\mathscr{p}(t)$  is the 
operator
\begin{align}
\label{generator}
L f:=\{f,\psi \} + \frac{1}{2} \sum_k \{\phi_\alpha \xi_k^\alpha ,\{\phi_\beta \xi_k^\beta, f\} \}.
\end{align}
\end{theorem}

\begin{proof}
We first compute the difference between the It\^o and the Stratonovich 
differential in the process $\mathscr{p}(t)$. 
Since $\{p^i, \phi_\alpha \}(\mathscr{p}(t))= \Pi^{ij}(\mathscr{p}(t)) 
\frac{\partial \phi_\alpha}{\partial p^j}(\mathscr{p}(t))$,
this difference is equal to
\begin{align*}
\frac{1}{2}  \left( \frac{\partial \phi_\alpha}{\partial p^j}
\frac{\partial \Pi^{ij}}{\partial p^l} +
\Pi^{ij} \frac{\partial^2 \phi_\alpha}{\partial p^l \partial p^j}\right) 
(\mathscr{p}(t))
  \Pi^{lm}(\mathscr{p}(t))\frac{\partial \phi_\beta}{\partial p^m}(\mathscr{p}(t))\left( \sum_k  \xi_k^\alpha \xi_k^\beta \right) dt =: 
  B^i (\mathscr{p}(t))dt.
\end{align*}
The generator then reads
$$
Lf(p)=\frac{1}{2} \left( \Pi^{im}\Pi^{jn}
\frac{\partial \phi_\alpha}{\partial p^m}
\frac{\partial \phi_\beta}{\partial p^n}\right) (p)
\left( \sum_k \xi_k^\alpha \xi_k^\beta \right)
\frac{\partial^2 f}{\partial p^i \partial p^j} +
\left( \Pi^{ij}\frac{\partial \psi}{\partial p^j}+B^i \right) (p)\frac{\partial f}{\partial p^i}
$$
which is precisely  the expression \eqref{generator}.
\end{proof}

Defining $\rho (t, p):=\mathbb{E}_p \left( f( \mathscr{p}(t)) \right)$ where 
$\mathscr{p}(0)=p$, the function $\rho$ satisfies \textit{Kolmogorov's 
backward equation}, namely,
\begin{equation}
\label{Kolmogorov_equ}
\frac{\partial \rho }{\partial t}= L \rho ,\qquad \rho (0,p)=f(p).
\end{equation}

If the generator $L$ is a hypoelliptic operator then there exists a probability density function $\tilde \rho(t, p, p')$, defined by  
$\mathbb{E}_p \left( f( \mathscr{p}(t)) \right)=\int_P f(p^{\prime} )
\tilde \rho (t,p, p^{\prime})dp^{\prime}$; here we assume that the Poisson 
manifold $P$ has a volume form $dp$ relative to which this integration 
is carried out. This function satisfies
the forward Kolmogorov (or Fokker-Planck) equation:
\begin{equation}
\label{forward_Kolmogorov}
\frac{\partial \tilde \rho }{\partial t}(t,p,p^{\prime})=
L^\ast_{p^{\prime}}\tilde \rho (t,p,p^{\prime})
\end{equation}
with $\tilde \rho (0,p, p^{\prime})$ equal to the Dirac measure $\delta (p^{\prime}-p)$ and
where $L^\ast$ denotes the adjoint of $L$.

Next, we give a sufficient condition, in terms of the measure on
$P$ used to define the probability density function $\rho(t,p,p')$,
ensuring a nice formula for the formal adjoint of
the operator $L$ defined in \eqref{generator}.
We take a measure on $P$ which is induced by a volume form $\Lambda \in 
\Omega^{\dim P}(P)$. We say that a volume form $\Lambda$ on $P$ is 
\textit{Hamiltonian}, if $0=\boldsymbol{\pounds}_{X_g} \Lambda = 
\mathbf{d} \mathbf{i}_{X_g} \Lambda+ \mathbf{i}_{X_g} \mathbf{d}\Lambda 
=\mathbf{d} \mathbf{i}_{X_g} \Lambda$, for all $g \in C^{\infty}(P)$.
Therefore
\begin{align*}
\operatorname{div}(fX_g)\Lambda &= \boldsymbol{\pounds}_{fX_g}\Lambda
= \mathbf{i}_{fX_g} \mathbf{d}\Lambda + \mathbf{d}\mathbf{i}_{fX_g}\Lambda
= \mathbf{d}(f\mathbf{i}_{X_g} \Lambda)\\
& = \mathbf{d}f \wedge \mathbf{i}_{X_g} \Lambda + 
f \mathbf{d}\mathbf{i}_{X_g} \Lambda = 
-\mathbf{i}_{X_g} (\mathbf{d}f \wedge \Lambda) + 
\left(\mathbf{i}_{X_g} \mathbf{d}f\right) \Lambda\\
&= X_g[f] \Lambda = \{f, g\} \Lambda\,.
\end{align*}
This shows that $\operatorname{div}(fX_g) = \{f, g\}$ for any 
$f, g \in C^{\infty}(P)$.

Hence, by the Stokes Theorem,
\begin{equation}
\label{bracket_identity}
\int_P \{f,g\} \Lambda = \int_P\operatorname{div}(fX_g)\Lambda
= \int_P \mathbf{d} \mathbf{i}_{fX_g} \Lambda = 
\int_{\partial P} \mathbf{i}_{fX_g} \Lambda = 
\int_{\partial P} f\mathbf{i}_{X_g} \Lambda = 
- \int_{\partial P} g\mathbf{i}_{X_f} \Lambda;
\end{equation}
in which the last equality follows by skew-symmetry of the Poisson bracket.

Now let $f, g, h \in C^{\infty}(P)$ and integrate the identity
$\{hf, g\} = h\{f, g\} + f\{h, g\}$ to get
\[
\int_P h\{f, g\}\Lambda  + \int_P f\{h, g\}\Lambda = 
\int_P \{hf, g\}\Lambda.
\]
By \eqref{bracket_identity}, the term on the right hand side vanishes if $\partial P = \varnothing$
or if at least one of $f$ or $g$ vanish on $\partial P$. In these cases, we have 
\begin{equation}
\label{nice_bracket_identity}
\int_P\{f,g\}h \Lambda = \int_Pf\{g, h\} \Lambda.
\end{equation}

If $(P, \omega)$ is a $2n$-dimensional symplectic manifold, the Liouville
volume $\Lambda: = \frac{(-1)^{n(n-1)/2}}{n!} \omega \wedge\ldots 
\wedge \omega$ ($n$ times) is Hamiltonian. Indeed, since 
$\boldsymbol{\pounds}_{X_g} \omega=0$ for any $g \in C^{\infty}(P)$, it 
immediately follows that $\boldsymbol{\pounds}_{X_g} \Lambda=0$.

\begin{corollary}
Let $(P, \{ \cdot , \cdot \})$ be a boundaryless Poisson manifold and 
$\Lambda$ a Hamiltonian volume form on $P$. Relative to the $L^2$-inner 
product on $P$ defined by $\Lambda$, the formal adjoint of the linear
operator $L$ defined in \eqref{generator} may be expressed as
\begin{equation}
\label{generator_adjoint}
L^* f = -\{f,\psi\} +
\frac{1}{2} \sum_k \{\phi_\beta \xi_k^\beta, 
\{\phi_\alpha \xi_k^\alpha, f\}\}.
\end{equation}
\end{corollary}

This corollary follows directly from \eqref{generator} and \eqref{nice_bracket_identity}.

Consider the Poisson manifold $\mathfrak{g}^\ast \times Q$ 
and the stochastic Hamiltonian \eqref{LegXform1}. Define the
semimartingale $u(\mathscr{m}(\mathscr{q}(t), 
\mathscr{p}(t)),\mathscr{q}(t))$, where $u \in 
C^{\infty}(\mathfrak{g}^\ast\times Q)$. In this case, 
Kolmogorov's backward equation for
$\rho (t,m,q):=\mathbb{E}_{(m,q)} \left( f( \mathscr{m}(\mathscr{q}(t), \mathscr{p}(t))), \mathscr{q}(t)) \right)$ takes the form
\begin{equation}
\label{Kol_back_g_star_q}
\frac{\partial \rho}{\partial t}= \{\rho,m_\alpha u^\alpha -\ell (u,q)\}
+ \frac{1}{2} \sum_k \{m_\alpha \xi_k^\alpha ,\{m_\beta \xi_k^\beta, \rho \}\}
\end{equation}
with $\rho (0,m,q)=f(m,q)$.

Now choose a Hamiltonian volume form $\Lambda$ on the Poisson manifold
$\mathfrak{g}^\ast \times T^*Q$. Using the measure defined by $\Lambda$,
and computing the formal adjoint of $L$ (the right hand side of 
\eqref{Kol_back_g_star_q}) given by \eqref{generator_adjoint}, we get
Kolmogorov's forward, or Fokker-Planck, equation
\begin{equation}
\label{Kol_forward_g_star_q}
\frac{\partial \tilde \rho}{\partial t}= -\{\tilde \rho ,m_\alpha u^\alpha -\ell (u,q)\}
+ \frac{1}{2} \sum_k \{m_\alpha \xi_k^\alpha ,\{m_\beta \xi_k^\beta, \tilde \rho \}\}
\end{equation}
with $\tilde \rho(0,(m,q), (m^{\prime}, q^{\prime}))=\delta ((m^{\prime}, q^{\prime})-(m,q))$.

\begin{remark}{\rm Assume we work on $\mathfrak{g}^\ast$, where 
$\mathfrak{g}$ is a compact Lie algebra, for simplicity. Then there is
an invariant inner product on $\mathfrak{g}$ and, using it, we define 
an invariant inner product on $\mathfrak{g}^\ast$ whose norm is
denoted by $\| \cdot \|$. In this case, $m \mapsto \|m\|^2$ is a Casimir
function. As Casimirs are conserved along the
stochastic flows of the stochastic Hamiltonian semimartingales, we have $d_t \| m \|^2 (\mathscr{q}(t),\mathscr{p}(t))=0$. As a consequence, there
exists (cf. \cite{Mi1973}) an invariant probability measure $\mu$ on  $\mathfrak{g}^\ast \times Q$  for the motion. Namely, the measure satisfies
\[
\int \mathbb{P} _{(m,q)} ( \mathscr{m}(\mathscr{q}(t), \mathscr{p}(t)), \mathscr{q}(t)) \in B)d\mu (m,q)=\mu (B)\,,
\]
for all Borel sets $B\subset \mathfrak{g}^\ast \times Q$.
This measure disintegrates along the  level sets of the Casimir $\|m\|^2$.
\hfill $\lozenge$ }
\end{remark}

\subsection*{Acknowledgements}
We are enormously grateful to our colleagues for their helpful 
encouraging remarks and interesting  enjoyable discussions:  
A. Arnaudon, M. Arnaudon, S. Albeverio, J.-M. Bismut, N. Bou-Rabee, A. L. Castro, 
M. D. Chekroun, G. Chirikjian, D. O. Crisan,  J. Eldering, M. Engel, 
F. Gay-Balmaz, A. Grandchamp, P. Lynch, J.-P. Ortega, G. Pavliotis, 
V. Putkaradze, and C. Tronci. We also acknowledge the Bernoulli Center  
at EPFL where parts of this work were initiated. ABC was partially 
supported by Portuguese FCT grant SFRH/BSAB/105789/2014, DDH by 
ERC Advanced Grant 267382 FCCA as well as EPSRC Grant EP/N023781/1, and TSR by NCCR SwissMAP grant of
the Swiss NSF.

\end{document}